\documentclass[a4paper,11pt,fleqn]{article}
\pdfoutput=1

\usepackage{jheppub}
\usepackage[T1]{fontenc}

\usepackage{enumitem}

\usepackage{amsmath}
\usepackage{amsthm}
\usepackage{bm}
\usepackage{mathtools}

\usepackage{caption}
\usepackage{subcaption}

\usepackage{tikz}
\usepackage{pgfplots}
\usetikzlibrary{positioning,arrows,patterns}
\usetikzlibrary{decorations.markings}
\usetikzlibrary{calc}

\newcommand{\C}{$\mathrm{C}$}
\newcommand{\Cstar}{$\mathrm{C}^\star$}
\newcommand{\CQED}{QED$_{\mathrm{C}}$}
\newcommand{\CQCDQED}{QCD$+$QED$_{\mathrm{C}}$}
\newcommand{\LQED}{QED$_{\mathrm{L}}$}
\newcommand{\U}{\mathrm{U}}
\newcommand{\SU}{\mathrm{SU}}
\renewcommand{\O}{\mathrm{O}}

\newcommand{\Z}{\mathbb{Z}}

\newcommand{\tr}{\mathrm{tr}\,}

\renewcommand{\vec}[1]{\mathbf{#1}}
\renewcommand{\overline}[1]{\bar{#1}}
\renewcommand{\Re}[0]{\operatorname{Re}}
\renewcommand{\Im}[0]{\operatorname{Im}}

\newcommand{\lrD}[1]{%
   \makebox[0pt][l]{%
      \raisebox{1.8ex}{%
         \scriptsize\hspace{.1ex}\text{%
            \scalebox{1.2}[.73]{$\leftrightarrow$}
         }
      }
   }D{}
}
\newcommand{\lrp}[1]{%
   \makebox[0pt][l]{%
      \raisebox{1.8ex}{%
         \scriptsize\hspace{-.3ex}\text{%
            \scalebox{1.1}[.73]{$\leftrightarrow$}
         }
      }
   }\partial{}
}
\newcommand{\lp}[1]{%
   \makebox[0pt][l]{%
      \raisebox{1.8ex}{%
         \scriptsize\hspace{-.3ex}\text{%
            \scalebox{1.1}[.73]{$\leftarrow$}
         }
      }
   }\partial{}
}
\newcommand{\rp}[1]{%
   \makebox[0pt][l]{%
      \raisebox{1.8ex}{%
         \scriptsize\hspace{-.3ex}\text{%
            \scalebox{1.1}[.73]{$\rightarrow$}
         }
      }
   }\partial{}
}

\theoremstyle{plain}
\newtheorem{theorem}{Theorem}[section]

\newtheorem{lemma}[theorem]{Lemma}
\newtheorem{proposition}[theorem]{Proposition}
\newtheorem{observation}[theorem]{Observation}

\title{Charged hadrons in local finite-volume QED+QCD with C$^\star$ boundary conditions}

\author[a]{B. Lucini,}
\author[b,c]{A. Patella,}
\author[b]{A. Ramos,}
\author[d,b]{N. Tantalo}

\affiliation[a]{Physics Department, College of Science,Swansea University,\\
Singleton Park, Swansea SA2 8PP, UK}
\affiliation[b]{PH-TH, CERN, CH-1211 Geneva 23, Switzerland}
\affiliation[c]{School of Computing and Mathematics \& Centre for Mathematical Science,\\
Plymouth University, Plymouth PL4 8AA, UK}
\affiliation[d]{Dipartimento di Fisica and INFN, Universit\`a di Roma ``Tor Vergata'',\\ 
Via della Ricerca Scientifica 1, I-00133 Roma, Italy\\}

\emailAdd{b.lucini@swansea.ac.uk}
\emailAdd{agostino.patella@cern.ch}
\emailAdd{alberto.ramos@cern.ch}
\emailAdd{nazario.tantalo@roma2.infn.it}

\preprint{CERN-PH-TH-2015-166}

\abstract{
In order to calculate QED corrections to hadronic physical quantities by means of lattice simulations, a coherent description of electrically-charged states in finite volume is needed. In the usual periodic setup, Gauss's law and large gauge transformations forbid the propagation of electrically-charged states. A possible solution to this problem, which does not violate the axioms of local quantum field theory, has been proposed by Wiese and Polley, and is based on the use of \Cstar{} boundary conditions. We present a thorough analysis of the properties and symmetries of QED in isolation and QED coupled to QCD, with \Cstar{} boundary conditions. In particular we learn that a certain class of electrically-charged states can be constructed in this setup in a fully consistent fashion, without relying on gauge fixing. We argue that this class of states covers most of the interesting phenomenological applications in the framework of numerical simulations. We also calculate finite-volume corrections to the mass of stable charged particles and show that these are much smaller than in non-local formulations of QED.
}

\begin{document}
\maketitle
\flushbottom

\parindent 0pt
\parskip 15pt

\abovedisplayskip 11pt
\belowdisplayskip 11pt

\section{Introduction}
\label{sec:intro}
Electromagnetic interactions contribute at the order of a few
percentage points to masses, decay rates and scattering cross-sections of
hadrons. Nevertheless these small contributions cannot be ignored if
one is interested in quantifying isospin breaking effects like the
charged-neutral mass splittings of baryons and 
mesons, or when one aims at percent accuracy in the calculation of
hadronic matrix elements. In these cases first-principle theoretical
predictions can be obtained only by means of lattice techniques, which
require a consistent formulation of QCD$+$QED in finite volume.  

The problem addressed in this paper arises every time one needs to
produce an electrically-charged state in a finite periodic box, as 
for instance in the calculation of the proton mass,
and is intrinsically related to the dynamics of the zero-modes of the
gauge field. In a torus with periodic boundary conditions for the
gauge fields, Gauss's law implies that only neutral states belong to
the physical Hilbert space of the theory. One might think to overcome
this limitation by gauge-fixing. For instance in Coulomb gauge the
Gauss's law is locally solved and the Hilbert space splits in
sectors labeled by the total electric charge. However states
generated by electrically-charged local operators in Coulomb gauge are
also charged under large gauge transformations which survive a local
gauge-fixing procedure. Because of this, even after gauge-fixing, the
two-point function $\langle \psi(x) \bar{\psi}(y) \rangle$ vanishes if
$x$ and $y$ are separated in a periodic box. In practice large gauge transformations
act on the gauge field by shifting the global zero-modes
$\int_{T L^3}d^4x\,A_\mu(x)$. Therefore the obstructions 
to the propagation of charged particles on a periodic torus can be
traced back to the functional 
integration over the global zero-modes. 

A possible solution to this problem can be found in
ref.~\cite{Duncan:1996xy} where the first lattice calculation of the
electromagnetic mass splitting of nucleons and light pseudoscalar
mesons has been attempted. The proposed solution consists in
\emph{quenching} a particular set of Fourier modes of the gauge field,
in such a way that the global zero-modes decouple from the dynamics. A
lot of theoretical and algorithmic progress has been made after the
pioneering work of ref.~\cite{Duncan:1996xy}, particularly in the past
few years, leading to recent determinations of the electromagnetic
mass splitting of light pseudoscalar mesons and light baryons, see
refs.~\cite{Borsanyi:2014jba,deDivitiis:2013xla,Basak:2014vca,Ishikawa:2012ix,Aoki:2012st,Blum:2010ym,Tantalo:2013maa,Portelli:2015wna}
for recent works on the subject (see also ref.~\cite{Carrasco:2015xwa}
for the discussion of a method to calculate QED radiative corrections
to the leptonic decays of pseudoscalar mesons). All these works rely
on finite-volume formulations of QED obtained by quenching some
Fourier modes of the gauge field.\footnote{
Recently other approaches have been proposed. In ref.~\cite{Endres:2015gda} the zero modes of the gauge field are lifted by adding a mass term for the photon. The proposal of refs.~\cite{Lehner:2015bga,lehnerlattice2015} consists in combining QCD matrix elements extracted from finite volume simulations with infinite volume QED kernels. 
} 

The particular formulation called \LQED{} is obtained by quenching the
spatial zero-modes of the gauge field at any time, i.e. by enforcing
the constraint $\tilde A_\mu(t,\vec 0)= \int_{L^3}d^3x\,
A_\mu(t,\vec{x}) =0$. As opposed to other formulations, \LQED{}
has a well defined transfer matrix. However
the constraint $\tilde A_\mu(t,\vec 0) = 0$ is \emph{non-local}. Even
though one can argue that the modification generated by the constraint
is a finite-volume effect, many properties of local quantum field
theories are not automatically guaranteed for \LQED{}.
Among these we mention renormalizability, volume-independence
of renormalization constants, the validity of the operator product
expansion and of the Symanzik improvement program. Mild
violations of locality may preserve some of these properties but this
needs to be shown explicitly case by case.


\LQED{} has been studied at one-loop in perturbation theory in
refs.~\cite{Hayakawa:2008an,Borsanyi:2014jba,Basak:2014vca}.
The quenching of the zero-modes does not generate
ultraviolet divergences at one loop, other than the
infinite-volume ones. However it does generate 
unusual phenomena, for instance particles and antiparticles do not decouple
in the non-relativistic limit~\cite{Davoudi:2014qua,Fodor:2015pna}.
This can be seen as a failure of the effective-theory description
which is not surprising if the underlying microscopic theory is
non-local. On the other hand the numerical results of lattice
simulations of \LQED{} performed in
refs.~\cite{Borsanyi:2014jba,Ishikawa:2012ix} might be
viewed as reassuring evidence that the non-localities of \LQED{}
have only mild effects on the hadronic spectrum.
Nevertheless we believe that \LQED{} is not sufficiently well understood
at all orders in perturbation theory.
Our approach is to eliminate any potential problems at the root,
by seeking a consistent formulation of the finite volume theory
that does not require quenching dynamical degrees of freedom.  

In this paper we consider a \emph{local} solution to the problem of
charged particles in finite volume. This solution is not new, it has
been proposed
in~\cite{Polley:1993bn,Wiese:1991ku,Kronfeld:1990qu,Kronfeld:1992ae}
and consists in enforcing \Cstar{} boundary conditions for all fields
along the spatial directions, i.e. in requiring that the fields are
periodic up to charge conjugation. In this theory, which we refer to
as \CQED{}, the zero-modes of the gauge field are absent by
construction because $A_\mu(x)$ is anti-periodic in space, and the
classical problems of the periodic setup are avoided from the very
beginning. We show that a complete description of a certain class of
electrically-charged states can be obtained without relying either on
perturbation theory or on gauge-fixing. As we shall discuss in detail,
this class of states covers most of the relevant spectroscopic
applications and includes the proton, the neutron, the charged pions, the charged
kaons, the charged $D$ and $B$ mesons and the $\Sigma^\pm$ baryons. The proposed
construction is based on the fact
that \Cstar{} boundary conditions break the global gauge symmetry
group $\U(1)$ down to its discrete subgroup $\mathbb{Z}_2$. In other
words charge conservation is partially violated by the boundary
conditions. The full group of gauge transformations splits in two
disconnected components: the subgroup of local gauge transformations
which are connected to the identity, and the set resulting by the
composition of local gauge transformations with the nontrivial global
gauge transformation. In this setup one can construct states that are
invariant under local gauge transformations but not under global
gauge transformations, and these can be identified as
electrically-charged states. 

Along with charge conservation, \Cstar{} boundary conditions partially
violate flavour conservation. This happens because flavour-charged
particles traveling once around the torus turn into their antiparticles,
and therefore change their flavour content. Being associated with the propagation 
of massive colorless particles, these effects are exponentially suppressed with the volume. 
We study in detail the pattern of flavour violation in \CQED{},
particularly in the case when electromagnetic interactions are coupled
to QCD, and quantify these effects in the framework of a generic effective theory of hadrons. 
In particular we show that, although the $\Omega^-$ and $\Xi^-$ baryons
can mix with lighter states because of the boundary conditions,
the exponential suppression is so strong that these mixings can hardly represent
a problem in numerical simulations.

Finite-volume effects on the masses of charged particles are
considerably smaller in \CQED{} than in \LQED{}. When these corrections are expanded in a power series in $1/L$, at order $\alpha_{em}$ in both theories the $1/L$ and $1/L^2$ finite-volume corrections to the mass of a charged particle are universal, i.e. they do not depend on the spin and on the internal structure of the particle (for \LQED{} see
refs.~\cite{Borsanyi:2014jba,Basak:2014vca,Davoudi:2014qua,Fodor:2015pna,Lee:2015rua}).
We show that these universal corrections are always appreciably larger in \LQED{} than in \CQED{}. For instance at $mL=4$ we gain a factor of about 2 with three \Cstar{}-periodic spatial directions and a factor of about 5 with a single \Cstar{}-periodic spatial direction, see figure~\ref{fig:fvecomparison}. More importantly, the spin and structure-dependent corrections are $\mathcal{O}(1/L^3)$ in \LQED{}, while they are only
$\mathcal{O}(1/L^4)$ in \CQED{}. This extra suppression can be seen as a direct effect of locality.

The paper is organized as follows. In section~\ref{sec:cstar} we
introduce \Cstar{} boundary conditions and study the symmetries of
\CQED{}. In section~\ref{sec:interpolating} we introduce the gauge
invariant interpolating operators for charged particles and study
their properties. In section~\ref{sec:qed+qcd} we couple
electromagnetic and strong interactions and study the symmetries of
\CQCDQED{}. In section~\ref{sec:perturbative} we discuss the finite
volume corrections to the masses of charged hadrons. In
section~\ref{sec:lattice} we discuss the details of the lattice
implementation of \Cstar{} boundary conditions and of the proposed
gauge invariant interpolating operators. We draw our conclusions in
section~\ref{sec:conclusions}. 

The paper contains four appendices with the explicit derivation of some of the results 
presented in the main body of the paper. 
The material discussed in the appendices is technical and some of it is, we believe, original. 
Appendix~\ref{app:mixing} presents a detailed study of some
flavour-violation processes in \CQCDQED{}, in the context of a generic effective theory of hadrons. This analysis requires an extension of the techniques developed to study finite-volume effects in~\cite{Luscher:1985dn}, and it is complicated by the need to keep track of flavour flow and violations through all possible Feynman diagrams.
In appendix~\ref{app:massformula} we give an \textit{ab-initio} derivation (i.e. without
using an effective description of hadrons) of the
power-law finite-volume corrections on the mass of charged hadrons in \CQCDQED{}. The
coefficients of the expansion in powers of $1/L$ are expressed in terms of physical
quantities, i.e. derivatives of the forward Compton amplitude for the scattering of a soft photon on the charged hadron.  The authors are
convinced that the technology developed in these appendices will find
other uses in the field.


\section{\CQED{}}
\label{sec:cstar}
In this section we introduce the finite-volume theory \CQED{} and
study its symmetries. For simplicity, we consider the case of a maximally symmetric torus with linear size equal to $L$, with fields obeying \Cstar{} boundary conditions in all space directions. The Euclidean time
direction can be either infinite or compact with linear size $T$. In
the latter case the corresponding boundary
conditions for the fields will be left unspecified. Common choices are
periodic, Schr\"odinger Functional (SF), open or open-SF
boundary conditions.  

The action of \CQED{} is given by 
\begin{gather}
S[A,\psi] = \int_{L^3T} d^4x\, \left\{ 
\frac{1}{4e^2}F_{\mu\nu}F_{\mu\nu} + \sum_{f=1}^{N_f}
\bar{\psi}_f\left(\gamma_\mu \lrD{}_\mu^f + m_f\right)\psi_f\right\}\ .
\label{eq:action}
\end{gather}
The field strength and covariant derivative are defined as
\begin{gather}
F_{\mu\nu}(x) = \partial_\mu A_\nu(x) - \partial_\nu A_\mu(x) \ , \nonumber \\
\lrD{}_\mu^f = \lrp{}_\mu -\imath q_f A_\mu \ ,
\end{gather}
where the
left-right derivative $\lrp{}_\mu = \frac{1}{2} ( 
\rp{}_\mu - \lp{}_\mu )$ is defined in terms of the partial derivative
$\rp{}_\mu$  acting to the right and the partial derivative
$\lp{}_\mu$ acting to the left. In our notation $q_f$ is the electric charge of the
$f$-th flavour normalized to the electric charge of the positron (i.e.
$q_f$ does not include the coupling constant $e$). Throughout the
paper we use this normalization for the electric charge.

Fields obey \Cstar{} boundary
conditions under translations in the 
three space directions,  
\begin{eqnarray} 
&&A_\mu(x+\hat L_i) = A_\mu^{\mathcal C}(x) = -A_\mu(x)\,,
\nonumber \\
&&\psi_f(x+\hat L_i) = \psi^{\mathcal C}_f(x) = C^{-1} \bar\psi^T_f(x)\,, 
\nonumber \\
&&\bar \psi_f(x+\hat L_i) = \bar \psi^{\mathcal C}_f(x) = - \psi^T_f(x) C \ ,
\label{eq:csbc}
\end{eqnarray}
where $\hat L_i$ is $L$ times the unit vector in direction $i$. The charge conjugation matrix can be taken to be any
invertible matrix $C$ with unit determinant such that 
\begin{gather}
C^{-1}\gamma_\mu C = -\gamma_\mu^T \ ,
\end{gather}
where $\gamma_\mu$ are the Euclidean gamma matrices. In four
dimensions such a matrix exists and satisfies 
\begin{gather}
C^T =-C \ , \qquad C^\dag = C^{-1} \ ,
\end{gather}
independently of the particular representation of the gamma matrices.

Notice that the action density eq.~(\ref{eq:action}) is the same as in
infinite volume and 
it is therefore invariant under charge conjugation. Since a shift of a
period in space corresponds to charge conjugation, the action density
is periodic in space.

We are now going to study the symmetries of \CQED, in turn gauge transformations, spatial translations, parity and flavour symmetries.

\subsection{Gauge transformations}

Gauge transformations are defined in the usual way
\begin{gather}
A_\mu^{[\alpha]}(x) = A_\mu(x) + \partial_\mu \alpha(x)\ ,
\nonumber \\
\psi_f^{[\alpha]}(x) = e^{\imath q_f  \alpha(x)}\psi_f(x) \ ,
\nonumber \\
\bar \psi_f^{[\alpha]}(x) = e^{-\imath q_f \alpha(x)}\bar \psi_f(x) \ .
\end{gather}
Only gauge transformations that do not change the boundary conditions of the fields are admissible. Translating the transformed field by a period along a spatial direction yields
\begin{gather}
A_\mu^{[\alpha]}(x+\hat L_i) =
A_\mu(x+\hat L_i) + \partial_\mu \alpha (x+\hat L_i)
\nonumber \\
\qquad =
-A_\mu(x) + \partial_\mu \alpha(x+\hat L_i) =
- A_\mu^{[\alpha]}(x)  + \partial_\mu [ \alpha(x+\hat L_i) + \alpha(x)]
\ .
\end{gather}
The transformed field $A_\mu^{[\alpha]}(x)$ is anti-periodic if and only if the gauge transformation satisfies
\begin{gather}
\partial_\mu\alpha(x+\hat L_i) = -\partial_\mu\alpha(x) \ ,
\end{gather}
i.e. $\alpha(x)$ can be decomposed into an anti-periodic function plus a generic constant. The boundary conditions for fermions constrain this constant. Translating a fermion field by a period along a spatial direction yields
\begin{gather}
\psi^{[\alpha]}_f(x+\hat L_i) = 
e^{\imath q_f  \alpha(x+\hat L_i)}\psi_f(x+\hat L_i)
\nonumber \\
\qquad =
e^{\imath q_f  \alpha(x+\hat L_i)} C^{-1}\bar\psi^T_f(x) =
e^{\imath q_f  [\alpha(x+\hat L_i)+\alpha(x)]} C^{-1} [\bar\psi^{[\alpha]}]^T_f(x)
\ .
\end{gather}
The transformed field $\psi^{[\alpha]}_f(x)$ satisfies \Cstar{} boundary conditions if and only if an integer $n_f$ exists such that
\begin{gather}
\alpha(x) = \beta(x) + \frac{n_f \pi}{q_f}\ ,
\qquad
\beta(x+\hat L_i) = -\beta(x) \ .
\end{gather}
Notice that this equation has to be satisfied for all fermion fields and
for any pair of charges. In the physically relevant case\footnote{If
  two of the charges have irrational ratio, 
then one of the $n_f$ has to be zero and consequently $\alpha(x)$ has to
be anti-periodic} all charges
$q_f$ are integer multiples of an elementary charge $q_{el}$,
therefore the gauge transformation $\alpha(x)$ preserves the boundary
conditions of all fields if and only if an integer $n$ exists such
that 
\begin{gather}
\alpha(x) = \beta(x) + \frac{n\pi}{q_{el}} \ .
\label{eq:residual_gauge_symmetry}
\end{gather}
Quantization of the electric charge can be seen as a consequence of the fact that the gauge group is the compact $\U(1)$. A generic gauge transformation is assigned by choosing a phase factor $\Lambda(x) = e^{i q_{el} \alpha(x)}$ in each point of spacetime. A matter field with charge $q_f$ transforms with $\Lambda(x)^{\hat{q}_f}$ where $\hat{q}_f = q_f/q_{el}$ is an integer, i.e. accordingly to some irreducible representation of the gauge group $\U(1)$. This analysis can be restated in terms of operators: given the electric-charge operator $Q$, the generator of global gauge transformations is
\begin{gather}
\hat{Q} = \frac{Q}{q_{el}} \ , 
\end{gather}
and has only integer eigenvalues. \Cstar{} boundary conditions break
the $\U(1)$ group of global gauge transformations. In fact
eq.~\eqref{eq:residual_gauge_symmetry} implies that the only allowed
global gauge transformations are $\Lambda = \pm 1$, i.e. the global
$\U(1)$ is broken down to $\mathbb Z_2$. Breaking of the global
$\U(1)$ implies a partial violation in electric-charge conservation:
$Q$ is not conserved but the quantum number $(-1)^{\hat{Q}}$ is. The
origin and consequences of this violation will be discusses in more
details in subsection~\ref{sub:flavour} for the case of \CQED{} in
isolation, and in section~\ref{sec:qed+qcd} for the case of
\CQCDQED{}. 

Eq.~\eqref{eq:residual_gauge_symmetry} implies that the group of gauge
transformations is disconnected. Only gauge transformations with
$n=0$, i.e. with $\alpha(x)$ anti-periodic in space, are continuously
connected to the identity. We will refer to these gauge
transformations as \textit{local gauge transformations}. Note that
the large gauge transformations have a very simple structure (they are
just the composition of a global gauge transformation and a local
gauge transformation). This contrasts with the case of periodic
boundary conditions in space, where large gauge 
transformations are linear in the coordinates (i.e. $\alpha(x) =
2\pi n x_i/L$ with some integer $n$).

\subsection{Translations}

\Cstar{} boundary conditions preserve translational invariance and charge conjugation. Even though in infinite volume the momentum and the \C{} quantum number are unrelated, this is not true in \CQED{}. Eqs.~\eqref{eq:csbc} imply that the translation of a generic (elementary or composite) field $\phi(x)$ by $\hat{L}_i$ is equivalent to a charge conjugation
\begin{gather}
\phi(x+\hat{L}_i) = 
\phi^{\mathcal{C}}(x) \ .
\end{gather}
The \C{}-even and \C{}-odd components of the field $\phi(x)$ are
\begin{gather}
\phi_{\pm}(x) = \frac{\phi(x) \pm \phi^{\mathcal C}(x)}{\sqrt{2}}\ .
\label{eq:phiplusminus}
\end{gather}
$\phi_+(x)$ is periodic in space while $\phi_-(x)$ is
anti-periodic. The two components have different Fourier
representations. Since we want to leave the time boundary conditions
unspecified, we expand our fields in the time-momentum representation, 
\begin{gather}
\phi_\pm(x) = \frac{1}{L^3} \sum_{\vec{p} \in \Pi_\pm} \tilde \phi_\pm(x_0,\vec{p}) e^{\imath \vec{p} \vec{x}} \ ,
\end{gather}
where $\Pi_+$ is the set of periodic momenta and $\Pi_-$ is the set of anti-periodic momenta,
\begin{gather}
\Pi_+ = \left\{ \frac{2\pi}{L} \vec{n} \ \big| \ \vec{n} \in \mathbb{Z}^3  \right\} \ , 
\nonumber \\
\Pi_- = \left\{ \frac{\pi}{L} \left( 2\vec{n} + \bar{\vec{n}} \right) \  \big| \ \vec{n} \in \mathbb{Z}^3 , \ \bar{\vec{n}} = \left( 1,1,1\right) \right\}
 \ .
\label{eq:fourier}
\end{gather}

Notice that the $A_\mu(x)$ field is \C{}-odd and it has only the anti-periodic component, while the fields $\psi_f(x)$ contain both,
\begin{gather}
A_\mu(x) = \frac{1}{L^3}\sum_{\vec{p} \in \Pi_-} \tilde
A_\mu(x_0,\vec{p}) e^{\imath \vec{p} \vec{x}} \ , 
\nonumber \\
\psi_{f,\pm}(x) = \frac{1}{L^3}\sum_{\vec{p} \in \Pi_\pm} \tilde
\psi_{f,\pm}(x_0,\vec{p}) e^{\imath \vec{p} \vec{x}} \ . 
\end{gather}
The two $\psi_{f,\pm}$ components of the fermion fields satisfy the (anti)~Majorana condition
\begin{gather}
\psi_{f,\pm}(x) = \pm C^{-1} [\bar{\psi}_{f,\pm} ]^T(x) \ .
\end{gather}

\subsection{Parity}

Even though not in a trivial fashion, parity is conserved by \Cstar{} boundary conditions. Under parity the fields transform like
\begin{align}
&A_0(x) \ \to\ A_0(x_P) \ , & &  \psi_f(x) \ \to\ \eta_P \gamma_0 \psi_f(x_P) \ , & \hspace*{2cm}
\nonumber \\
& A_k(x) \ \to\ -A_k(x_P) \ , & & \bar{\psi}_f(x) \ \to\ \eta_P^* \bar{\psi}_f(x_P) \gamma_0 \ ,
\end{align}
where $x_P = (x_0,-\vec{x})$. In infinite volume $\eta_P$ is a generic complex phase. For each choice of $\eta_P$ one obtains a different but equally good parity operator. A customary choice amounts to $\eta_P=1$. However the parity operator defined in this way does not commute with the charge conjugation operator that we have used to define the \Cstar{} boundary conditions. A more natural choice is $\eta_P=\imath$. The corresponding parity transformation $\mathcal{P}$ commutes with the charge conjugation operator.
This can be shown explicitly by acting on the elementary fields with
charge conjugation $\mathcal{C}$ first and parity $\mathcal{P}$
after, and by comparing the result with the same operations applied in
reversed order. For example, in the case of the fermion field we have
\begin{gather}
\psi_f(x)\ \xrightarrow{\mathcal{C}}\ C^{-1} \bar{\psi}^T_f(x)\ \xrightarrow{\mathcal{P}}\ 
- \imath C^{-1} \gamma_0^T \bar{\psi}^T_f(x_P) \ ,
\end{gather}
and
\begin{gather}
\psi_f(x)\ \xrightarrow{\mathcal{P}}\ \imath \gamma_0 \psi_f(x_P)\ \xrightarrow{\mathcal{C}}\ 
\imath \gamma_0 C^{-1} \bar{\psi}^T_f(x_P) \ .
\end{gather}
The results of the two transformations are shown to be equal by using
$C^{-1} \gamma_0^T C = - \gamma_0$. The reader can check that this
conclusions applies to the other fields.

Since $\mathcal{P}$ leaves the action and the \Cstar{} boundary conditions
unchanged, it is an exact symmetry in finite volume. Even though parity
will play no special role in this paper, we notice that the parity
transformations can be easily used to construct operators that have
definite parity.

\subsection{Flavour symmetries}
\label{sub:flavour}

\Cstar{} boundary conditions violate flavour (and consequently
electric charge) conservation. The violation arises because a flavour-charged
particle flips the sign of its flavour content by turning into its antiparticle
when it travels once around the torus. We are now going to show that flavour is violated by two units
at the time in this process and that this effect is exponentially
suppressed with the volume.  In this subsection and in
section~\ref{sec:qed+qcd} we will argue that flavour violation does
not represent a limitation to the use of \Cstar{} boundary conditions
in most of the relevant applications. 

\tikzset{
  photon/.style={decorate, decoration={snake}, draw=black},
  fermionpm/.style={draw=black, postaction={decorate},decoration={markings,mark=at position .55 with {\arrow[scale=2]{>}}}},
  fermionpp/.style={draw=black, postaction={decorate},decoration={markings,mark=at position .45 with {\arrow[scale=2]{>}},mark=at position .60 with {\arrow[scale=2]{<}}}},
  fermionmm/.style={draw=black, postaction={decorate},decoration={markings,mark=at position .45 with {\arrow[scale=2]{<}},mark=at position .60 with {\arrow[scale=2]{>}}}},
  vertex/.style={draw,shape=circle,fill=black,minimum size=3pt,inner sep=0pt},
}

We start by considering the theory with a single species of charged
particles with unit charge, e.g. the electron. In this case flavour
coincides with the electric charge $Q$ and with the generator
$\hat{Q}$ of global gauge transformations. The detailed way charge
conservation is violated by finite-volume effects can be easily
understood by means of Feynman diagrams. We assume here some gauge
fixing that we do not need to specify at this level. The theory in
finite volume has the same interaction vertex as the infinite-volume
one which, in particular, conserves electric charge. The violation of
charge conservation is visible in those terms in the action that are
sensitive to the \Cstar{} boundary conditions, i.e. the ones
containing spatial derivatives. In other words, charge violation is
generated by the propagators, which we will discuss in detail. 

In order to write down the free propagators, one needs to keep into
account the fact that a free particle is able to travel around the
torus. If it travels once around a direction with \Cstar{} boundary
conditions, the particle turns into its antiparticle. The winding
numbers of the particle world-line around each spatial direction can be
organized into a vector $\vec{n} \in \Z^3$. By defining  
\begin{gather}
\langle \vec{n} \rangle = \sum_{i=1}^3 n_i \mod 2
\label{eq:langlenrangle}
\end{gather}
we can separate those winding numbers characterised by $\langle
\vec{n} \rangle = 1$ that flip the electric charge of the particle
from the winding numbers characterised by $\langle \vec{n} \rangle =
0$ that do not. We do not need the explicit expression of the gauge
field propagator as the photon carries neither electric nor flavour
charge. Concerning the matter field, in coordinate space we have 
\begin{gather}
\langle \psi(x) \bar{\psi}(y) \rangle
\ \, \quad =\quad
\tikz[node distance=15mm] {
   \coordinate[label=above:{\small $x$}] (v1);
   \coordinate[right=of v1,label={\small $y$}] (v2);
   \draw[fermionpm] (v1) -- (v2);
}
\quad =\quad\quad
\sum_{\langle \vec{n} \rangle=0} S(x - y + \hat{L}_i n_i) \ ,
\\
\langle \psi(x) \psi^T(y) \rangle
\quad =\quad
\tikz[node distance=15mm] {
   \coordinate[label=above:{\small $x$}] (v1);
   \coordinate[right=of v1,label={\small $y$}] (v2);
   \draw[fermionpp] (v1) -- (v2);
}
\quad =\quad
- \sum_{\langle \vec{n} \rangle=1} S(x - y + \hat{L}_i n_i)C^{-1} \ ,
\label{eq:pp_prop}
\\
\langle \bar{\psi}^T(x) \bar{\psi}(y) \rangle
\quad =\quad
\tikz[node distance=15mm] {
   \coordinate[label=above:{\small $x$}] (v1);
   \coordinate[right=of v1,label={\small $y$}] (v2);
   \draw[fermionmm] (v1) -- (v2);
}
\quad =\quad\quad
\sum_{\langle \vec{n} \rangle=1} C S(x - y + \hat{L}_i n_i) \ ,
\label{eq:pbpb_prop}
\end{gather}
where $S(x)$ is the infinite-volume fermion propagator. Notice that
the $\psi\psi^T$ and $\bar{\psi}^T\bar{\psi}$ propagators vanish in
infinite spatial volume as the sums in eqs.~\eqref{eq:pp_prop} and
\eqref{eq:pbpb_prop} do not include $\vec{n}=0$. They are precisely
the source of violation of charge conservation. The violation is not
arbitrary, but amounts to a $\Delta Q= \pm 2$ every time one of these
propagator is inserted. This shows explicitly that the electric charge
$Q$ is not conserved, but the quantum number $(-1)^{\hat{Q}}$ is,
which means 
\begin{gather}
\Delta \hat{Q}= 0 \mod 2 \ .
\label{eq:dQh2}
\end{gather}
Time evolution mixes all sectors with odd electric charge among each
other, and all sectors with even electric charge among each other. For
example a single-electron state can mix with a three-electron state
but not with the vacuum, see figure~\ref{fig:mixing}. This in
particular means that, chosen some suitable interpolating operator as
we will discuss in section~\ref{sec:interpolating}, single-electron
states can be selected by looking at the leading decaying exponential
in two-point functions. However two-electron states cannot be
extracted in the same way, as the leading decaying exponential in a
two-point function constructed with an operator with charge equal to 2
will select always the vacuum. As we will discuss in
section~\ref{sec:qed+qcd} this is sufficient in most of the
interesting low-energy applications in \CQCDQED{}. 

From the discussion above it is obvious that the violation arises only
from the charged particle that travels at least once around the
torus. If the fermion is massive we have 
\begin{gather}
\langle \psi(x) \psi^T(y) \rangle \sim \langle \bar{\psi}^T(x) \bar{\psi}(y) \rangle 
\sim \left( \frac{m}{L} \right)^{\frac{3}{2}} e^{-m L} \ ,
\end{gather}
for $L \to \infty$, and charge-violating diagrams are
exponentially suppressed.

\begin{figure}
\centering

\begin{subfigure}{0.4\textwidth}
  \begin{tikzpicture}[node distance=15mm]
    \coordinate[vertex] (v1);
    \coordinate[right=of v1,label=right:$\gamma$] (f1);
    \coordinate[above left =of v1,label=left :$e^+$] (e1);
    \coordinate[below left =of v1,label=left :$e^+$] (e2);
    \draw[fermionpm] (e1) -- (v1);
    \draw[fermionpp] (v1) -- (e2);
    \draw[photon] (v1) -- (f1);
  \end{tikzpicture}

  \caption{}
  \label{fig:mixing:a}
\end{subfigure}
\hspace{1cm}
\begin{subfigure}{0.4\textwidth}
  \begin{tikzpicture}[node distance=15mm]
    \coordinate[vertex] (v1);
    \coordinate[vertex,below right =of v1] (v2);
    \coordinate[above left =of v1,label=left :$e^+$] (e1);
    \coordinate[below left =of v1,label=left :$e^+$] (e2);
    \coordinate[below =of e2,label=left :$e^+$] (e3);
    \coordinate[right=of v2,label=right:$e^-$] (e4);
    \draw[fermionpm] (e1) -- (v1);
    \draw[fermionpp] (e2) -- (v1);
    \draw[fermionpm] (e3) -- (v2);
    \draw[fermionpp] (e4) -- (v2);
    \draw[photon] (v1) -- (v2);
  \end{tikzpicture}

  \caption{}
  \label{fig:mixing:b}
\end{subfigure}

\caption{\label{fig:mixing} (a) Diagram contributing to the $e^+ e^+
  \to \gamma$ process, which involves one $e^+$ traveling around the
  torus and flipping charge. (b) Diagram contributing to
  the $e^+ e^+ e^+ \to e^-$ process.} 
\end{figure}
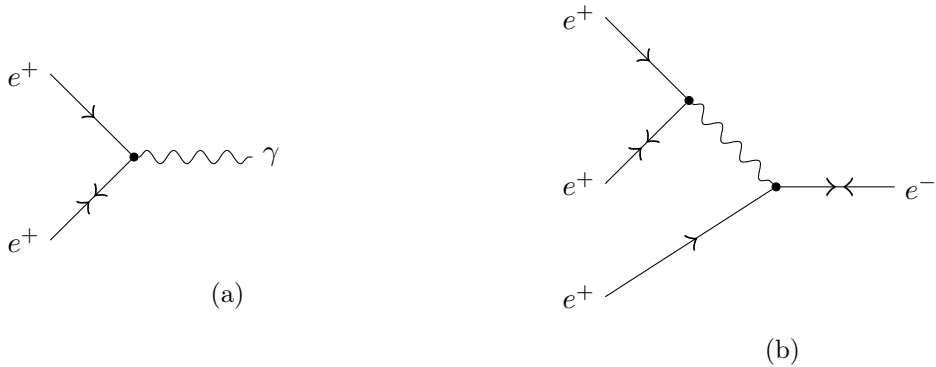

In the case of $N_f$ flavours the infinite-volume theory has a $\U(1)^{N_f}$ flavour symmetry corresponding to independent phase rotations of each flavour. We will denote the generator of the $f$-th $\U(1)$ by $F_f$. Notice that the electric charge is a linear combination of the flavour-symmetry generators,
\begin{gather}
Q = \sum_{f=1}^{N_f} q_f F_f \ ,
\label{eq:QF}
\end{gather}
where $q_f$ is the electric charge of the $f$-th flavour. In infinite volume each $F_f$ is conserved independently. \Cstar{} boundary conditions break the flavour symmetry group down to a $\Z_2^{N_f}$,\footnote{
If $n_f$ out of the $N_f$ flavours are degenerate (i.e. same mass and same electric charge), the $\U(1)^{n_f}$ flavour subgroup is lifted to a $\U(n_f)$ flavour symmetry. \Cstar{} boundary conditions break this down to its natural $\O(n_f)$ subgroup. We mention this special case for completeness, but it is not relevant for the purpose of this paper.
}
and this implies that only each $(-1)^{F_f}$ is conserved,
i.e. violations can occur only in multiples of two,
\begin{gather}
\Delta F_f= 0 \mod 2\ .
\label{eq:dF2}
\end{gather}
Notice that the $\Delta Q$ has to be a multiple of 2 only if all flavours
have equal electric charge. In the general case it is replaced by
eqs.~\eqref{eq:dQh2}, \eqref{eq:QF} and~\eqref{eq:dF2}. This
observation will play an important role in section~\ref{sec:qed+qcd}
where we will discuss QCD coupled to QED.

\section{Gauge-invariant interpolating operators}
\label{sec:interpolating}
We are concerned with physical observables, i.e. observables that are
invariant under local gauge transformations. Often these observables
are extracted from intermediate quantities defined in a particular
gauge. For instance masses of charged particles are usually extracted
from the long-distance behaviour of two-point functions after the
photon field has been gauge-fixed. Although this is a necessary step
in perturbation theory it can be completely avoided non-perturbatively  
without adding any particular complication. Keeping in mind that no
issue arises with gauge fixing for QED and that a particular gauge can
be chosen at any time, we think that it is more natural to rely on a
completely gauge-invariant formulation. In this section we show how to
construct states that are invariant under local gauge transformations
and electrically charged at the same time, i.e. they have
$(-1)^{\hat{Q}}=-1$. This will be achieved by acting with
suitably-constructed interpolating operators on the vacuum. Even
though we discuss primarily how to apply this construction to the
calculation of charged-particle masses from two-point functions, the
same interpolating operators can be used to extract other physical
quantities, e.g. decay rates, in a completely gauge-invariant
fashion. 

To simplify the notation in this section we consider a single matter
field with charge $q$. The generalization of the
following discussion to the case of several flavours with different
charges is completely straightforward. Consider the operator~\cite{Dirac:1955uv} 
\begin{gather}
\Psi_J(x) = e^{\imath q \int d^4y\, A_\mu(y) J_\mu(y-x)}\ \psi(x) \ ,
\label{eq:interpolating}
\end{gather}
where $\psi(x)$ is the matter field and $J_\mu(x)$ is a generic function or distribution that satisfies
\begin{gather}
\partial_\mu J_\mu(x) = \delta^4(x) \ , 
\qquad
J_\mu(x+\hat{L}_i) = -J_\mu(x) \ .
\label{eq:inter-J-eq}
\end{gather}
In case of periodic boundary conditions in time $J_\mu(x)$ is chosen to be periodic as well. Under a global transformation $\psi(x) \to e^{\imath q \alpha} \psi(x)$, the above operator transforms like $\Psi_J(x) \to e^{\imath q \alpha}\Psi_J(x)$, which implies that in infinite volume $\Psi_J(x)$ would have electric charge equal to $q$. In finite volume we have already noticed that $\alpha$ can be only $0$ or $\pi/q$, which implies that the operator $\Psi_J(x)$ has quantum number $(-1)^{\hat{Q}}=-1$.  The non-local factor
\begin{gather}
  \Theta(x) = e^{\imath q \int d^4y \ A_\mu(y) J_\mu(y-x)}
\end{gather}
transforms under a local gauge transformation that is anti-periodic in space as
\begin{align}
\Theta(x)
\to&
\Theta(x)\ e^{\imath q \int d^4y \ \partial_\mu \alpha(y)\, J_\mu(y-x)}
\ =\
\Theta(x)\ e^{-\imath q \int d^4y \ \alpha(y)\, \partial_\mu J_\mu(y-x)}
\nonumber \\
&=
\Theta(x)\ e^{-\imath q \alpha(x)} \ .
\end{align}
Notice that the product $\alpha(y) J_\mu(y-x)$ is periodic with respect to $\vec{y}$. Given also the boundary conditions in time, no boundary terms arise from the integration by parts. The extra factor $e^{-\imath q \alpha(x)}$ obtained by gauge-transforming $\Theta(x)$ cancels the analogous factor coming from the transformation of $\psi(x)$, making $\Psi_J(x)$ invariant. 

Summarising, the non-local operator $\Psi_J(x)$ has $(-1)^{\hat{Q}}=-1$ and is invariant under local gauge transformations. It also satisfies the same boundary conditions as the field $\psi(x)$, and therefore operators with definite momentum can be easily constructed by considering the \C-even and \C-odd components of $\Psi_J(x)$ as done in eqs.~\eqref{eq:phiplusminus} for a generic operator $\phi(x)$.

If the function $J_\mu(x)$ is chosen to be proportional to $\delta(x_0)$, then the operator $\Psi_J(x)$ is local in time, i.e. it is a function of the elementary fields at the time $x_0$ only. In this case $\Psi_J(x)$ maps naturally to an operator acting on the Hilbert space. The state $\Psi_J(x) | 0 \rangle$ obtained acting with the interpolating operator on the vacuum is invariant under local gauge transformations and has electric charge $(-1)^{\hat{Q}}=-1$. By decomposing the Euclidean two-point function $\langle\Psi_J(x)
\bar\Psi_J(0)\rangle$ in decaying exponentials in $x_0$, one can
extract the spectrum of the 
gauge-invariant Hamiltonian.
The energy levels are gauge-invariant by construction and they do not
depend on the particular choice of $J_\mu(x)$, as they are a property
of the Hamiltonian rather than of the interpolating operator (as long
as this is local in time). We will refer to the energy of the lightest
state propagating in the Euclidean two-point function as the finite-volume
mass of the charged particle. We assume that this quantity has an infinite-volume limit which can be interpreted as the mass of the charged particle.\footnote{This issue is not trivial in QED because of the absence of a mass gap. See for instance the discussion in chapter 6 of \cite{HaagBook} or chapter 6 of~\cite{StrocchiBook}, and references therein.}

The whole construction presented above is based on the assumption that solutions of eq.~\eqref{eq:inter-J-eq} exist. If periodic boundary conditions were employed in all spatial directions eq.~\eqref{eq:inter-J-eq} would have no solutions. In the case of \Cstar{} boundary conditions we will construct explicitly some possible choices for the function $J_\mu(x)$. The first one is defined by the equations 
\begin{gather}
J_0(x) = 0 \ , 
\qquad
J_k(x) = \delta(x_0) \partial_k \Phi(\vec{x}) \ , 
\qquad
\partial_k\partial_k \Phi(\vec{x}) = \delta^{3}(\vec{x}) \ ,
\end{gather}
where $x=(x_0,\vec{x})$ and $\Phi(\mathbf x)$ is anti-periodic. An explicit (convergent) representation for $\Phi(\mathbf x)$ is given in terms of the heat-kernel
\begin{gather}
\Phi(\mathbf x) = - \frac{1}{L^3} \int_0^\infty du \ \sum_{\vec{p} \in \Pi_-} e^{-u \vec{p}^2 + \imath \vec{p} \vec{x}} \ .
\end{gather}
With this choice the operator $\Psi_J(x)$ can be written like 
\begin{gather}
\Psi_{\mathbf c} (x) = 
e^{-\imath q \int d^{3}y\, \partial_k A_k(x_0,\vec{y})\, \Phi(\vec{y}-\vec{x})}\ \psi(x) \ . 
\label{eq:cont_inter_coulomb}
\end{gather}
Notice that in Coulomb gauge $\Psi_{\mathbf c}(x)=\psi(x)$, and therefore the gauge invariant correlator $\langle \Psi_{\mathbf c}(x)\bar \Psi_{\mathbf c}(y)\rangle$ is identical to usual correlator $\langle\psi(x)\bar \psi(y)\rangle$ in Coulomb gauge. In other words, $\Psi_{\mathbf c}(x)$ is the unique gauge-invariant extension of the operator $\psi(x)$ defined in Coulomb gauge. This in particular shows explicitly the gauge-invariance of the mass extracted in Coulomb gauge.

Another possible choice is given by
\begin{gather}
J_\mu(x) =
\frac{1}{2} \delta_{\mu,k} \ \text{sgn}(x_k) \ \prod_{\nu \neq k}
\delta(x_\nu) \ .
\end{gather}
Once this equation is inserted in eq.~\eqref{eq:interpolating}, it yields the following interpolating operator
\begin{gather}
\Psi_{\mathbf s}(x) = e^{-\frac{\imath q}{2} \int_{-x_k}^0 ds \ A_k(x + s \hat{k})} \psi(x) e^{\frac{\imath q}{2} \int_{0}^{L-x_k} ds \ A_k(x + s \hat{k})} \ .
\label{eq:cont_inter_string}
\end{gather}
This choice generates a string wrapping around the torus along the direction $k$, chosen among the ones with \Cstar{} boundary conditions (see figure~\ref{fig:string}). The operator $\Psi_{\mathbf s}(x)$ is less symmetric with respect to $\Psi_{\mathbf c}(x)$ but, as discussed in section~\ref{sec:lattice}, it might be more practical to use in numerical simulations, especially in the framework of compact \CQED.

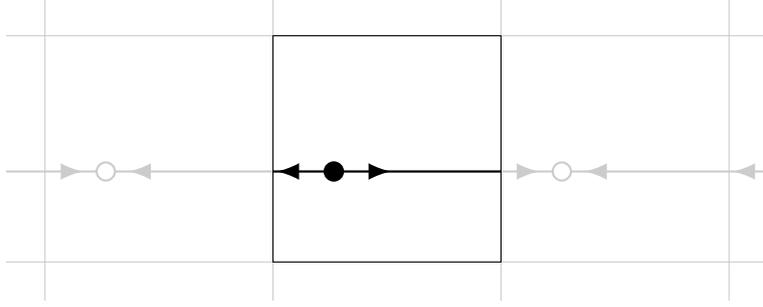
\begin{figure}
\centering

\begin{tikzpicture}

\begin{scope}
\clip (-3.5,-0.5) rectangle (6.5,3.5);

\draw[step=3] (-6,-3) grid (9,6);

\draw[thick,postaction={decorate},decoration={markings,mark=at position .25 with {\arrow[scale=1.6]{latex}},mark=at position .9 with {\arrow[scale=1.6]{latex}}}] (0.8,1.2) -- ++(3,0);
\draw[thick,postaction={decorate},decoration={markings,mark=at position .25 with {\arrow[scale=1.6]{latex}},mark=at position .9 with {\arrow[scale=1.6]{latex}}}] (0.8,1.2) -- ++(-3,0);
\draw[thick,postaction={decorate},decoration={markings,mark=at position .25 with {\arrow[scale=1.6]{latex}},mark=at position .9 with {\arrow[scale=1.6]{latex}}}] (0.8,1.2)++(6,0) -- ++(-3,0);
\draw[thick,postaction={decorate},decoration={markings,mark=at position .25 with {\arrow[scale=1.6]{latex}},mark=at position .9 with {\arrow[scale=1.6]{latex}}}] (0.8,1.2)++(-6,0) -- ++(3,0);

\filldraw[thick,fill=black,draw=black] (0.8,1.2) circle (.12);
\filldraw[thick,fill=white,draw=black] (0.8,1.2) ++(-3,0) circle (.12);
\filldraw[thick,fill=white,draw=black] (0.8,1.2) ++(3,0) circle (.12);

\fill[white,fill opacity=.8] (-6,-3) rectangle (9,6);

\end{scope}

\draw (0,0) rectangle (3,3);

\clip (0,0) rectangle (3,3);

\draw[thick,postaction={decorate},decoration={markings,mark=at position .25 with {\arrow[scale=1.6]{latex}},mark=at position .9 with {\arrow[scale=1.6]{latex}}}] (0.8,1.2) -- ++(3,0);
\draw[thick,postaction={decorate},decoration={markings,mark=at position .25 with {\arrow[scale=1.6]{latex}},mark=at position .9 with {\arrow[scale=1.6]{latex}}}] (0.8,1.2) -- ++(-3,0);
\draw[thick,postaction={decorate},decoration={markings,mark=at position .25 with {\arrow[scale=1.6]{latex}},mark=at position .9 with {\arrow[scale=1.6]{latex}}}] (0.8,1.2)++(6,0) -- ++(-3,0);
\draw[thick,postaction={decorate},decoration={markings,mark=at position .25 with {\arrow[scale=1.6]{latex}},mark=at position .9 with {\arrow[scale=1.6]{latex}}}] (0.8,1.2)++(-6,0) -- ++(3,0);

\filldraw[thick,fill=black,draw=black] (0.8,1.2) circle (.12);
\filldraw[thick,fill=white,draw=black] (0.8,1.2) ++(-3,0) circle (.12);
\filldraw[thick,fill=white,draw=black] (0.8,1.2) ++(3,0) circle (.12);

\end{tikzpicture}

\caption{\label{fig:string} Graphical representation of the interpolating operator $\Psi_{\mathbf{s}}$ defined in eq.~\eqref{eq:cont_inter_string}. The black circle represents the electric charge, and the white circles are the image anti-charges. The lines with arrows represent the electric flux (i.e. the Wilson lines), which has to escape the box in a symmetric way through the two opposite planes because of the boundary conditions.
}
\end{figure}

Another choice that might look more convenient because of its explicit $\text{O}(4)$ covariance is given by
\begin{gather}
J_\mu(x) = \partial_\mu \Phi(x) \ , 
\qquad
\partial_\mu\partial_\mu \Phi(x) = \delta^4(x) \ ,
\end{gather}
where $\Phi(x)$ is anti-periodic in space and has appropriate boundary conditions in time. With this choice the operator $\Psi_J(x)$ can be written as
\begin{gather}
\Psi_{\bm \ell}(x) = e^{-\imath q \int d^4y\, \partial_\rho A_\rho(y)\, \Phi(y-x)}\ \psi(x) \ .
\end{gather}
In Landau gauge we get $\Psi_{\bm \ell}(x)=\psi(x)$, and the operator
$\Psi_{\bm \ell}(x)$ is the unique gauge-invariant extension of the
operator $\psi(x)$ defined in Landau gauge. Even though the Landau and
other covariant gauges are often used in perturbative calculations,
notice that the operator $\Psi_{\bm \ell}(x)$ is non-local in time and
interferes with the dynamics by effectively generating a
time-dependent contribution to the Hamiltonian. One can show that this
contribution vanishes at large time separations, and therefore the same masses
will be obtained, but in practical situations the asymptotic behavior
could be reached very slowly. These complications can be avoided
in the first place by sticking to a gauge-invariant formalism with the
local-in-time interpolating operators introduced before.

\section{Flavour symmetry in \CQCDQED{}}
\label{sec:qed+qcd}
QCD is coupled to QED in the standard way
\begin{gather}
  S[A,\psi] = \int_{L^3T}d^4x\, \left\{ 
    \frac{1}{4e^2}F_{\mu\nu}F_{\mu\nu} +
    \frac{1}{2g^2} \tr G_{\mu\nu}G_{\mu\nu} +
    \sum_{i=f}^{N_f}\overline
    \psi_f(\gamma_\mu \lrD{}^f_\mu + m_f)\psi_f 
  \right\} \;,
\end{gather}
where the chromo-magnetic field strength and the covariant derivative are
\begin{gather}
G_{\mu\nu}(x) = \partial_\mu B_\nu(x) - \partial_\nu B_\mu(x) - \imath
[B_\mu(x), B_\nu(x)] \ , 
\nonumber \\
\lrD{}_\mu^f =  \lrp{}_\mu-\imath q_f A_\mu - \imath B_\mu \ ,
\end{gather}
and $B_\mu(x)$ denotes the colour gauge field. $B_\mu(x)$ is defined to be a traceless hermitian $3\times 3$ matrix. Up-type  and down-type quarks have electric charge $q_f=2/3$ and $q_f=-1/3$ respectively. Since quark fields obey \Cstar{} boundary conditions, the colour gauge field must obey \Cstar{} boundary conditions as well in order to ensure periodicity of the action density,
\begin{gather}
B_\rho(x+\hat L_i) = - B_\rho(x)^* \;.
\end{gather}
Let us now focus on the violation of flavour and electric-charge
conservation, since they are 
substantially different from the case of \CQED{} alone. 

Since the elementary charge is $1/3$, from the discussion in
section~\ref{sec:cstar} it might seem that processes with a $\Delta
Q=\pm 2/3$ violation are allowed by the boundary conditions. However
fractional charges are confined in hadrons which have integer electric
charge. \textit{If the box size is large enough} only colourless
particles can travel around the torus, implying that charge violation
can be produced only in multiples of $\Delta Q= \pm 2$. Consequently a
proton state can mix with an antiproton state, or with a $p \pi^+
\pi^+$ state. 

One might wonder whether \Cstar{} boundary conditions can induce a
spurious mixing of the proton with some lighter state. This issue is
surely relevant if one wants to extract the proton properties from the
long-distance behaviour of two-point functions from lattice
simulations. It is also not entirely trivial, considering that
\Cstar{} boundary conditions produce a violation of the baryon-number
conservation. When a hadron travels around the torus its baryon number
changes sign, which in turn implies that baryon-number violation can
be produced only in multiples of $2$. A proton state
cannot mix with states with zero baryon number, i.e. with lighter
states. 

Both charge and baryon number are linear combinations of the
individual species numbers, which we refer to as flavour numbers, 
\begin{gather}
Q = \sum_f q_f F_f \ , \qquad 
B = \frac{1}{3} \sum_f F_f \;.
\end{gather}
Since each flavour-number conservation law is violated by the \Cstar{} boundary conditions, one might wonder for instance whether a pion state can mix with a kaon state. This is not the case, as individual flavour conservation can be violated again only in multiples of two,
\begin{gather}
\Delta F_f = 0\mod 2\;. 
\end{gather}
Also notice that $\Delta B$ being a multiple of 2 implies that
total-flavour $F = \sum_f F_f$ violation is produced only in multiples
of six,
\begin{gather}
\Delta F = 0 \mod 6\;. 
\end{gather}
For instance, if only strangeness conservation is violated in a given
process, this violation must be produced in multiples of $6$. If
strangeness violation amounts to a multiple of 2 which is 
not a multiple of 6, then it must be accompanied by violation in the
conservation of some other flavour. For example the $\Omega^-=sss$
will mix, via a $K^-=s\bar{u}$ traveling around the torus, 
with ${\Sigma}^++2\gamma$ where ${\Sigma}^+=suu$ and with other two particle states
like $\Lambda^0 \pi^+$.  
This process has $\Delta F_s=-2$ and $\Delta F_u=+2$. In particular
this implies that the 
$\Omega^-$ mass cannot be extracted from the long-distance behaviour
of a two-point function at finite volume. In order to extract the
$\Omega^-$ mass one has to take the infinite-volume limit of the
two-point function (or effective mass) first, and then extract the
long-distance behaviour. Similarly the $\Xi^-=ssd$ mixes with
the $p=uud$ via a $K^-=s\bar{u}$ traveling around the torus  (see
figure~\ref{fig:omega}).  
This process has again $\Delta F_s=-2$ and $\Delta F_u=+2$.

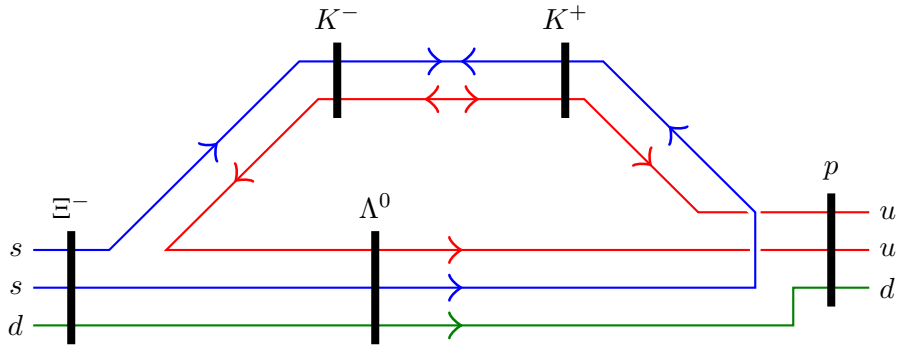
\begin{figure}
\centering

\begin{tikzpicture}[scale=.5]


\draw[red,thick] (22,1) coordinate (H) -- ++(-4.5,0) coordinate (I) -- ++(-3,3) coordinate (L) -- ++(-7,0) coordinate (M) -- ++(-4,-4) coordinate (N) -- (22,0) coordinate (O);

\draw[preaction={draw,line width=4,white},blue,thick] (0,0) coordinate (A) -- ++(2,0) coordinate (B) -- ++(5,5) coordinate (C) -- ++(8,0) coordinate (D) -- ++(4,-4) coordinate (E) -- ++(0,-2) coordinate (F) -- (0,-1) coordinate (G);

\draw[green!50!black,thick] (0,-2) coordinate (P) -- ++(20,0) coordinate (Q) -- ++(0,1) coordinate (R) -- (22,-1) coordinate (S);

\draw[decorate,decoration={markings,mark=at position .57 with {\arrow[blue,scale=3]{>}}}] (B) -- (C);
\draw[decorate,decoration={markings,mark=at position .47 with {\arrow[blue,scale=3]{>}},mark=at position .58 with {\arrow[blue,scale=3]{<}}}] (C) -- (D);
\draw[decorate,decoration={markings,mark=at position .5 with {\arrow[blue,scale=3]{<}}}] (D) -- (E);

\draw[decorate,decoration={markings,mark=at position .5 with {\arrow[red,scale=3]{<}}}] (I) -- (L);
\draw[decorate,decoration={markings,mark=at position .45 with {\arrow[red,scale=3]{<}},mark=at position .6 with {\arrow[red,scale=3]{>}}}] (L) -- (M);
\draw[decorate,decoration={markings,mark=at position .55 with {\arrow[red,scale=3]{>}}}] (M) -- (N);

\draw[decorate,decoration={markings,mark=at position .65 with {\arrow[red,scale=3]{>}}}] (10,0) -- ++(2,0);
\draw[decorate,decoration={markings,mark=at position .65 with {\arrow[blue,scale=3]{>}}}] (10,-1) -- ++(2,0);
\draw[decorate,decoration={markings,mark=at position .65 with {\arrow[green!50!black,scale=3]{>}}}] (10,-2) -- ++(2,0);

\draw[line width=3] (1,-2.5) -- (1,.5) node[above] {$\Xi^-$};
\draw[line width=3] (9,-2.5) -- (9,.5) node[above] {$\Lambda^0$};
\draw[line width=3] (8,3.5) -- (8,5.5) node[above] {$K^-$};
\draw[line width=3] (14,3.5) -- (14,5.5) node[above] {$K^+$};
\draw[line width=3] (21,-1.5) -- (21,1.5) node[above] {$p$};

\path (A) node[left] {$s$};
\path (G) node[left] {$s$};
\path (P) node[left] {$d$};

\path (H) node[right] {$u$};
\path (O) node[right] {$u$};
\path (S) node[right] {$d$};

\end{tikzpicture}

\caption{\label{fig:omega}
Schematic representation of a possible process responsible for the
$\Xi^-$/$p$ mixing. The process goes through a $u \bar{u}$ pair
creation. The colourless $K^-=s \bar{u}$ travels around the torus and
turns into a $K^+=\bar{s}u$. Finally an $s \bar{s}$ pair annihilates. 
}
\end{figure}

In QCD$_{\mathrm{C}}$ alone, flavour violation is an
exponentially-suppressed effect in the size of the box, like any other
finite volume correction. Adding electromagnetic interactions make
finite volume corrections generically inverse powers of $L$, due
to the massless photon. The detailed analysis of flavour violating
process in \CQCDQED~requires to keep track of the flavour numbers in
the process. This analysis, in the framework of an effective field
theory of hadrons, is carried out in detail in the
appendix~\ref{app:mixing}, but the main results that we prove in this
appendix can be easily explained. Flavour violating process in
\CQCDQED~cannot be mediated by the photon. A particle with the same
flavour numbers that are violated must travel around the torus, and
since only massive particles carry flavour in \CQCDQED, these effects
are exponentially suppressed. 

For example, in the case of the already-mentioned mixing between the
$\Xi^-$ and the proton, the one-loop diagram of figure~\ref{fig:omega}   
is of order $\exp(-m_K L)$. But the general case is much more
complicated, since the $\Xi^-$ can also mix with the proton and an
arbitrary number of photons, or with a neutron-$\pi^+$ state. As it is
proved in 
appendix~\ref{app:mixing}, flavour violating process in this case are
suppressed by a factor $\exp(-\mu L)$ with 
\begin{equation}
  \mu = \left[ M_{K^\pm}^2 - \left( \frac{M_{\Xi^-}^2 - M_{\Lambda^0}^2 +
        M_{K^\pm}^2}{2M_{\Xi^-}} \right)^2 \right]^{1/2}\,.
\end{equation}

Note that this effects are generically very suppressed, since 
the corresponding coefficient in the $\Xi^-$ two-point function is
proportional to the square of the transition 
amplitude, i.e. to $\exp(-2\mu L) \sim \mathcal O(10^{-10})$. A similar
analysis for the case of the mixing of the $\Omega^-$ results in an
amplitude suppressed by a factor $\mathcal O(10^{-8})$.

We close this section by remarking that the renormalization of \CQCDQED{} is not affected by electric charge and flavour breaking effects discussed in this section. Indeed these are induced by the boundary conditions and locality guarantees that the ultraviolet structure of the theory is independent of them. This applies both to the couplings of the Lagrangian and to the renormalization constants and mixing coefficients of any composite operator.


\section{Finite-volume effects on the masses of charged hadrons}
\label{sec:perturbative}

The finite-volume corrections to the mass of a stable hadron of
non-vanishing charge $q$, which is valid only at first order in $e^2$
and up to corrections in the size of the box that fall off faster than any power, can be
written as 
\begin{gather}
\frac{\Delta m(L)}{m}
= \frac{e^2}{4\pi} \left\{ \frac{q^2 \xi(1)}{2 m L}
+ \frac{q^2 \xi(2)}{\pi (m L)^2}
- \frac{1}{4 \pi m L^4}
\sum_{\ell =1}^\infty
\frac{(-1)^{\ell} (2\ell)!}{\ell! L^{2(\ell-1)}}\,
\mathcal{T}_{\ell}\, \xi(2+2\ell)
\right\}
\ +\ \dots \ ,
\label{eq:finalallspin}
\end{gather}
where $m$ is the particle mass in infinite volume, and
$m(L) = m + \Delta m(L)$ is the particle mass in finite
volume. Typical examples of stable hadrons to which this formula
applies are the proton, the neutron, the charged pions, the charged kaons, $D$ and $B$ mesons.

The derivation of eq.~\eqref{eq:finalallspin} is given in
appendix~\ref{app:massformula}. Here we discuss the structure of
eq.~\eqref{eq:finalallspin} that is in fact very simple.
$\mathcal{T}_{\ell}$ is the $\ell$-th derivative with respect to
$\vec{k}^2$ of the (infinite-volume) forward Compton amplitude
for the scattering of a photon with energy $|\vec{k}|$ on
the charged hadron at rest, in the limit $\vec{k} \to \vec{0}$.
The boundary conditions enter only in the definition of
the generalised zeta function
\begin{gather}
\xi(s) = \sum_{\vec{n} \neq \vec{0}} \frac{(-1)^{\langle \vec{n} \rangle}}{|\vec{n}|^s} \ .
\end{gather}
This formula is valid for real $s > 3$, while the values $s=1$ and $2$ are
obtained by analytic continuation. An explicit
representation of the $\xi(s)$ coefficients, which is valid for
all values of $s$ we are interested in, is given in
eq.~\eqref{eq:zeta-function-final}. The values of the first three coefficients
$\xi(s)$ are given in table~\ref{tab:xicoefficients} in the case
of \Cstar{} boundary conditions in 1, 2 or 3 spatial directions.

\begin{table}[tbp]
\centering
\begin{tabular}{lr@{$.$}lr@{$.$}lr@{$.$}l}
& \multicolumn{2}{c}{1\Cstar{}} & \multicolumn{2}{c}{2\Cstar{}} & \multicolumn{2}{c}{3\Cstar{}} \\[0.5ex]
\hline
\\[-1.5ex]
$\xi(1)$ & $-0$&$77438614142$ & $-1$&$4803898065$ & $-1$&$7475645946$ \\[0.5ex]
$\xi(2)$   & $-0$&$30138022444$ & $-1$&$8300453641$ & $-2$&$5193561521$ \\[0.5ex]
$\xi(4)$   &  $0$&$68922257439$ & $-2$&$1568872986$ & $-3$&$8631638072$ \\[1.0ex]
\hline
\end{tabular}
\caption{\label{tab:xicoefficients} Values of the first three coefficients $\xi(s)$ in the case of \Cstar{} boundary conditions in 1, 2 or 3 spatial directions and periodic boundary conditions in the others (columns 2,3 and 4 respectively).}
\end{table}

The  $1/L$ and $1/L^2$ terms are universal, i.e. they depend only on
the mass and charge of the hadron, and not on its spin and internal
structure. The dependence upon spin and internal structure is encoded in
the coefficients $\mathcal{T}_{\ell}$ and is suppressed with respect to the universal
part, as it contributes at $\mathcal{O} (1/L^4)$. No inverse odd
power of $L$ appears in the expansion, other than the leading $1/L$
point-like contribution.

A formula very similar to eq.~\eqref{eq:finalallspin} has been derived
in refs.~\cite{Borsanyi:2014jba,Davoudi:2014qua,Fodor:2015pna} in the case of
\LQED{}, i.e. the theory with the quenched spatial zero-modes of the
electromagnetic field. According to
refs.~\cite{Borsanyi:2014jba,Davoudi:2014qua,Fodor:2015pna} the $1/L$ and $1/L^2$
terms are universal also in \LQED{}, while spin and
structure-dependent terms contribute at $\mathcal{O}
(1/L^3)$.

\begin{figure}
\centering

\includegraphics[width=.9\textwidth]{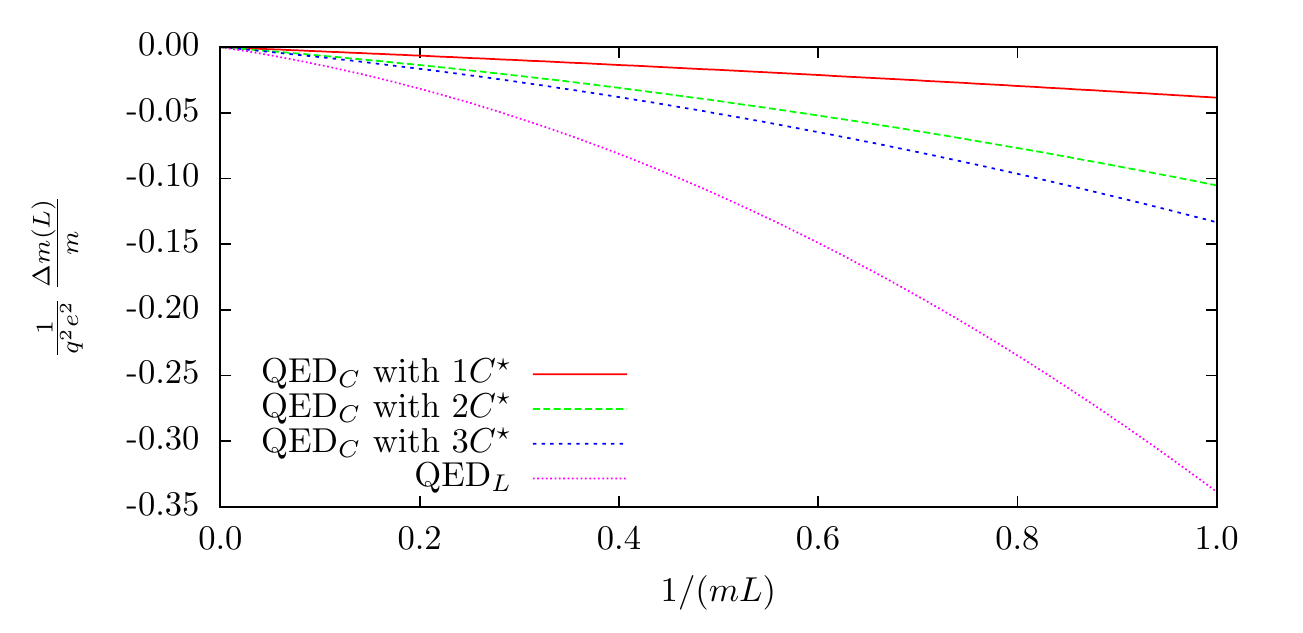}

\caption{Leading finite-volume corrections to the mass of a stable particle of charge $qe$ in \CQED{} and \LQED{}. The plot shows the universal $1/L$ and $1/L^2$ contributions. Spin and structure-dependent contributions are $O\left(1/L^3\right)$ in \LQED{}, and $O\left(1/L^4\right)$ in \CQED{}.}
\label{fig:fvecomparison}
\end{figure}

It is remarkable that, because of the locality of \CQED{}, spin and
structure-dependent contributions are much more suppressed with
respect to \LQED{}. Moreover, also the universal $1/L$ and $1/L^2$
contributions are considerably smaller in \CQED{} with respect to
\LQED{}. This is shown in figure~\ref{fig:fvecomparison} where the
results of refs.~\cite{Borsanyi:2014jba,Davoudi:2014qua} are compared
with the $1/L$ and $1/L^2$ terms of eq.~\eqref{eq:finalallspin}.

\section{Lattice formulation}
\label{sec:lattice}
In the context of lattice non-compact QED, the implementation of
\Cstar{} boundary conditions and of the proposed interpolating
operators is straightforward. One can extract the leading order
$\mathcal O(e^2)$ electromagnetic contributions either with the
techniques described in~\cite{deDivitiis:2013xla}, or by doing a 
QED dynamical simulation as suggested in
refs.~\cite{Borsanyi:2014jba,deDivitiis:2013xla,Basak:2014vca,Ishikawa:2012ix,Aoki:2012st,Blum:2010ym}. 
In the non-compact formulation a way to damp the longitudinal 
modes of the gauge field is needed, gauge-fixing being the most
common choice.  Here we will focus on the compact formulation of the
theory in a manifest gauge-invariant way.

In the compact formulation on the lattice, the gauge field is replaced
by the link variable $U(x,\mu)$ which lives in the gauge group $\U(1)$
and satisfies the boundary conditions 
\begin{gather}
U(x+\hat{L}_k,\rho) = U(x,\rho)^* \ ,
\end{gather}
along the spatial directions, and generic boundary conditions
(i.e. periodic, SF, open, open-SF) along the temporal direction. For
the moment we focus on the simpler case of \CQED{} coupled to a single
fermion field with unitary charge. The generalization to \CQCDQED{}
will be discussed at the end of this section. 

We want to argue now that, in order to be able to discretize the
interpolating operators proposed in section~\ref{sec:interpolating} in
a completely gauge-invariant fashion, we need a rather unconventional
action for compact \CQED{}. Notice that in the standard formulation of
compact QED, the perturbative series is generated by identifying
$U(x,\mu) = e^{\imath A_\mu(x)}$ and by expanding in powers of the
gauge field. The discretization of the interpolating
operator~\eqref{eq:cont_inter_string} would need to take the square
root of the link variable. This operation is not gauge covariant, and
should be avoided while aiming at a completely gauge-invariant
formulation. The root of this complication lies in the fact that,
because of the boundary conditions, the electric flux generated by a
single charge must escape the box in a symmetric way through the
$x_k=0$ and $x_k=L$ planes. The dynamical unit charge generated by the
interpolating operator of eq.~\eqref{eq:cont_inter_string} is located
in $x$ and sees effectively two image half charges located in
$x+\hat{L}_k$ and $x-\hat{L}_k$. 

As it will be clear by the end of this section, this issue is completely
removed by choosing the following action for compact \CQED{} with a
single matter field, 
\begin{gather}
S = S_\gamma + S_m \ , 
\nonumber \\
S_\gamma = \frac{2}{e^2} \sum_x \sum_{\mu\nu} [1 - P(x,\mu,\nu)] \ ,
\nonumber \\
S_m = \sum_x \bar{\psi}(x) D[U^2] \psi(x) \ .
\label{eq:compact:action}
\end{gather}
The plaquette is defined as usual,
\begin{gather}
P(x,\mu,\nu) = U(x,\mu)U(x+\hat{\mu},\nu)U(x+\hat{\nu},\mu)^{-1}U(x,\nu)^{-1} \ ,
\end{gather}
while the Wilson-Dirac operator has an unconventional coupling to the gauge field,
\begin{gather}
D[U^2] = m +  \frac{1}{2}\sum_{\mu=0}^3 \left\{ 
\gamma_\mu (\nabla^*_\mu[U^2] + \nabla_\mu[U^2]) - \nabla^*_\mu[U^2] \nabla_\mu[U^2] 
\right\} \ ,
\nonumber \\
\nabla_\mu[U^2] \psi(x) = U(x,\mu)^2 \psi(x+\hat{\mu}) - \psi(x) \ , 
\nonumber \\
\nabla^*_\mu[U^2] \psi(x) = \psi(x) - U(x-\hat{\mu},\mu)^{-2} \psi(x-\hat{\mu}) \ .
\label{eq:compact:dirac}
\end{gather}
Any other discretization of the Dirac operator (preserving charge conjugation) can be employed as well.

The proposed action is invariant under local gauge transformations of the form
\begin{gather}
U(x,\mu) \to \Lambda(x) U(x,\mu) \Lambda(x+\hat{\mu})^{-1} \ , 
\nonumber \\
\psi(x) \to \Lambda(x)^2 \psi(x) \ , 
\nonumber \\
\bar{\psi}(x) \to \bar{\psi}(x) \Lambda(x)^{-2} \ ,
\label{eq:compact:gauge}
\end{gather}
where $\Lambda(x) \in \U(1)$ satisfies boundary conditions
\begin{gather}
\Lambda(x+\hat{L}_k) = \Lambda(x)^* \ .
\end{gather}
The action possesses also a $\mathbb{Z}_2^4$ center symmetry. For each direction $\mu$ one can flip the sign of all link variables in the direction $\mu$ on the three-dimensional slice defined by $x_\mu=0$ without changing the value of the action. Before discussing the interpolating operators, we want to show that the action~\eqref{eq:compact:action} is perturbatively equivalent to the usual QED action in the continuum limit.

In order to set up a perturbative expansion, we need to identify the minima of the action at $\mathcal{O}(e^0)$. These are given by all configurations with $P(x,\mu,\nu)=1$. In appendix~\ref{app:compact} we show that there is a discrete set of gauge-inequivalent minima labeled by the elements of the set
\begin{gather}
\Omega = \{ (z_0,z_1,z_2,1) \ |\ z_0^2=z_1^2=z_2^2=1 \} \ .
\end{gather}
Given a minimum of the action at $\mathcal{O}(e^0)$, it is always possible to find a vector $z \in \Omega$ such that the chosen minimum is gauge-equivalent to the following gauge field
\begin{gather}
\bar{U}_z(x,\mu) = \begin{cases}
z_\mu \qquad & \text{if } x_\mu=L_\mu-1 \ , \\
1 & \text{otherwise} \ .
\end{cases}
\end{gather}
Because of center symmetry one might expect also minima with $z_3=-1$. However each minimum with $z_3=-1$ is gauge-equivalent to some minimum with $z_3=1$ (this is a byproduct of the construction given in appendix~\ref{app:compact}).

The perturbative expansion around the minimum $\bar{U}_z(x,\mu)$ is set up by defining
\begin{gather}
U(x,\mu) = \bar{U}_z(x,\mu) e^{\frac{\imath}{2} A_\mu(x)} \ ,
\label{eq:link_variable}
\end{gather}
and by adding a gauge-fixing term $S_{\text{gf}}$ to the action, which we will not do explicitly. We only observe that $S_{\text{gf}}$ is a function of the fluctuation $A_\mu(x)$ only, and not of the classical vacuum $\bar{U}_z(x,\mu)$. Given a generic functional $\mathcal{F}[U,\psi,\bar{\psi}]$ of the fields, the perturbative expansion to some order $\mathcal{O}(e^n)$ of its expectation value is given by
\begin{gather}
\langle \mathcal{F}[U,\psi,\bar{\psi}] \rangle = \frac{1}{Z} \sum_{z \in \Omega} \int \mathcal{D}A \mathcal{D}\bar{\psi} \mathcal{D}\psi \ 
\mathcal{F} [\bar{U}_z e^{\frac{\imath}{2} A}, \psi,\bar{\psi} ]\ e^{-S[e^{\frac{\imath}{2} A}, \psi,\bar{\psi}] - S_{\text{gf}}[A]} + \mathcal{O}(e^n) \ ,
\end{gather}
where we have used that the action is center-invariant and therefore it does not depend on $\bar{U}_z(x,\mu)$ once the substitution~\eqref{eq:link_variable} is used. The normalization $Z$ is given by
\begin{gather}
Z = 8 \int \mathcal{D}A \mathcal{D}\bar{\psi} \mathcal{D}\psi \ 
e^{-S[e^{\frac{\imath}{2} A}, \psi,\bar{\psi}] - S_{\text{gf}}[A]} \  \left\{ 1 + \mathcal{O}(e^n) \right\} \ .
\end{gather}
If the observable $\mathcal{F}$ is charged under center symmetry, its expectation value vanishes. On the other hand, center-invariant observables get mapped naturally into corresponding observables in the non-compact setup. In fact, if $\mathcal{F}$ is invariant under center symmetry, then the $z$ dependence drops out of the path integral and the standard perturbation expansion about $\bar{U}=1$ is recovered,
\begin{gather}
\langle \mathcal{F}[U,\psi,\bar{\psi}] \rangle = \frac{8}{Z} \int \mathcal{D}A \mathcal{D}\bar{\psi} \mathcal{D}\psi \ 
\mathcal{F}[e^{\frac{\imath}{2} A}, \psi,\bar{\psi}]\ e^{-S[e^{\frac{\imath}{2} A}, \psi,\bar{\psi}] - S_{\text{gf}}[A]} + \mathcal{O}(e^n) \ ,
\end{gather}
provided that $S[e^{\frac{\imath}{2} A}, \psi,\bar{\psi}]$ is the standard QED action up to irrelevant operators.

This can be verified by replacing the definition~\eqref{eq:link_variable} into the action and by expanding in powers of the fields. The $1/2$ factor in the exponent of~\eqref{eq:link_variable} combines with the unconventional normalization of the gauge action in eqs.~\eqref{eq:compact:action} in such a way that the canonical normalization of the gauge field is restored,
\begin{gather}
P(x,\mu,\nu) = 1 + \frac{\imath}{2} F_{\mu\nu}(x) - \frac{1}{8} F^2_{\mu\nu}(x) + \dots \ , 
\nonumber \\
S_\gamma = \frac{2}{e^2} \sum_x \sum_{\mu\nu} [1 - P(x,\mu,\nu)] = \frac{1}{4 e^2} \sum_x \sum_{\mu\nu} F^2_{\mu\nu}(x) + \text{irrelevant operators} \ .
\end{gather}
Also the same $1/2$ factor in the exponent of~\eqref{eq:link_variable} combines with the second power of the link variable in the Dirac operator~\eqref{eq:compact:dirac} in such a way that the correct coupling of the electron to the gauge field is restored,
\begin{gather}
U(x,\mu)^2 = 1 + \imath A_\mu(x) + \dots \ , 
\nonumber \\
S_m = \sum_x \bar{\psi}(x) \left\{ \gamma_\mu \left[ \frac{ \partial_\mu + \partial_\mu^* }{2} + \imath A_\mu(x) \right] + m \right\} \psi(x) + 
\text{irrelevant operators} \ .
\end{gather}
The elementary charge (charge is quantized in compact QED) interacts with the gauge field with strength $1/2$. However the dynamical fermion has an electric charge that is twice the elementary charge, which generates a coupling of strength $1$ to the gauge field. This structure is also reflected by the gauge transformations~\eqref{eq:compact:gauge}.

In the proposed setup, thanks to the identification~\eqref{eq:link_variable}, the interpolating operator~\eqref{eq:cont_inter_string} can be discretized in a straightforward fashion,
\begin{gather}
\Psi_{\mathbf s}(x) =
\prod_{s=-x_k}^{-1} U(x+s \hat{k},k)^{-1}\
\psi(x)\
\prod_{s=0}^{L-x_k-1} U(x+s \hat{k},k) \ .
\label{eq:lat_inter_string}
\end{gather}
Notice that the the above operator is charged under center symmetry. However in practice only the product $\Psi_{\mathbf s}(x) \bar{\Psi}_{\mathbf s}(y)$ is relevant, which is center invariant.

For completeness we present also a possible discretization of the operator~\eqref{eq:cont_inter_coulomb}. We introduce the field
\begin{gather}
A^{\mathbf c}_\mu(x) = \Delta^{-1} \partial_k^* \hat{F}_{k \mu}(x) \ ,
\end{gather}
where $\Delta = \partial_k \partial_k^*$ is the three-dimensional discrete Laplace operator defined with anti-periodic boundary conditions, and $\hat{F}_{\rho\sigma}(x)$ is some discretization of the field tensor (e.g. the clover plaquette). It is straightforward to verify that $A^{\mathbf c}_\mu(x)$ satisfies the discrete Coulomb constraint $\partial_k^* \, A^{\mathbf c}_k(x) = 0$. In the continuum limit $A^{\mathbf c}_\mu(x)$ is nothing but the gauge field in Coulomb gauge.\footnote{
In the continuum limit:
\begin{gather}
\partial_k A^{\mathbf c}_k(x) = \Delta^{-1} \partial_k \partial_j \hat{F}_{k j}(x) = 0 \ , \nonumber \\
A^{\mathbf c}_\mu(x) = \Delta^{-1} \partial_k \{ \partial_k A_\mu(x) - \partial_\mu A_k(x) \} = A_\mu(x) - \partial_\mu \left\{ \Delta^{-1} \partial_k A_k(x) \right\} \ , \nonumber
\end{gather}
i.e. $A^{\mathbf c}_\mu(x)$ is gauge-equivalent to $A_\mu(x)$ and satisfies the Coulomb-gauge contraint.
} The operator~\eqref{eq:cont_inter_coulomb} can be discretized by using the relation
\begin{gather}
\Psi_{\mathbf c}(x) = \Psi_{\mathbf s}(x)\ e^{-\frac{\imath}{2} \sum_{s=0}^{L} A^{\mathbf c}_k(x + s \hat{k})} \ ,
\label{eq:inter_relation}
\end{gather}
which is exact in the continuum limit, and easily verified in Coulomb gauge. In fact in Coulomb gauge and in the continuum eq.~\eqref{eq:inter_relation} is completely equivalent to eq.~\eqref{eq:cont_inter_string}, given the relations $A_\mu(x)=A^{\mathbf c}_\mu(x)$ and $\Psi_{\mathbf c}(x)=\psi(x)$.

The generalization of the proposed strategy to the case of compact \CQCDQED{} is straightforward. We need to introduce the link variables $V(x,\mu) \in \SU(3)$ for the colour field with the boundary conditions
\begin{gather}
V(x+\hat{L}_k,\rho) = V(x,\rho)^* \ ,
\end{gather}
and the corresponding plaquette:
\begin{gather}
Q(x,\mu,\nu) = V(x,\mu)V(x+\hat{\mu},\nu)V(x+\hat{\nu},\mu)^{-1}V(x,\nu)^{-1} \ .
\end{gather}
For sake of simplicity we choose the standard Wilson action for the colour field. The photon action requires a further rescaling, since quarks have fractional charge,
\begin{gather}
S = S_g + S_\gamma + S_m \ , 
\nonumber \\
S_g = \frac{1}{g^2} \sum_x \sum_{\mu\nu} \tr [1 - Q(x,\mu,\nu)] \ ,
\nonumber \\
S_\gamma = \frac{18}{e^2} \sum_x \sum_{\mu\nu} [1 - P(x,\mu,\nu)] \ ,
\nonumber \\
S_m = \sum_f \sum_x \bar{\psi}_f(x) D_f[U,V] \psi_f(x) \ .
\label{eq:compact:action:qcd+qed}
\end{gather}
Moreover the Dirac operator has to implement the correct coupling of up-type ($q_f=2/3$) and down-type ($q_f=-1/3$) quarks to the electromagnetic field,
\begin{gather}
D_f[U,V] = m_f +  \frac{1}{2}\sum_{\mu=0}^3 \left\{ 
\gamma_\mu (\nabla^*_\mu[U^{6q_f} V] + \nabla_\mu[U^{6q_f} V]) - \nabla^*_\mu[U^{6q_f} V] \nabla_\mu[U^{6q_f} V] 
\right\} \ .
\label{eq:compact:dirac:qcd+qed}
\end{gather}

We remind that QCD with non-degenerate Wilson-Dirac quarks (with or without QED) has a mild sign problem, i.e. the fermionic determinant is positive in the continuum limit but can be negative because of lattice artefacts. In appendix~\ref{app:sign} we show that \Cstar{} boundary conditions do not make this sign problem worse.

\section{Conclusions}
\label{sec:conclusions}
A local solution to the problem of electrically charged particles in a
finite volume was proposed
in~\cite{Polley:1990tf,Kronfeld:1990qu,Wiese:1991ku,Kronfeld:1992ae},
and it is based on \Cstar{} boundary conditions for all fields along
one or more spatial directions. Because of the boundary conditions
Gauss's law does not prevent the propagation of charged particles on a
finite volume (as opposed to the case of periodic boundary
conditions). We have analyzed in detail the properties of QED in isolation
and of QED coupled to QCD with \Cstar{} boundary conditions (\CQED{}
and \CQCDQED{} respectively), and we have discussed how this setup can be
used in spectroscopy calculations. 

We have devoted part of the paper to construct interpolating operators
that have the quantum numbers of charged particles and that are also
invariant under local gauge transformations. These can be used to
probe the physical sector of the Hilbert space of the theory with
non-perturbative accuracy without having to rely on gauge-fixing at
intermediate stages of calculation. To this end we have discussed the
details of the implementation of the proposed interpolating operators
in the compact lattice formulation of the theory. 

We have discussed the symmetries of \CQED{} and \CQCDQED{} in
depth. In particular we signal that \Cstar{} boundary conditions
violate flavour and electric-charge conservation partially, in such a
way that this does not represent a limitation to the use of \Cstar{}
boundary conditions in most of the relevant applications. Even though
finite-volume effects vanish generally like some inverse power of the
box size because of the photon, we have shown that flavour and
electric-charge violations are exponentially suppressed in the box
size. 

We have calculated the finite-volume corrections to the masses of
charged particles with \Cstar{} boundary conditions at
$\mathcal{O}(\alpha_{em})$. We have shown that the leading $1/L$ and
$1/L^2$ finite-volume corrections are universal, i.e. they depend on
neither spin nor internal structure. 
Similar results have been previously obtained in the non-local
formulation \LQED{}. When compared with these previous result, the
finite-volume corrections of \CQED{} are found to be significantly
smaller. In particular, the non-universal spin and structure-dependent
corrections are $\mathcal{O}(1/L^3)$ in \LQED{} and
$\mathcal{O}(1/L^4)$ in \CQED{}. We have also shown 
that these non-universal terms are related with
physical quantities, namely the derivatives of the forward Compton
scattering amplitudes.

\acknowledgments

We warmly thank Martin L\"uscher for illuminating discussions and constant encouragement at all stages of this work. We profited from discussions with Luigi Del Debbio, Liam Keegan, Laurent Lellouch, Marina Marinkovi\'c, Antonin Portelli, Kalman Szab\'o and Peter Weisz. N.T. thanks his colleagues of the RM123 collaboration for discussions related to the subjects covered in this work. We thank R.C. for the many convivial moments.

\appendix

\section{Exponential suppression of flavour mixing}
\label{app:mixing}
Let $\Xi(x)$ be some interpolating operator for some fixed spin-component of the negatively charged $\Xi^-=ssd$ and $\Xi_+(x)$ its \C{}-even component. We consider the finite-volume Minkowskian retarded two-point function at zero momentum, its spectral decomposition and the dispersion relation:
\begin{flalign}
C(E;L) = & \, \imath \int_{\mathbb{R} \times L^3} d^4 x \ \theta(x_0) e^{\imath E x_0} \langle \Xi_+(x)^\dag \Xi_+(0) \rangle =
\int_0^\infty d\mu \ \frac{\rho(\mu;L)}{\mu - E - \imath \epsilon} \ ,
\label{appW:C0} \\
\rho(E;L) = & \frac{1}{\pi} \Im C(E;L) \ .
\label{appW:DR0}
\end{flalign}
In infinite volume the spectral density vanishes for $E<M_{\Xi^-}$. In a finite box, because of \Cstar{} boundary conditions, the lowest state contributing to the two-point function is a proton state (via a strangeness-violating process) and the spectral density vanishes only for $E<M_p(L)$. We want to show that the spectral density vanishes exponentially with the volume for energies lower than $M_{\Xi^-}$.

More precisely we choose some energy $E<M_{\Xi^-}$ and a smooth test function $\phi_{E}(\mu)$ which vanishes for $\mu>E$. Then we want to show that, in the $L \to \infty$ limit
\begin{gather}
\ln \int_0^\infty d\mu \ \rho(\mu;L) \phi_{E}(\mu) \le - 2 L \mathcal{M}(E) + \mathcal{O}(\ln L) \ ,
\label{appW:finalbound}
\end{gather}
where the mass that controls the exponential decay is a decreasing function of $E<M_{\Xi^-}$ and
\begin{gather}
\mathcal{M}(E) \ge \mathcal{M}(M_{\Xi^-}) = \left[ M_{K^\pm}^2 - \left( \frac{M_{\Xi^-}^2 - M_{\Lambda^0}^2 + M_{K^\pm}^2}{2M_{\Xi^-}} \right)^2 \right]^{1/2} \ .
\end{gather}
Notice that the spectral density at finite volume is a sum of delta functions localised on the eigenvalues of the Hamiltonian in the given channel,\footnote{At fixed order in perturbation theory, the spectral density is a sum of delta functions and their derivatives, localized on the eigenvalues of the free Hamiltonian} which is the reason why we need to consider a test function in order to write a precise statement. Before proceeding we comment on the fact that the analysis presented in this appendix can be easily extended to other channels (e.g. to the mixing of the $\Omega^-$ with lightest states).

We assume that the leading finite-volume corrections in the two-point function $C(E;L)$ are described by some arbitrarily-complicated Lagrangian field theory with small couplings, which effectively describe the dynamics of hadrons and photons at large distance in the framework of a perturbative expansion (after gauge-fixing for the photon). In order to avoid IR divergences at any stage of our calculation, we assume that the photon is massive. We will find that the mass $\mathcal{M}(E)$ does not depend on the mass of the photon. Each infinite-volume stable particle is described by an elementary field in the effective Lagrangian. Because of locality the finite-volume theory is described by the infinite-volume Lagrangian density. In particular vertices conserve flavour. Fields are assumed to have definite flavour numbers, and fields with all flavour numbers equal to zero are assumed to have definite \C{}-parity.

We think of the two-point function $C(E;L)$ order by order in perturbation theory as a sum of Feynman diagrams. According to Cutkosky's rules, a Feynman diagram contributes to the spectral density $\rho(E;L)$ via the dispersion relation \eqref{appW:DR0} only if a cut between its external vertices exists such that the sum of the masses of the cut propagators is smaller than $E$. Notice that all states propagating between the two external vertices must have $(-1)^B = -1$ and $(-1)^S = 1$ where $B$ is the baryon number and $S$ is the strangeness number. It is easy to check that (at the physical masses) states with such quantum numbers and with energy lower that $M_{\Xi^-}$ can contain no strange particle. Therefore a Feynman diagram for the two-point function contributes to the spectral density for $E<M_{\Xi^-}$ only if a cut exists between the external vertices such that no strange particle propagates through the cut. The set of these diagrams is denoted by $\mathcal{D}_s$ (see fig.~\ref{fig:appW:diagram} for an example). The spectral density for $E<M_{\Xi^-}$ can be represented as
\begin{gather}
\rho(E;L) = \frac{1}{\pi} \Im \sum_{\mathcal{G} \in \mathcal{D}_s} 
\imath \left\{ \prod_{a \in V(\mathcal{G})-\xi_0} \int_{\mathbb{R} \times L^3} d^4 x(a) \right\} \ \theta(x_0(\xi_1)) e^{\imath E x_0(\xi_1)} F^M_{\mathcal{G}}(x)|_{x(\xi_0)=0} \ ,
\label{appW:DR1}
\end{gather}
where $\xi_1$ and $\xi_0$ are the two external vertices of the diagram $\mathcal{G}$, and $V(\mathcal{G})$ is the set of all vertices. $F^M_{\mathcal{G}}(x)$ is a function of the coordinates of all vertices of the graph, and it is given by a product of propagators in coordinate space, various derivatives of propagators and numerical coefficients. We will refer to the function $F^M_{\mathcal{G}}(x)$ as \textit{Feynman integrand}.

\begin{figure}
\centering

\includegraphics[width=.7\textwidth]{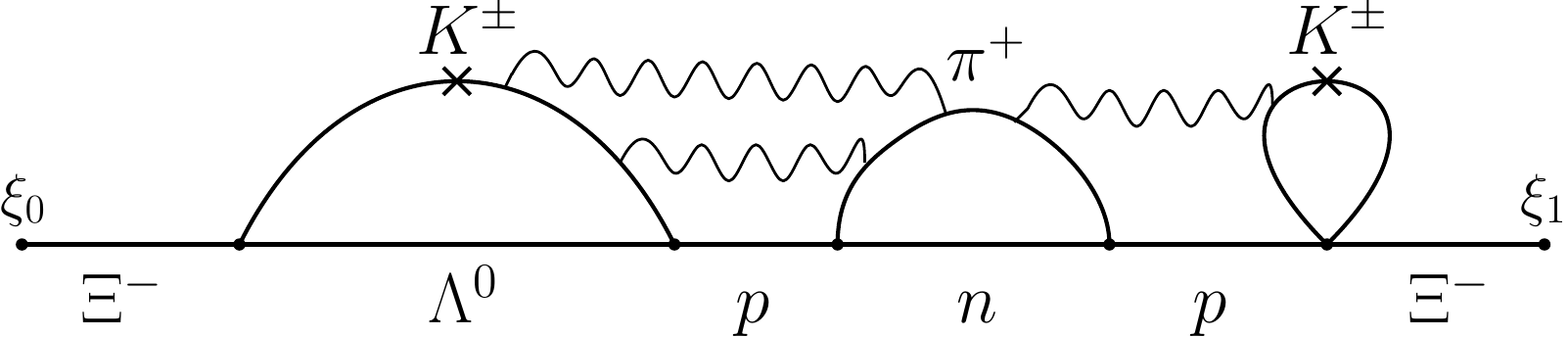}

\caption{\label{fig:appW:diagram}
Example of a Feynman diagram in $\mathcal{D}_s$. This is a contribution to the $\langle \Xi_+(x)^\dag \Xi_+(0) \rangle$ two-point function. Propagators with a cross are of the flavour-violating type, and they exist only because of the \Cstar{}-boundary conditions. Either the proton or the pair $n + \pi^+$ (in both cases with an arbitrary number of photons) can go on-shell with an energy lower than $M_{\Xi^-}$.
}
\end{figure}

\paragraph{Structure of the Feynman integrand.} We denote by $L(\mathcal{G})$ the set of lines of the diagram $\mathcal{G}$. Each line $\ell$ originates from the vertex $i(\ell)$ and terminates in the vertex $f(\ell)$ (line orientation is chosen arbitrarily). Each line $\ell$ corresponds to the Wick contraction of two fields located in $i(\ell)$ and $f(\ell)$. Let $F(\ell,i(\ell))$ and $F(\ell,f(\ell))$ be the vectors that contain all flavour numbers of these two fields, $F=(U,D,S,\dots)$. It is convenient to define $F(\ell,a)=0$ if the line $\ell$ is not attached to the vertex $a$. Since a field can be Wick-contracted either with itself or with its charge conjugate, each line is uniquely associated to a mass $m_\ell$. The line $\ell$ corresponds to a propagator
\begin{gather}
\Delta^M_{C_\ell}(\delta_\ell x;m_\ell,L) = \sum_{\vec{n}_\ell \in \mathbb{Z}^3} C^{\mathcal{G}}_\ell(\vec{n}_\ell) \ \Delta_M(\delta_\ell x_0, \delta_\ell \vec{x} + L n_\ell;m_\ell) \ ,
\label{appW:propagators}
\end{gather}
where $\Delta_M(x,m)$ is the infinite-volume Minkowkian propagator, $\delta_\ell x = x(f(\ell))-x(i(\ell))$, $n_\ell = (0,\vec{n}_\ell)$, and $C^{\mathcal{G}}_\ell(\vec{n}_\ell)$ is a function defined by one of the following possibilities:
\begin{description}[noitemsep,topsep=-.7em]
\item[Type 1.] Line $\ell$ arises from the Wick contraction of two flavourless \C{}-even fields,
\begin{gather}
F(\ell,i(\ell))=F(\ell,f(\ell))=0 \ , \qquad C^{\mathcal{G}}_\ell(\vec{n}_\ell) = 1 \ .
\end{gather}
\item[Type 2.] Line $\ell$ arises from the Wick contraction of two flavourless \C{}-odd fields,
\begin{gather}
F(\ell,i(\ell))=F(\ell,f(\ell))=0 \ , \qquad C^{\mathcal{G}}_\ell(\vec{n}_\ell) = (-1)^{\langle \vec{n}_\ell \rangle} \ .
\end{gather}
\item[Type 3.] Line $\ell$ arises from the Wick contraction of two flavourful fields with opposite flavour,
\begin{gather}
F(\ell,i(\ell))=-F(\ell,f(\ell)) \neq 0 \ , \qquad C^{\mathcal{G}}_\ell(\vec{n}_\ell) = \delta_{\langle \vec{n}_\ell \rangle,0} \ .
\end{gather}
\item[Type 4.] Line $\ell$ arises from the Wick contraction of two flavourful fields with same flavour,
\begin{gather}
F(\ell,i(\ell))=F(\ell,f(\ell)) \neq 0 \ , \qquad C^{\mathcal{G}}_\ell(\vec{n}_\ell) = \delta_{\langle \vec{n}_\ell \rangle,1} \ .
\end{gather}
\end{description}
Notice that conservation of flavour at the internal vertices is expressed by the equation
\begin{gather}
\sum_\ell F(\ell,a) = 0 \ , \quad \text{for } a \in V(\mathcal{G})-\{\xi_0,\xi_1\} \ .
\label{appW:flavourconservation}
\end{gather}

The function $F^M_{\mathcal{G}}(x)$ has the following structure
\begin{gather}
F^M_{\mathcal{G}}(x) = \sum_{\vec{n}} F^M_{\mathcal{G}}(x,\vec{n}) \ , \\
F^M_{\mathcal{G}}(x,\vec{n}) =
\mathbb{V}^M_{\mathcal{G}}
\prod_{\ell \in L(\mathcal{G})} C^{\mathcal{G}}_\ell(\vec{n}_\ell) \ \Delta_M(\delta_\ell x +  L n_\ell;m_\ell) \ ,
\end{gather}
where $\mathbb{V}^M_{\mathcal{G}}$ is a differential operator acting an all the coordinates $x(a)$ (in fact derivatives might be inserted in between propagators), and includes also couplings and combinatorial factors.
The detailed structure of $\mathbb{V}^M_{\mathcal{G}}$ is of no interest for the current discussion.

\paragraph{Boundary conditions for the Feynman integrand.} \Cstar{}-boundary conditions for the fields imply some peculiar boundary conditions for the Feynman integrand. Shifting the coordinate of a single vertex by
\begin{gather}
x(a) \to x(a) + \hat{L}_i
\end{gather}
is equivalent to replacing the operator inserted at the vertex $a$ with its charge-conjugate (the propagators have to be modified accordingly). The diagram obtained by this procedure is denoted by $c_a(\mathcal{G})$. Since the interaction Lagrangian and the considered interpolating operators are invariant under charge-conjugation, $\mathcal{G}$ is a diagram contributing to the two-point function if and only if $c_a(\mathcal{G})$ is a diagram contributing to the two-point function. It is also easy to check that the class $\mathcal{D}_s$ is closed under the action of $c_a$. By iterating the action of $c_a$, and by noticing that $c_a^2$ is the identity we get
\begin{gather}
F^M_{\mathcal{G}}(x_0,\vec{x} + L \bm{\lambda}) = F^M_{c^{\bm{\lambda}}(\mathcal{G})}(x_0,\vec{x}) \ ,  \quad \bm{\lambda}(a) \in \mathbb{Z}^3 \ ,
\label{appW:BC1}
\\
c^{\bm{\lambda}}(\mathcal{G}) = \prod_{a \in V(\mathcal{G})} c_a^{\langle \bm{\lambda}(a) \rangle}(\mathcal{G}) \ .
\end{gather}
Under the action of $c_a$, the vertex operator does not change. Propagators of type 1 attached to the vertex $a$ (at only one of the endpoints) are invariant, propagators of type 2 flip sign, propagators of type 3 are replaced with propagators of type 4 and \textit{vice versa}. In formulae this is equivalent to:
\begin{gather}
\mathbb{V}^M_{c^{\bm{\lambda}}(\mathcal{G})} = \mathbb{V}^M_{\mathcal{G}} \ , \\
C^{c^{\bm{\lambda}}(\mathcal{G})}_\ell(\vec{n}_\ell) = C^{\mathcal{G}}_\ell(\vec{n}_\ell - \delta_\ell \bm{\lambda}) \ .
\end{gather}
The boundary conditions for the function $F^M_{\mathcal{G}}(x,\vec{n})$ follow
\begin{gather}
F^M_{\mathcal{G}}(x_0,\vec{x} + L \bm{\lambda},\vec{n}) = F^M_{c^{\bm{\lambda}}(\mathcal{G})}(x_0,\vec{x},\vec{n} + \delta \bm{\lambda}) \ .
\label{appW:BC2}
\end{gather}

We can use the above boundary conditions to restrict the sum over the pairs $(\mathcal{G},\vec{n})$ and simultaneously extend the coordinate integration range to the whole $\mathbb{R}^4$, by means of a construction that has been discussed already in~\cite{Luscher:1985dn}. We will say that the two pairs $(\mathcal{G},\vec{n})$ and $(\mathcal{G}',\vec{n}')$ are \textit{gauge-equivalent} if and only if $\bm{\lambda}(a) \in \mathbb{Z}^3$ with $a \in V(\mathcal{G})$ exists such that:
\begin{gather}
\mathcal{G}' = c^{\bm{\lambda}}(\mathcal{G}) \ , \\
\vec{n}'_\ell = \vec{n}_\ell + \delta_\ell \bm{\lambda} = \vec{n}_\ell + \bm{\lambda}(f(\ell)) - \bm{\lambda}(i(\ell)) \ .
\end{gather}
The field $\vec{n}$ defined on lines is referred to as \textit{gauge field}, and the field $\bm{\lambda}$ defined on vertices is referred to as \textit{gauge transformation}. The set of all possible pairs $(\mathcal{G},\vec{n})$ splits into equivalence classes denoted by $[(\mathcal{G},\vec{n})]$. As shown in~\cite{Luscher:1985dn}, given two equivalent gauge fields the gauge transformation that relates the two of them is unique up to a global gauge transformation. The sum over the pairs $(\mathcal{G},\vec{n})$ can be written as a sum over the equivalence classes and a sum over the gauge transformations with $\bm{\lambda}(\xi_0)=0$. The spectral density for $E<M_{\Xi^-}$ becomes
\begin{flalign}
\rho(E;L) = \frac{1}{\pi} \Im \sum_{\substack{[(\mathcal{G},\vec{n})]\\ \text{s.t. } \mathcal{G} \in \mathcal{D}_s}} 
\hspace{3mm}
\sum_{\substack{\bm{\lambda} \text{ s.t. } \\ \bm{\lambda}(\xi_0)=0}} &
 \imath \left\{ \prod_{a \in V(\mathcal{G})-\xi_0} \int_{\mathbb{R} \times L^3} d^4 x(a) \right\} \times \nonumber \\
& \times \theta(x_0(\xi_1)) e^{\imath E x_0(\xi_1)} F^M_{c^{\bm{\lambda}}(\mathcal{G})}(x,\vec{n} + \delta \bm{\lambda})|_{x(\xi_0)=0} \ .
\label{appW:DR2}
\end{flalign}
By using the boundary conditions~\eqref{appW:BC2}, one can use the sum over the gauge transformations to reconstruct the integrals over $\mathbb{R}^4$:
\begin{flalign}
\rho(E;L) = \frac{1}{\pi} \Im \sum_{\substack{[(\mathcal{G},\vec{n})]\\ \text{s.t. } \mathcal{G} \in \mathcal{D}_s}} &
\imath \left\{ \prod_{a \in V(\mathcal{G})-\xi_0} \int_{\mathbb{R}^4} d^4 x(a) \right\} \times \nonumber \\
& \times 
\theta(x_0(\xi_1)) e^{\imath E x_0(\xi_1)} F^M_{\mathcal{G}}(x,\vec{n})|_{x(\xi_0)=0} \ .
\label{appW:DR3}
\end{flalign}

\paragraph{Strangeness flow in Feynman diagrams.} We introduce some general definitions. 

\textit{Paths.} A path $P$ connecting the vertices $a \neq b$ is a set of lines, with the property that a sequence $a = v_1, v_2, \dots , v_N = b$ of pairwise different vertices exists, together with a labelling $\ell_1,\ell_2,\dots,\ell_{N-1}$ of all lines in $P$, such that $v_k$, $v_{k+1}$ are the endpoints of $\ell_k$.

\textit{Loops and Wilson loops.} A loop $C$ is a set of lines, with the property that a sequence $v_1, v_2, \dots , v_N$ of pairwise different vertices exists, together with a labelling $\ell_1,\ell_2,\dots,\ell_N$ of all lines in $C$, such that $v_k$, $v_{k+1}$ are the endpoints of $\ell_k$ for $k=1,\cdots,N-1$, and $v_N$, $v_1$ are endpoints of $\ell_N$. The Wilson loop associated to $C$ is defined by
\begin{gather}
W(C,\vec{n}) = \sum_{\ell \in C} \vec{n}_\ell \ .
\end{gather}
Notice that a Wilson loop is gauge invariant.

\textit{Trees and axial gauge.} A tree $\mathcal{T}$ in a connected graph is a maximal set of lines that contains no loops. If $\mathcal{T}$  is a tree of $\mathcal{G}$, for any pair of vertices $a \neq b$ of $\mathcal{G}$, there is a unique path $P \subseteq \mathcal{T}$ connecting $a$ and $b$. Any connected graph with more than one vertex has at least one non-empty tree. Given a tree $\mathcal{T}$, a gauge field $\vec{n}$ is said to be in axial gauge with respect to $\mathcal{T}$ if $\vec{n}_\ell = \vec{0}$ for any $\ell \in \mathcal{T}$. It is easy to show that any gauge field is gauge-equivalent to some gauge field in axial gauge with respect to $\mathcal{T}$.

We give a closer look at the flavour structure of diagrams in $\mathcal{D}_s$, focusing in particular on strangeness. A line $\ell$ is said to be \textit{strange} if $S(\ell,i(\ell)) \neq 0$, and \textit{strangeless} otherwise. For any diagram $\mathcal{G} \in \mathcal{D}_s$, we define the subdiagram $\mathcal{G}_s$ by taking only the strange lines in $\mathcal{G}$ and vertices that are attached to these lines. Clearly the vertices $\xi_0$ and $\xi_1$ corresponding to the interpolating operators belong to $\mathcal{G}_s$. Because of the defining property of $\mathcal{D}_s$, there is no path in $\mathcal{G}_s$ connecting $\xi_0$ and $\xi_1$. We define $\mathcal{G}_{s,0}$ as the connected component of $\mathcal{G}_s$ containing $\xi_0$.  By specializing eq.~\eqref{appW:flavourconservation} to strangeness, and by observing that strangeness cannot flow outside of $\mathcal{G}_{s,0}$, one obtains the following:

\begin{proposition}
\label{appW:lemma:Gs0}
Strangeness is conserved within $\mathcal{G}_{s,0}$ at all vertices of $\mathcal{G}_{s,0}$ except $\xi_0$, i.e.
\begin{gather}
\sum_{\ell \in L(\mathcal{G}_{s,0})} S(\ell,a) = 0 \ , \quad \text{for } a \in V(\mathcal{G}_{s,0})-\xi_0 \ .
\label{appW:strangenessconservation}
\end{gather}

\end{proposition}

\bigskip

\begin{lemma}
\label{appW:lemma:loop}
For any pair $(\mathcal{G},\vec{n})$ such that $\mathcal{G} \in \mathcal{D}_s$ and $F^M_{\mathcal{G}}(x,\vec{n})$ is not identically zero, then a loop $C$ exists in $\mathcal{G}_{s,0}$ such that $|W(C,\vec{n})| \ge 1$.
\end{lemma}

\begin{proof}
We consider a tree $\mathcal{T}$ in $\mathcal{G}_{s,0}$, and we assume without loss of generality that $(\mathcal{G},\vec{n})$ is in axial gauge with respect to the tree $\mathcal{T}$. Notice that all lines in $\mathcal{G}_{s,0}$ are flavourful, therefore they can be either of type 3 (\textit{flavour-preserving}) or of type 4 (\textit{flavour-violating}). If $\ell \in \mathcal{T}$ then $\vec{n}_\ell=\vec{0}$ by definition of axial gauge. In this case $\ell$ must be flavour-preserving, otherwise the propagator would contribute with a $\delta_{\langle \vec{n}_\ell \rangle,1}=0$ factor to $F^M_{\mathcal{G}}(x,\vec{n})$.

If all lines in $\mathcal{G}_{s,0}$ were flavour-preserving, by conservation of strangeness at the internal vertices, eq.~\eqref{appW:strangenessconservation}, we would have
\begin{flalign}
0
=
& \sum_{a \in V(\mathcal{G}_{s,0}) - \xi_0} \sum_{\ell \in L(\mathcal{G}_{s,0})} S(\ell,a)
=
\sum_{a \in V(\mathcal{G}_{s,0})} \sum_{\ell \in L(\mathcal{G}_{s,0})} S(\ell,a) - \sum_{\ell \in L(\mathcal{G}_{s,0})} S(\ell,\xi_0)
= \nonumber \\
= & \sum_{\ell \in L(\mathcal{G}_{s,0})} [ S(\ell,i(\ell)) + S(\ell,f(\ell))] - \sum_{\ell \in L(\mathcal{G}_{s,0})} S(\ell,\xi_0) = - \sum_{\ell \in L(\mathcal{G}_{s,0})} S(\ell,\xi_0) \ ,
\end{flalign}
which is in contradiction with the fact that the interpolating operator has strangeness equal to $\pm 2$. Therefore at least a flavour-violating line $\bar{\ell} \in L(\mathcal{G}_{s,0}) - \mathcal{T}$ exists. Since the propagator of a flavour-violating line comes with a factor $\delta_{\langle \vec{n}_{\bar{\ell}} \rangle,1}$ then necessarily $\vec{n}_{\bar{\ell}} \neq \vec{0}$.

Consider the path $P$ in $\mathcal{T}$ that connects $i(\bar{\ell})$ to $f(\bar{\ell})$, then it is straightforward to prove that $C=P \cup \{ \bar{\ell} \}$ is a loop in $\mathcal{G}_{s,0}$ and
\begin{gather}
\vert W(C,\vec{n})\vert = |\vec{n}_{\bar{\ell}}| \neq 0 \ .
\end{gather}
The thesis follows from the fact that the components of $\vec{n}_{\bar{\ell}}$ are integers.

\end{proof}


\paragraph{Time-ordering and Euclidean kernel.} It is convenient to separate different time-orderings in integral~\eqref{appW:DR3}. Formally the set $T(\mathcal{G})$ of all time-orderings of $\mathcal{G}$ is defined as the set of permutations of $V(\mathcal{G})$, i.e. the set of all bijective functions between $V(\mathcal{G})$ and $\{1,2,\dots, |V(\mathcal{G})| \}$ where $|V(\mathcal{G})|$ is the number of vertices of $\mathcal{G}$. Given a time-ordering $\tau \in T(\mathcal{G})$, and two vertices $a,b \in V(\mathcal{G})$ such that $\tau(a) < \tau(b)$, we will say that $a$ is \textit{before} $b$ and $b$ is \textit{after} $a$ (with respect to $\tau$). We define the time-ordering function associated to $\tau$ as
\begin{gather}
\theta_\tau(x) = \prod_{\substack{a,b \in V(\mathcal{G})\\ \text{s.t. } \tau(a) < \tau(b)}} \theta(x_0(b)-x_0(a)) \ ,
\end{gather}
and we insert the identity in the integral~\eqref{appW:DR3} in the form of
\begin{gather}
1 = \sum_{\tau \in T(\mathcal{G})} \theta_\tau(x) \ .
\end{gather}
Not all time-orderings contribute to the spectral density for $E<M_{\Xi^-}$. States with no strange particles need to be able to propagate at some intermediate time between the two external vertices. Therefore only time-orderings such that $\xi_1$ is after any vertex in $\mathcal{G}_{s,0}$ contribute. We define $T_s(\mathcal{G})$ the set of such time-orderings. Eq.~\eqref{appW:DR3} becomes
\begin{flalign}
\rho(E;L) = \frac{1}{\pi} \Im \sum_{\substack{[(\mathcal{G},\vec{n})]\\ \text{s.t. } \mathcal{G} \in \mathcal{D}_s}} \sum_{\tau \in T_s(\mathcal{G})} &
\imath \left\{ \prod_{a \in V(\mathcal{G})-\xi_0} \int_{\mathbb{R}^4} d^4 x(a) \right\} \times \nonumber \\
& \times 
\theta_\tau(x) e^{\imath E x_0(\xi_1)} F^M_{\mathcal{G}}(x,\vec{n})|_{x(\xi_0)=0} \ .
\label{appW:DR4}
\end{flalign}

We choose a diagram and a time-ordering contributing to the previous formula. Without loss of generality we can assume that each line is oriented in such a way that its final point is not before its initial point.
Let $v^\tau_-$ be the latest vertex in $\mathcal{G}_{s,0}$, i.e. the vertex in $\mathcal{G}_{s,0}$ with the property that any other vertex in $\mathcal{G}_{s,0}$ is before it. Let $v^\tau_+$ be the vertex right after $v^\tau_-$, i.e. the vertex defined by $\tau(v^\tau_+) = \tau(v^\tau_-)+1$.

We construct the subdiagram $\mathcal{G}^\tau_+$ (resp. $\mathcal{G}^\tau_-$) by taking all the vertices in $\mathcal{G}$ not before $v^\tau_+$ (resp. not after $v^\tau_-$) and all lines in $\mathcal{G}$ connecting any pair of these vertices. Let $L^\tau_0$ be the set of lines with one endpoint in $\mathcal{G}^\tau_-$ and one in $\mathcal{G}^\tau_+$. This decomposition induces a factorization of the integrand in eq.~\eqref{appW:DR4}. The time-ordering function factorizes as
\begin{gather}
\theta_\tau(x) = \theta(x_0(v^\tau_+)-x_0(v^\tau_-)) \theta_{\tau_+}(x) \theta_{\tau_-}(x) \ , \\
\theta_{\tau_\pm}(x) = \prod_{\substack{a,b \in V(\mathcal{G}^\tau_\pm) \\ \text{s.t. } \tau(a) < \tau(b)}}
\theta(x_0(b)-x_0(a)) \ ,
\label{appW:decomposition1}
\end{gather}
where $\tau_+$ (resp. $\tau_-$) is the time-ordering restricted to $\mathcal{G}^\tau_+$ (resp. $\mathcal{G}^\tau_-$). The Feynman integrand factorizes as
\begin{gather}
F^M_{\mathcal{G}}(x,\vec{n})
=
\sum_\mu F^M_{\mu,\mathcal{G}^\tau_+}(x,\vec{n}) F^M_{\mu,\mathcal{G}^\tau_-}(x,\vec{n})
\prod_{\ell \in L^\tau_0} \  C^{\mathcal{G}}_\ell(\vec{n}_\ell) \ \Delta_M^+(\delta_\ell x +  L n_\ell;m_\ell)
\ ,
\label{appW:decomposition2}
\end{gather}
where $\mu$ is some collective Lorentz index (in the following we will omit the sum over $\mu$). Notice that $F^M_{\mu,\mathcal{G}^\tau_+}(x,\vec{n})$ (resp. $F^M_{\mu,\mathcal{G}^\tau_-}(x,\vec{n})$) contain differential operators acting on coordinates of the final (resp. initial) points of the lines in $L^\tau_0$. $\Delta_M^+(x,m)$ is the retarded Minkowskian propagator which we can conveniently write in time-momentum representation
\begin{gather}
\Delta_M^+(x;m) = \theta(x_0) \int \frac{d^3 p}{(2\pi)^3 2E} e^{-\imath (E x_0 - \vec{p} \vec{x})} \ , \quad \text{with } E=\sqrt{m^2 + \vec{p}^2} \ .
\end{gather}
We plug the factorizations~\eqref{appW:decomposition1} and~\eqref{appW:decomposition2} and the explicit representation of the retarded propagator into eq.~\eqref{appW:DR4}, we substitute $x_0(a) \to x_0(a) + x_0(v_-)$ for any $a \in V(\mathcal{G}_+)$, and we use invariance under translations of $F^M_{\mu,\mathcal{G}_+}(x,\vec{n})$. The integrals over the coordinates factorizes over the two subgraphs:
\begin{flalign}
\rho(E;L) = \frac{1}{\pi} \Im \sum_{\substack{[(\mathcal{G},\vec{n})]\\ \text{s.t. } \mathcal{G} \in \mathcal{D}_s}} \sum_{\tau \in T_s(\mathcal{G})} &
\left\{ \prod_{\ell \in L^\tau_0} \ \int \frac{d^3 p_\ell}{(2\pi)^3 2E_\ell} \right\} \times \nonumber \\
& \times 
(2\pi)^3 \delta^3 ({\textstyle \sum_{\ell \in L^\tau_0}} \vec{p}_\ell ) K^-_{\mu,\mathcal{G},\tau}(E,\vec{p},\vec{n}) R^+_{\mu,\mathcal{G},\tau}(E,\vec{p},\vec{n})
\ .
\label{appW:DR5}
\end{flalign}
The explicit expression for $R^+_{\mu,\mathcal{G},\tau}(E,\vec{p},\vec{n})$ is given for sake of completeness, but we will not need it in the current discussion
\begin{flalign}
R^+_{\mu,\mathcal{G},\tau}(E,\vec{p},\vec{n}) = &
\frac{\prod_{\ell \in L^\tau_0} C^{\mathcal{G}}_\ell(\vec{n}_\ell)}{\sum_{\ell \in L^\tau_0} E_\ell - E -\imath \epsilon } \left\{ \prod_{a \in V(\mathcal{G}^\tau_+)} \int_{\mathbb{R}^4} d^4 x(a) \right\} \theta_{\tau_+}(x)
\times  \nonumber
\\
& \times
F^M_{\mu,\mathcal{G}^\tau_+}(x,\vec{n}) \delta^4(x(v^\tau_+))
e^{\imath E x_0(\xi_1)} \prod_{\ell \in L^\tau_0} e^{-\imath E_\ell x_0(f(\ell))}  e^{\imath \vec{p}_\ell \vec{x}(f(\ell))} \ .
\end{flalign}
The explicit expression for $K^-_{\mu,\mathcal{G},\tau}(E,\vec{p},\vec{n})$ is given after Wick rotation of the coordinate integrals $x_0(a) \to - \imath x_0(a)$:
\begin{flalign}
K^-_{\mathcal{G},\tau}(E,\vec{p},\vec{n})
= &
w_{\mu,\mathcal{G},\tau}
\left\{ \prod_{a \in V(\mathcal{G}^\tau_-)} \int_{\mathbb{R}^4} d^4 x(a) \right\}
\delta^4(x(\xi_0)) \, \theta_{\tau_-}(x) F^E_{\mu,\mathcal{G}^\tau_-}(x,\vec{n}) \times  \nonumber
\\
& \times
e^{E x_0(v^\tau_-)} 
e^{- \sum_{\ell \in L^\tau_0} E_\ell [x_0(v^\tau_-)-x_0(i(\ell))]} e^{- \imath  \sum_{\ell \in L^\tau_0} \vec{p}_\ell [\vec{x}(i(\ell)) + L \vec{n}_\ell]} \ ,
\label{appW:EK0}
\end{flalign}
where $E_\ell = \sqrt{m_\ell^2 + \vec{p}_\ell^2}$, $F^E_{\mathcal{G}^\tau_-}(x,\vec{n})$ is constructed as $F^M_{\mathcal{G}^\tau_-}(x,\vec{n})$ except that the Minkowskian propagators and vertices are replaced by the Euclidean ones, and $w_{\mu,\mathcal{G},\tau}$ is some constant phase factor. The Wick rotation is allowed for $K^-_{\mathcal{G},\tau}(E,\vec{p},\vec{n})$ because all states propagating in between the interpolating operator $\xi_0$ and the latest vertex $v^\tau_-$ have energy that is higher than $E$ by construction. We will refer to the function $K^-_{\mathcal{G},\tau}(E,\vec{p},\vec{n})$ as \textit{Euclidean kernel} associated to the diagram $\mathcal{G} \in \mathcal{D}_s$ and the time-ordering $\tau \in T_s(\mathcal{G})$.

We highlight the following observation, which follows from the construction of the Euclidean kernel.
\begin{observation}
The latest vertex $v^\tau_-$ of $\mathcal{G}_{s,0}$ is also the latest vertex of $\mathcal{G}^\tau_-$. Since $\mathcal{G}^\tau_-$ contains all vertices of $\mathcal{G}$ not after $v^\tau_-$, and all lines of $\mathcal{G}$ with both endpoints not after $v^\tau_-$, then $\mathcal{G}_{s,0}$ is a subdiagram of $\mathcal{G}^\tau_-$.
\end{observation}


\paragraph{Large-volume behaviour of the Euclidean kernel.} We will see that the Euclidean kernel provides the exponential suppression in the infinite-volume limit. We will establish this result in a few steps.

\begin{theorem}
\label{appW:exp1}
For large $L$, we have
\begin{gather}
\ln |K^-_{\mu,\mathcal{G},\tau}(E,\vec{p},\vec{n})| = - L \mathcal{E}_{\mathcal{G},\tau}(E,\vec{p},\vec{n}) + \mathcal{O}(\ln L) \ ,
\label{appW:asymptK}
\end{gather}
where the function $\mathcal{E}_{\mathcal{G},\tau}(E,\vec{p},\vec{n})$ is given by
\begin{gather}
\mathcal{E}_{\mathcal{G},\tau}(E,\vec{p},\vec{n}) = \nonumber \\
= \min_{x \in D_{\mathcal{G},\tau}} \left\{ - E x_0(v^\tau_-) + \sum_{\ell \in L^\tau_0} E_\ell [x_0(v^\tau_-)-x_0(i(\ell))] + \sum_{\ell \in L(\mathcal{G}^\tau_-)} m_\ell | \delta_\ell x + n_\ell| \right\} \ ,
\label{appW:eps0}
\\
D_{\mathcal{G},\tau} = \{ ( \, x(a) \, )_{a \in V(\mathcal{G}^\tau_-)} \ | \ x(a)  \in \mathcal{R}^4 \ , \ \theta_{\tau_-}(x) = 1 \ , \ x(\xi_0)=0 \} \ .
\end{gather}
\end{theorem}


\begin{proof}

We use the heat-kernel representation for the Euclidean propagator
\begin{gather}
\Delta_E(x;m) = \int_0^\infty ds \ \frac{e^{-\frac{x^2}{4s} -m^2 s}}{(4\pi s)^2} \ .
\end{gather}
When plugging this in eq.~\eqref{appW:EK0}, we need to introduce a Schwinger parameter $s_\ell$ for each line $\ell \in L(\mathcal{G}^\tau_-)$. By substituting $s_\ell \to s_\ell L / 2 m_\ell$ and $x(a) \to L x(a)$ we get the general form for the Euclidean kernel
\begin{flalign}
K^-_{\mu,\mathcal{G},\tau}(E,\vec{p},\vec{n})
= &
\prod_{\ell \in L(\mathcal{G}^\tau_-)} 
\int_0^\infty ds_\ell
\prod_{a \in V(\mathcal{G}^\tau_-)-\xi_0} \int_{\mathbb{R}^4} d^4 x(a) \delta(x_0(\xi_0)) \, \theta_{\tau_-}(x)
\times \nonumber \\
& \times
L^\alpha \frac{P_\mu(s,x,\vec{n})}{Q(s)} e^{-L \left\{ X(s,x,E,\vec{p},\vec{n}) - \imath \sum_{\ell \in L^\tau_0} \vec{p}_\ell [\vec{x}(i(\ell)) + \vec{n}_\ell] \right\} }
\ ,
\label{appW:EK1}
\end{flalign}
where $P(s,x,\vec{n})$ and $Q(s)$ are polynomials that come from the derivatives in the vertex operator and explicit powers of $s$ in the heat-kernel representation, and
\begin{gather}
X(s,x,E,\vec{p},\vec{n}) = \nonumber \\
=
- E x_0(v^\tau_-) + \sum_{\ell \in L^\tau_0} E_\ell [x_0(v^\tau_-)-x_0(i(\ell))] + \sum_{\ell \in L(\mathcal{G}^\tau_-)} \frac{m_\ell}{2} \left\{ s_\ell + \frac{1}{s_\ell} | \delta_\ell x + n_\ell|^2 \right\} \ .
\end{gather}
The large-$L$ expansion of the integral in eq.~\eqref{appW:EK1} is given by a saddle-point approximation, i.e. by expanding about the minima of $X$. This yields the asymptotic behaviour~\eqref{appW:asymptK} with
\begin{gather}
\mathcal{E}_{\mathcal{G},\tau}(E,\vec{p},\vec{n}) = \min_{x \in D_{\mathcal{G},\tau}} \ 
\min_{s_\ell \in [0,\infty)} \ X(s,x,E,\vec{p},\vec{n}) \ .
\end{gather}
Eq.~\eqref{appW:eps0} is obtained by performing the trivial minimization over each $s_\ell$.

\end{proof}

The statement of theorem~\ref{appW:exp1} makes sense because $\mathcal{E}_{\mathcal{G},\tau}(E,\vec{p},\vec{n})$ is strictly positive. We will show this fact in two steps, and we will provide also a lower bound for this function.

\begin{theorem}
\label{appW:exp2}
The function $\mathcal{E}_{\mathcal{G},\tau}(E,\vec{p},\vec{n})$ defined in eq.~\eqref{appW:eps0} satisfies
\begin{gather}
\mathcal{E}_{\mathcal{G},\tau}(E,\vec{p},\vec{n}) \ge 0 \ , \quad \text{for } E<M_{\Xi^-} \ .
\end{gather}
\end{theorem}

\begin{proof}

Let $N$ be the number of vertices of the diagram $\mathcal{G}^\tau_-$. We assume that the orientation of lines in $\mathcal{G}^\tau_-$ is chosen in such a way that
\begin{gather}
\tau(f(\ell)) \ge \tau(i(\ell)) \ .
\end{gather}

We construct an auxiliary graph $\mathcal{H}$ by starting from an empty graph, and by adding elements to it accordingly to a set of rules.
\begin{enumerate}[noitemsep,topsep=-.7em]

\item We grow $\mathcal{H}$ by adding all the vertices of $\mathcal{G}^\tau_-$. At this stage:
\begin{gather}
V(\mathcal{H}) = V(\mathcal{G}^\tau_-) \ , \qquad L(\mathcal{H}) = \varnothing \ .
\end{gather}
We equip the graph $\mathcal{H}$ with the time-ordering function $\tau$ inherited from $\mathcal{G}^\tau_-$.

\item For each \textit{mother} line $\ell$ of $\mathcal{G}^\tau_-$ such that $\tau(f(\ell)) - \tau(i(\ell)) \le 1$, we add the \textit{daughter} line $\ell$ to $\mathcal{H}$ with the same incidence relations (notice that the endpoints of $\ell$ are already in $\mathcal{H}$).

\item For each \textit{mother} line $\ell$ of $\mathcal{G}^\tau_-$ such that $\tau(f(\ell)) - \tau(i(\ell)) > 1$, we imagine to cut the line $\ell$ at each intermediate timeslice. More precisely we add the $n=\tau(f(\ell)) - \tau(i(\ell))-1$ intermediate vertices $v_1, \dots, v_n$ and the $n+1$ \textit{daughter} lines $\ell_0, \dots \ell_n$ to $\mathcal{H}$ with the following incidence relations:
\begin{gather}
i(\ell_0) = i(\ell) \ , \quad f(\ell_0) = v_1 \ , \\
i(\ell_k) = v_k \ , \quad f(\ell_k) = v_{k+1} \ , \quad \text{for } k=1,\dots,n-1 \ , \\
i(\ell_n) = v_n \ , \quad f(\ell_n) = f(\ell) \ .
\end{gather}
We also time-order the extra vertices and we declare that
\begin{gather}
\tau(v_k) = \tau(i(\ell))+k \ .
\end{gather}
We associate a mass and a $\mathbb{Z}^3$ gauge field to each line of $\mathcal{H}$. The masses of the daughter lines are all equal to the mass of the mother line, while the gauge field of the original line is transferred completely only to the first daughter line, i.e.
\begin{gather}
M_{\ell_k} = m_\ell \ , \\
\vec{n}_{\ell_k} = \vec{n}_\ell \delta_{k,0} \ .
\end{gather}
All quantum numbers are also naturally transferred from the mother line to the daughter lines.

\item For each \textit{mother} line $\ell$ of $L^\tau_0$ such that $\tau(v^\tau_-) - \tau(i(\ell)) > 0$, we add a final point at the latest timeslice and we imagine to cut the line $\ell$ at each intermediate timeslice. More precisely we add the $n=\tau(v^\tau_-) - \tau(i(\ell))$ vertices $v_1, \dots, v_n$ and the $n$ \textit{daughter} lines $\ell_0, \dots \ell_{n-1}$ to $\mathcal{H}$ with the following incidence relations:
\begin{gather}
i(\ell_0) = i(\ell) \ , \quad f(\ell_0) = v_1 \ , \\
i(\ell_k) = v_k \ , \quad f(\ell_k) = v_{k+1} \ , \quad \text{for } k=1,\dots,n-1 \ .
\end{gather}
Analogously to the previous case, we declare:
\begin{gather}
\tau(v_k) = \tau(i(\ell))+k \ , \\
M_{\ell_k} = E_\ell \ , \\
\vec{n}_{\ell_k} = \vec{n}_\ell \delta_{k,0} \ ,
\end{gather}
and we transfer the quantum numbers of the mother line to the daughter lines.

\end{enumerate}

\begin{figure}
\centering

\includegraphics[width=.6\textwidth]{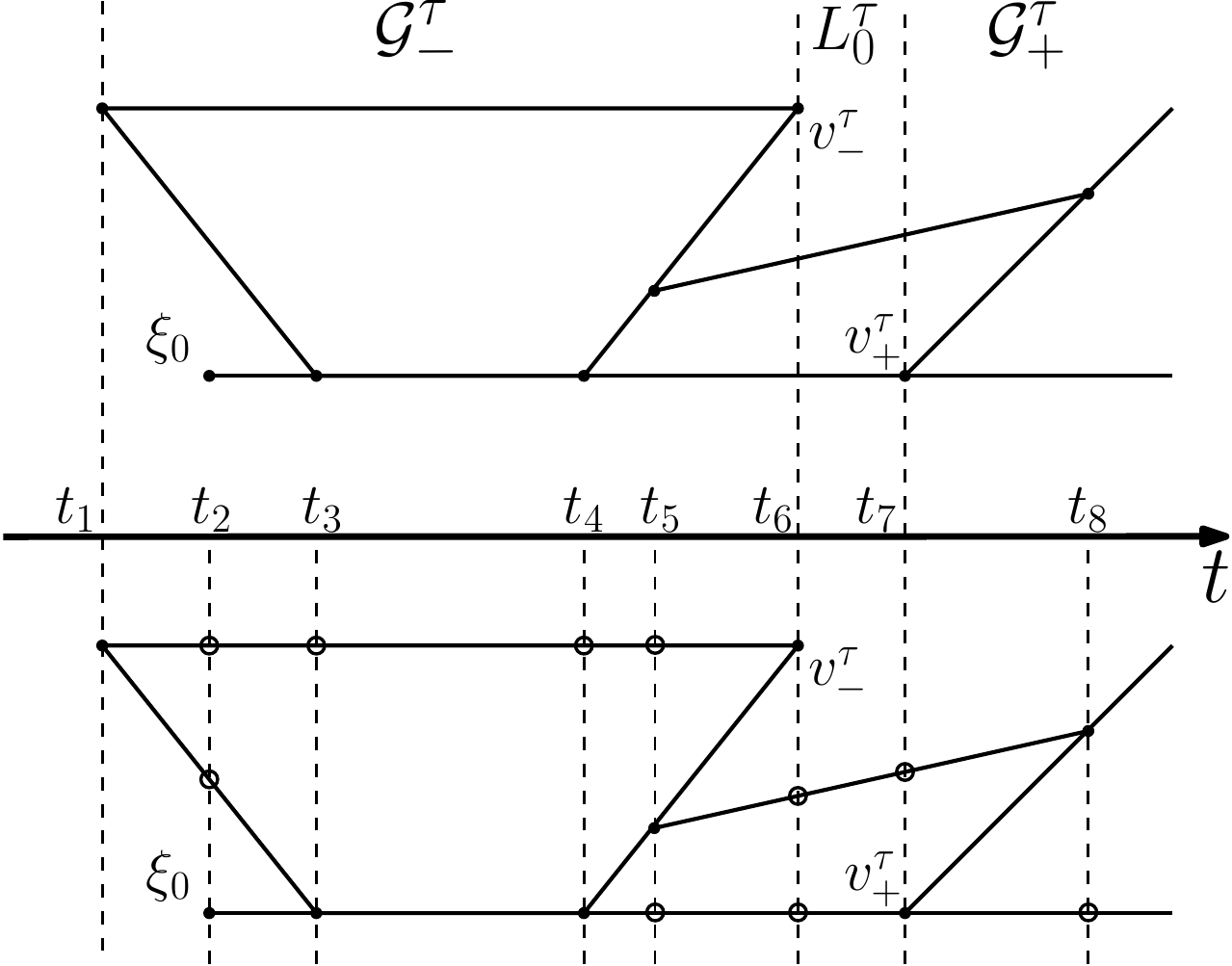}

\caption{\label{fig:appW:t-ordering}
Exemplification of the construction given in the proof of theorem~\ref{appW:exp2}. We start with a diagram whose vertices have been time-ordered (top diagram). Each vertex defines a timeslice. The diagram in the bottom is obtained from the top one by inserting a new vertex every time a line intersects a timeslice. Vertices in the same timeslice are forced to have the same time coordinate.
}
\end{figure}

The above construction is schematically represented in fig.~\ref{fig:appW:t-ordering}. The auxiliary diagram $\mathcal{H}$ is useful to rewrite the function $\mathcal{E}_{\mathcal{G},\tau}(E,\vec{p},\vec{n})$ defined in eq.~\eqref{appW:eps0} in a way that its timeslice structure be manifest.

We define $V_k(\mathcal{H})$ for $k=1,\dots,N$ as the set that contains all vertices $a$ with $\tau(a)=k$, i.e.
\begin{gather}
V_k(\mathcal{H}) = \tau^{-1}(k) \ .
\end{gather}
The sets $V_k(\mathcal{H})$ for $k=1,\dots,N$ constitute a partition of $V(\mathcal{H})$,
\begin{gather}
V(\mathcal{H}) = \bigsqcup_{k=1}^N V_k(\mathcal{H}) \ .
\end{gather}
We will refer to $V_k(\mathcal{H})$ as the $k$-th \textit{timeslice} of $\mathcal{H}$. We define $L_k(\mathcal{H})$ for $k=1,\dots,N-1$ as the set that contains all lines $\ell$ such that $i(\ell)$ is in the $k$-th timeslice and $f(\ell)$ is in the $(k+1)$-th timeslice. The only lines of $\mathcal{H}$ that are left out are the looplines $\ell$ with $i(\ell)=f(\ell)$ which we collect in the set $L_{loop}(\mathcal{H})$,
\begin{gather}
L(\mathcal{H}) = \bigsqcup_{k=1}^{N-1} L_k(\mathcal{H}) \ \sqcup \ L_{loop}(\mathcal{H}) \ .
\end{gather}

We change the way in which we associate coordinates to vertices. We associate the time-coordinate $t_k$ to the timeslice $V_k(\mathcal{H})$, and independent spatial coordinates $\vec{x}(a)$ to each vertex $a \in V(\mathcal{H})$. We force all vertices in the same timeslice to have the same time coordinate. Therefore the four-dimensional coordinate vector is mapped into
\begin{gather}
x(a) = (t_{\tau(a)},\vec{x}(a)) \ .
\end{gather}
Notice that in each timeslice $V_k(\mathcal{H})$ there is only one vertex of the original diagram $\mathcal{G}^\tau_-$, therefore $t_k$ does coincide with the temporal coordinate of this vertex.

By using iteratively the property that the shortest path between two points is the straight line, in the form of
\begin{gather}
\min_{\vec{x}_1} \left\{ \sqrt{(t_{k+2}-t_{k+1})^2 + |\vec{x}_2 - \vec{x}_1|^2} + \sqrt{(t_{k+1}-t_{k})^2 + |\vec{x}_1 - \vec{x}_0|^2} \right\} = \nonumber \\
= \sqrt{(t_{k+2}-t_{k})^2 + |\vec{x}_2 - \vec{x}_0|^2} 
\ , \quad \text{for } t_{k+2}>t_{k+1}>t_k
\ ,
\end{gather}
and the property that the shortest path between a point and a plane is the the straight line that is orthogonal to the plane, in the form of
\begin{gather}
\min_{\vec{x}_1} \sqrt{(t_{k+1}-t_{k})^2 + |\vec{x}_1 - \vec{x}_0|^2}
= t_{k+1}-t_{k}  \ ,
\quad \text{for } t_{k+1}>t_k
\ ,
\end{gather}
and by recalling that $\tau(v^\tau_-)=N$ as $v^\tau_-$ is the latest vertex in $\mathcal{G}^\tau_-$, one can easily prove that
\begin{gather}
\mathcal{E}_{\mathcal{G},\tau}(E,\vec{p},\vec{n}) = \nonumber \\
= \min_{\vec{x}} \min_{t_1 < t_2 < \dots < t_N } \left\{ - E [t_N - t_{\tau(\xi_0)}] + \sum_{\ell \in L(\mathcal{H})} M_\ell | x(f(\ell)) - x(i(\ell)) + n_\ell| \right\}
\ .
\label{appW:eps1}
\end{gather}

We choose quantities $\Delta_\ell$ associated to the lines of $\mathcal{H}$ satisfying the following constraints
\begin{gather}
\Delta_\ell = 0 \ , \quad \text{for looplines } \ell \in L_{loop}(\mathcal{H}) \ , \nonumber \\
\Delta_\ell = 0 \ , \quad \text{for } \ell \in L_k(\mathcal{H}) \text{ with } k<\tau(\xi_0) \ , \nonumber \\
0 \le \Delta_\ell \le M_\ell \ , \quad \text{for any } \ell \ , \nonumber \\
E = \sum_{\ell \in L_k(\mathcal{H})} \Delta_\ell \ , \quad \text{for any } k \ge \tau(\xi_0) \ .
\label{appW:deltaconstr}
\end{gather}
Notice that the last two constraints are in tension with each other. Quantities $\Delta_\ell$ satisfying such constraints exist for $E<M_{\Xi^-}$ thanks to the following observation.

\begin{observation}
By construction of $\mathcal{G}^\tau_-$, only states with energy not lower than $M_{\Xi^-}$ propagate at any time between the vertices $\xi_0$ and $v^\tau_-$, i.e.
\begin{gather}
\sum_{\ell \in L_k(\mathcal{H})} M_\ell \ge M_{\Xi^-} \ , \quad \text{for any } k \ge \tau(\xi_0) \ .
\end{gather}
\end{observation}

Besides the above constraints, the quantities $\Delta_\ell$ are largely arbitrary (and we will use this arbitrariness later on). The contribution $E [t_N - t_{\tau(\xi_0)}]$ in eq.~\eqref{appW:eps1} can be distributed over the lines of $\mathcal{H}$ yielding
\begin{gather}
\mathcal{E}_{\mathcal{G},\tau}(E,\vec{p},\vec{n}) = \nonumber \\
= \min_{\vec{x}} \min_{t_1 < t_2 < \dots < t_N } \sum_{\ell \in L(\mathcal{H})} \left\{ - \Delta_\ell [x_0(f(\ell)) - x_0(i(\ell))] +  M_\ell | x(f(\ell)) - x(i(\ell)) + n_\ell| \right\}
\ ,
\label{appW:eps2}
\end{gather}
We can provide a lower bound for each term of the above sum by using the following inequality
\begin{gather}
- \Delta | \delta t | + M \sqrt{(\delta t)^2 + z^2} \ge |z| \sqrt{M^2 - \Delta^2} \ ,
\end{gather}
valid for any $z$ and $\delta t$, as long as $M \ge \Delta$. Therefore we get an estimate for $\mathcal{E}_{\mathcal{G},\tau}(E,\vec{p},\vec{n})$ in which the time minimization has disappeared
\begin{gather}
\mathcal{E}_{\mathcal{G},\tau}(E,\vec{p},\vec{n}) \ge \min_{\vec{x}} \sum_{\ell \in L(\mathcal{H})} | \vec{x}(f(\ell)) - \vec{x}(i(\ell)) + \vec{n}_\ell| \sqrt{M_\ell^2 - \Delta_\ell^2}
\ ,
\label{appW:eps3}
\end{gather}
and which shows explicitly that
\begin{gather}
\mathcal{E}_{\mathcal{G},\tau}(E,\vec{p},\vec{n}) \ge 0 \ , \quad \text{for } E<M_{\Xi^-} \ .
\end{gather}

\end{proof}


\begin{theorem}
\label{appW:exp3}
The function $\mathcal{E}_{\mathcal{G},\tau}(E,\vec{p},\vec{n})$ defined in eq.~\eqref{appW:eps0} satisfies
\begin{gather}
\mathcal{E}_{\mathcal{G},\tau}(E,\vec{p},\vec{n}) \ge \mathcal{M}(E) \ , \quad \text{for } E<M_{\Xi^-} \ ,
\end{gather}
where the function $\mathcal{M}(E)$ is given by
\begin{gather}
\mathcal{M}(E) = 
\begin{cases}
M_{K^\pm} & \text{if } 0 < E \le M_p \\
\left[ M_{K^\pm}^2 - \left( \frac{E-M_p}{2} \right)^2 \right]^{1/2} \quad & \text{if } M_p \le E \le \frac{M_{\Lambda^0}^2 - M_{K^\pm}^2}{M_p} \\
\left[ M_{K^\pm}^2 - \left( \frac{E^2 - M_{\Lambda^0}^2 + M_{K^\pm}^2}{2E} \right)^2 \right]^{1/2} \quad & \text{if } \frac{M_{\Lambda^0}^2 - M_{K^\pm}^2}{M_p} \le E < M_{\Xi^-}
\end{cases}
\ .
\label{appW:emm}
\end{gather}

\end{theorem}


\begin{proof}

We use the construction in the proof of theorem~\ref{appW:exp2}. Since each term in the sum in the r.h.s. of eq.~\eqref{appW:eps3} is positive, a looser lower bound on the function $\mathcal{E}_{\mathcal{G},\tau}(E,\vec{p},\vec{n})$ can be provided by restricting the sum to a subset of $L(\mathcal{H})$.

We choose a strange loop $C$ in $\mathcal{G}_{s,0}$ with $|W(C,\vec{n})| \ge 1$. Such a loop exists thanks to lemma~\ref{appW:lemma:loop}. Since $\mathcal{G}_{s,0}$ is a subgraph of $\mathcal{G}^\tau_-$, $C$ is also a loop in $\mathcal{G}^\tau_-$. We construct the loop $\tilde{C}$ of $\mathcal{H}$ by replacing each mother line in $C$ with the corresponding daughter lines. Since the gauge field of each mother line is equal to the sum of the gauge fields of the corresponding daughter lines, it is clear that
\begin{gather}
|W(\tilde{C},\vec{n})| = |W(C,\vec{n})| \ge 1 \ .
\end{gather}

By using this particular loop in eq.~\eqref{appW:eps3}, we get the following bound for $\mathcal{E}_{\mathcal{G},\tau}(E,\vec{p},\vec{n})$:
\begin{flalign}
\mathcal{E}_{\mathcal{G},\tau}(E,\vec{p},\vec{n}) \ge & \min_{\vec{x}} \sum_{\ell \in \tilde{C}} | \vec{x}(f(\ell)) - \vec{x}(i(\ell)) + \vec{n}_\ell| \sqrt{M_\ell^2 - \Delta_\ell^2}
\ge \nonumber \\
\ge & \mathcal{M}(E,\Delta) \min_{\vec{x}} \sum_{\ell \in \tilde{C}} | \vec{x}(f(\ell)) - \vec{x}(i(\ell)) + \vec{n}_\ell|
\ge \nonumber \\
\ge & \mathcal{M}(E,\Delta) | W(\tilde{C},\vec{n}) | \ge \mathcal{M}(E,\Delta) \ ,
\label{appW:eps4}
\end{flalign}
where we have used the triangular inequality iteratively along the loop, and we have defined
\begin{gather}
\mathcal{M}(E,\Delta) = \min_{\ell \in \tilde{C}} \sqrt{M_\ell^2 - \Delta_\ell^2} \ .
\label{appW:emm0}
\end{gather}
The $\Delta_\ell$'s are arbitrary quantities satisfying the constraints in eq.~\eqref{appW:deltaconstr}. We can use this arbitrariness in order to optimize the above lower bound.

We separate two cases: either $\tilde{C}$ is a single loopline or $\tilde{C}$ has more than one line, none of which being a loopline. In the first case, $\Delta_\ell=0$ for $\ell \in \tilde{C}$ and
\begin{gather}
\mathcal{M}(E,\Delta) = M_\ell \ge M_{K^\pm} \ ,
\end{gather}
where the last inequality comes from the fact that the $K^\pm$ is the lightest strange particle.

We consider now the second possibility, i.e. $\tilde{C}$ has more than one line none of which being a loopline,
\begin{gather}
\tilde{C} \subseteq \bigsqcup_{k=1}^{N-1} L_k(\mathcal{H}) \ .
\end{gather}
We define
\begin{gather}
\tilde{C}_k = \tilde{C} \cap L_k(\mathcal{H}) \ .
\end{gather}
The minimization in eq.~\eqref{appW:emm0} can be performed in two steps, according to
\begin{gather}
\mathcal{M}(E,\Delta) = \min_{\substack{k \text{ s.t.} \\ \tilde{C}_k \neq \varnothing}} \  \min_{\ell \in \tilde{C}_k} \sqrt{M_\ell^2 - \Delta_\ell^2} \ .
\label{appW:emm1}
\end{gather}
We choose a $k$ such that $\tilde{C}_k$ is not empty. If $k<\tau(\xi_0)$ then all $\Delta_\ell$ are equal to zero therefore:
\begin{gather}
\min_{\ell \in \tilde{C}_k} \sqrt{M_\ell^2 - \Delta_\ell^2} 
=
\min_{\ell \in \tilde{C}_k} M_\ell
\ge M_{K^\pm}
\ .
\end{gather}
Let us consider $k\ge \tau(\xi_0)$. Since the loop has to close, $\tilde{C}_k$ contains an even number of lines. We separate two possibilities again: the number of lines in $\tilde{C}_k$ with $(-1)^B = 1$ is either even or odd.

If the number of lines in $\tilde{C}_k$ with $(-1)^B = 1$ is even, since all states propagating between $\xi_0$ and $\xi_1$ must have $(-1)^B=-1$, there must be a line $\ell_0$ in $L_k(\mathcal{H}) - \tilde{C}$ with $(-1)^B=-1$. Since the proton is the lightest particle with $(-1)^B=-1$ then surely $M_{\ell_0} \ge M_p$. We also pick two lines $\ell_1$ and $\ell_2$ in $\tilde{C}_k$ and we observe that $M_{\ell_1},M_{\ell_2} \ge M_{K^\pm}$. We assign
\begin{gather}
\Delta_{\ell_0} = \min \left\{ E , M_p \right\} \ , \nonumber \\
\Delta_{\ell_1} = \Delta_{\ell_2} = \max \left\{ 0 , \frac{E-M_p}{2} \right\} \ , \nonumber \\
\Delta_\ell=0 \ , \quad \text{for } \ell \in L_k(\mathcal{H}) - \{ \ell_0 , \ell_1, \ell_2 \} \ .
\end{gather}
By using the fact that $E < M_{\Xi^-}$ one can easily check that the constraints in eq.~\eqref{appW:deltaconstr} are satisfied. By replacing all masses in the strange loop with $M_{K^\pm}$ we get
\begin{gather}
\min_{\ell \in \tilde{C}_k} \sqrt{M_\ell^2 - \Delta_\ell^2} 
\ge 
\begin{cases}
M_{K^\pm} & \text{if } E \le M_p \\
\left[ M_{K^\pm}^2 - \left( \frac{E-M_p}{2} \right)^2 \right]^{1/2} \quad & \text{if } E \ge M_p
\end{cases}
\ .
\end{gather}

If the number of lines in $\tilde{C}_k$ with $(-1)^B = 1$ is odd, then there is at least a line $\ell_1$ in $\tilde{C}_k$ with $(-1)^B=-1$. Since the $\Lambda^0$ is the lightest strange particle with $(-1)^B=-1$ then surely $M_{\ell_1} \ge M_{\Lambda^0}$. We also pick another line $\ell_2$ in $\tilde{C}_k$. We assign
\begin{gather}
\Delta_{\ell_1} = \min \left\{ E , \frac{E^2 + M_{\Lambda^0}^2 - M_{K^\pm}^2}{2E} \right\} \ , \nonumber \\
\Delta_{\ell_2} = \max \left\{ 0 , \frac{E^2 - M_{\Lambda^0}^2 + M_{K^\pm}^2}{2E} \right\} \ , \nonumber \\
\Delta_\ell=0 \ , \quad \text{for } \ell \in L_k(\mathcal{H}) - \{ \ell_1, \ell_2 \} \ .
\end{gather}
By using the fact that $E < M_{\Xi^-}$ one can easily check that the constraints in eq.~\eqref{appW:deltaconstr} are satisfied. By replacing the mass $M_{\ell_1}$ with $M_{\Lambda^0}$ and all other masses in the strange loop with $M_{K^\pm}$ we get
\begin{gather}
\min_{\ell \in \tilde{C}_k} \sqrt{M_\ell^2 - \Delta_\ell^2} 
\ge 
\begin{cases}
M_{K^\pm} & \text{if } E \le \sqrt{M_{\Lambda^0}^2 - M_{K^\pm}^2} \\
\left[ M_{K^\pm}^2 - \left( \frac{E^2 - M_{\Lambda^0}^2 + M_{K^\pm}^2}{2E} \right)^2 \right]^{1/2} \quad & \text{if } E \ge \sqrt{M_{\Lambda^0}^2 - M_{K^\pm}^2}
\end{cases}
\ .
\end{gather}

By combining all discussed cases, we get that it is always possible to choose the quantities $\Delta_\ell$ in such a way that the constraints in eq.~\eqref{appW:deltaconstr} are satisfied and 
\begin{gather}
\mathcal{M}(E,\Delta) \ge \mathcal{M}(E) \ ,
\end{gather}
where $\mathcal{M}(E)$ is defined in eq.\eqref{appW:emm}. This concludes the proof of the theorem.

\end{proof}


\paragraph{Large-volume behaviour of the spectral density.} One can easily reproduce the whole construction presented in this appendix around the external vertex $\xi_1$ instead of $\xi_0$. The general structure of the spectral density is:
\begin{flalign}
\rho(E;L) = & \frac{1}{\pi} \sum_{[(\mathcal{G},\vec{n})]} \sum_{\tau}
\left\{ \prod_{\ell \in L^\tau_-} \ \int \frac{d^3 p_\ell}{(2\pi)^3 2E_\ell} \right\}
(2\pi)^3 \delta^3 ({\textstyle \sum_{\ell \in L^\tau_-}} \vec{p}_\ell ) \times \nonumber \\
& \times
\left\{ \prod_{\ell \in L^\tau_+} \ \int \frac{d^3 p_\ell}{(2\pi)^3 2E_\ell} \right\}
(2\pi)^3 \delta^3 ({\textstyle \sum_{\ell \in L^\tau_+}} \vec{p}_\ell ) \times \nonumber \\
& \times
K^-_{\mu,\mathcal{G},\tau}(E,\vec{p},\vec{n}) R^0_{\mu\nu,\mathcal{G},\tau}(E,\vec{p},\vec{n}) K^+_{\nu,\mathcal{G},\tau}(E,\vec{p},\vec{n})
\ ,
\label{appW:DR6}
\end{flalign}
where $K^-_{\mu,\mathcal{G},\tau}(E,\vec{p},\vec{n})$ and $K^+_{\nu,\mathcal{G},\tau}(E,\vec{p},\vec{n})$ are Euclidean kernels that include a flavour-violating strange loop connected to $\xi_0$ and $\xi_1$ respectively. The Euclidean kernels satisfy
\begin{gather}
\ln |K^\pm_{\mu,\mathcal{G},\tau}(E,\vec{p},\vec{n})| \le - L \mathcal{M}(E) + \mathcal{O}(\ln L) \ ,
\end{gather}
thanks to theorems~\ref{appW:exp1} and~\ref{appW:exp3}. The desired eq.~\eqref{appW:finalbound} follows by observing that $R^0_{\mathcal{G},\tau}(E,\vec{p},\vec{n})$ has a finite large-$L$ limit and the phase-space integral generates at most powers in the volume.

\section{Corrections to charged-hadron masses in finite volume}
\label{app:massformula}
In this appendix we want to calculate the power-like finite-volume corrections to the masses of stable hadrons due to electromagnetic interactions at order $e^2$.

Stable hadrons are identified by their (finite-volume) flavour numbers. For instance the charged pion has baryon number $B=0 \mod 2$, strangeness $F_s = 0 \mod 2$ and electric charge $Q=1 \mod 2$. Given some flavour sector defining the target stable hadron $h$ and some momentum $\vec{p}$, we denote by $|h(\vec{p}),\sigma\rangle$ the lightest eigenstates of the QCD Hamiltonian $H_0$ in the given flavour sector, with momentum $\vec{p}$ and with energy $E_{h,0}(\vec{p},L)$. We assume that states with zero momentum are lighter than the others, and we refer to their energy as the $\mathcal{O}(e^0)$ finite-volume mass
\begin{gather}
E_{h,0}(\vec{p},L) > E_{h,0}(\vec{0},L) \equiv m_0(L) \ , \quad \text{if } \vec{p} \neq \vec{0} \ .
\end{gather}
The states $| h(\vec{0}),\sigma \rangle$ with $\sigma=1,\dots,d_s$ transform under some (possibly spinorial) representation of the cubic group. Normalization is chosen such that one recovers the relativistic normalization in infinite volume
\begin{gather}
\langle h(\vec{p}),\sigma | h(\vec{p}'),\sigma' \rangle = 2 E_{h,0}(\vec{p},L) L^3 \delta_{\vec{p},\vec{p}'} \delta_{\sigma,\sigma'} \ .
\end{gather}

The mass shift due to electromagnetic interactions is given by the Cottingham formula~\cite{Cottingham:1963zz}, which can be generalized easily to the case of finite volume by replacing the momentum integrals with the appropriate sums, yielding for the mass of the hadron $h$ in finite volume
\begin{gather}
m(L) =  m_0(L) - \frac{e^2}{4 m_0(L)} 
\frac{1}{L^3} \sum_{\vec{k} \in \Pi_-} \int \frac{dk_0}{2\pi} \frac{T_{\mu\mu}(k;L)}{k^2} \ ,
\label{eq:eu-cottingham-bare}
\\
T_{\mu\nu}(k;L) = \int d^4 x \, e^{- \imath k x} \langle h(\vec{0}) | \mathrm{T} \{ j_\mu(x) j_\nu(0) \} | h(\vec{0}) \rangle_c
\ ,
\label{eq:eu-compton}
\end{gather}
where the subscript $c$ stands for
\begin{gather}
\langle \psi | A | \psi' \rangle_c = \langle \psi | A | \psi' \rangle - \langle \psi | \psi' \rangle \, \langle 0 | A | 0 \rangle \ ,
\end{gather}
and $j_\mu(x)$ is the Heisenberg electric-current operator in Euclidean spacetime
\begin{gather}
j_\mu(x) = e^{x_0 H_0} e^{- \imath \vec{x} \vec{P}} j_\mu(0) e^{- x_0 H_0} e^{\imath \vec{x} \vec{P}} \ ,
\label{eq:eu-current}
\end{gather}
evolved with the QCD Hamiltonian $H_0$, and normalized in such a way that the electric charge is
\begin{gather}
Q(x_0) = - \imath \int d^3x \, j_0(x_0,\vec{x}) \ .
\label{eq:eu-charge}
\end{gather}
Notice that the electric current $j_\mu(x)$ is \C{}-odd and therefore it is also anti-periodic in space. It follows that the spatial momentum $\vec{k}$ in eq.~\eqref{eq:eu-compton} must belong to the set $\Pi_-$. In eq.~\eqref{eq:eu-compton} we have also used the shorthand notation
\begin{gather}
\langle h(\vec{0}) | A | h(\vec{0}) \rangle = \frac{1}{d_s} \sum_\sigma \langle h(\vec{0}),\sigma | A | h(\vec{0}),\sigma \rangle
\end{gather}
which is particularly useful as the spin will play no special role in the calculation of this appendix.

Because of the exponentially-raising operator $e^{x_0 H}$ in eq.~\eqref{eq:eu-current}, it is not obvious that the $x_0$-integral in eq.~\eqref{eq:eu-compton} converges. However the states that propagate in between the two currents have the same flavour numbers as the external state and they are odd under charge conjugation. In particular it follows that they cannot have zero momentum and they are therefore strictly heavier than the external state. In this situation the integral is shown to converge and an explicit calculation yields
\begin{flalign}
T_{\mu\nu}(k;L) = & M_{\mu\nu}(k;L) + M_{\nu\mu}(-k;L)
\label{eq:eu-compton-integr-0}
\ ,
\\
M_{\mu\nu}(k;L) = &
\int d^4 x \, \theta(x_0) e^{- \imath k x} \langle h(\vec{0}) | j_\mu(x) j_\nu(0) | h(\vec{0}) \rangle_c
\nonumber \\
= & \langle h(\vec{0}) | j_\mu(0) \frac{ L^3 \delta_{\vec{P},\vec{k}} }{ H_0-m_0(L) + \imath k_0 } j_\nu(0) | h(\vec{0}) \rangle
\nonumber \\
& \qquad -
2 m_0(L) L^3 \, \langle 0 | j_\mu(0) \frac{ L^3 \delta_{\vec{P},\vec{k}} }{ H_0 + \imath k_0 } j_\nu(0) | 0 \rangle
\ .
\label{eq:eu-compton-integr-1}
\end{flalign}
The retarded function $M_{\mu\nu}(k;L)$ is analytical for any complex $k_0$ except for the simple poles along the positive imaginary axis. In order to make contact with the original Minkowskian Cottingham formula~\cite{Cottingham:1963zz}, the reader can easily check that $T_{\mu\nu}(\imath k_0 - \epsilon,\vec{k};\infty)$ is the forward Compton amplitude for the scattering of a virtual photon with quadrimomentum $k$ from the hadron $h$ at rest.

Formula~\eqref{eq:eu-cottingham-bare} contains UV divergences which need to be subtracted. Notice that the electric current $j_\mu(x)$ does not require renormalization. Therefore the purely-QCD expectation value $\langle h(\vec{0}) | \mathrm{T} \{ j_\mu(x) j_\nu(0) \} | h(\vec{0}) \rangle_c$ appearing in eq.~\eqref{eq:eu-compton} is UV-finite for $x \neq 0$. The operator product expansion of $j_\mu(x) j_\nu(0)$ implies that $T_{\mu\nu}(k,L)$ vanishes like $k^{-2}$ (up to logarithms) at large $k$, which makes the integral in eq.~\eqref{eq:eu-cottingham-bare} logarithmically divergent. Following~\cite{Collins:1978hi} we renormalize the Euclidean Cottingham formula by introducing a Pauli-Villard regulator for the photon propagator and by adding appropriate counterterms:
\begin{gather}
m(L) =  m_0(L) +
\nonumber \\
\qquad +
\lim_{\Lambda \to \infty} \left\{
- \frac{e^2}{4 m_0(L)} 
\frac{1}{L^3} \sum_{\vec{k} \in \Pi_-} \int \frac{dk_0}{2\pi} \frac{T_{\mu\mu}(k;L) \Lambda^2}{k^2(k^2+\Lambda^2)} 
+
\langle h(\vec{0}) | C(\Lambda) | h(\vec{0}) \rangle_c
\right\}
\ .
\label{eq:eu-cottingham-ren}
\end{gather}
In this formula we assume that the regulator needed to define $T_{\mu\nu}(k,L)$ has been already removed. The counterterms have the form:
\begin{gather}
C(\Lambda) = c_\theta(\Lambda)  \theta_{\mu\mu}(0) + \sum_f c_f(\Lambda) m_f \bar{\psi}_f \psi_f(0) \ ,
\end{gather}
where $\theta_{\mu\nu}$ is the (Euclidean) energy-momentum tensor. Since \CQCDQED{} is a local theory, the coefficients $c(\Lambda)$ can be chosen to be $L$-independent by choosing renormalization conditions in infinite volume.

We are ready now to manipulate eq.~\eqref{eq:eu-cottingham-ren} in order to extract the power-like finite-volume corrections to the mass.

\medskip

\begin{lemma}
\label{prop:inf}
The QCD quantities appearing in eq.~\eqref{eq:eu-cottingham-ren} have only exponentially-suppressed finite-volume corrections,
\begin{gather}
m_0(L) - m_0(\infty) = \mathcal{O}(e^{-m_\pi L}) \ , \nonumber \\
T_{\mu\mu}(k;L) - T_{\mu\mu}(k;\infty) = \mathcal{O}(e^{-\frac{\sqrt{3}}{2} m_\pi L}) \ , \quad \text{for any real } k \neq 0 \ , \nonumber \\
\langle h(\vec{0}) | C(\Lambda) | h(\vec{0}) \rangle_c - \lim_{L \to \infty} \langle h(\vec{0}) | C(\Lambda) | h(\vec{0}) \rangle_c = \mathcal{O}(e^{-\frac{\sqrt{3}}{2} m_\pi L}) \ .
\end{gather}
\end{lemma}

\begin{proof}
A possible proof of this lemma, which we will not give here, can be obtained by using intermediate results and theorems in~\cite{Luscher:1985dn}, under the assumption that the leading finite-volume corrections are described by some arbitrarily-complicated Lagrangian massive field theory with small couplings, which effectively describe the dynamics of hadrons at large distance. All above quantities can be decomposed in terms of dressed propagators and (1PI) proper vertices with possible insertions and with two on-shell external legs. For all these quantities, the general conclusions of theorems 2.4, 2.5 and 2.6 hold, leading to a proof of the lemma. Notice that for a general theory, the finite-volume effects on the masses are $\mathcal{O}(e^{-\frac{\sqrt{3}}{2} m_{gap} L})$ however this is not the case in QCD~\cite{ProgressBook}. Some of the technology of~\cite{Luscher:1985dn} is adapted to the case of \Cstar{}-boundary conditions in app.~\ref{app:mixing}.
\end{proof}

Thanks to lemma~\ref{prop:inf}, we can write for the finite-volume correction to the mass
\begin{gather}
\Delta m(L) \equiv m(L) - m(\infty) 
\nonumber \\
= - \frac{e^2}{4 m_0} 
\lim_{\Lambda \to \infty}
\left\{ \frac{1}{L^3} \sum_{\vec{k} \in \Pi_-} - \int \frac{d^3k}{(2\pi)^3} \right\} \int \frac{dk_0}{2\pi} \frac{T_{\mu\mu}(k) \Lambda^2}{k^2(k^2+\Lambda^2)}  +
\nonumber \\
\qquad 
+ \mathcal{O}(e^{- m_\pi L}) + e^2 \mathcal{O}(e^{- \frac{\sqrt{3}}{2}  m_\pi L}) \ ,
\label{eq:deltam-1}
\end{gather}
where it is understood that we mean $L=\infty$ whenever we drop the $L$ dependence.

We introduce an arbitrary function $\eta(z)$ of a real variable $z$ with the properties: \textit{(a)} $\eta(z)$ is infinitely differentiable for any value of $z$, \textit{(b)} $\eta(z)=\eta(-z)$, \textit{(c)} $\eta(z) = 1$ for $|z| \le M^2/2$ for some arbitrary $M>0$, \textit{(d)} $\eta(z)=0$ for $|z| \ge M^2$. We rewrite the finite-volume correction to the mass as
\begin{gather}
\Delta m(L) = 
- \frac{e^2}{4 m_0} 
\left\{ \frac{1}{L^3} \sum_{\vec{k} \in \Pi_-} - \int \frac{d^3k}{(2\pi)^3} \right\}
\eta(\vec{k}^2) \int \frac{dk_0}{2\pi} \frac{T_{\mu\mu}(k)}{k^2} 
+ R(L) \ ,
\label{eq:deltam-2}
\end{gather}
where the reminder $R(L)$ is
\begin{gather}
R(L) = 
\lim_{\Lambda \to \infty}
\left\{ \frac{1}{L^3} \sum_{\vec{k} \in \Pi_-} - \int \frac{d^3k}{(2\pi)^3} \right\} \mathcal{I}_\Lambda(\vec{k})
+ \mathcal{O}(e^{-m_\pi L}) \ ,
\nonumber \\
\mathcal{I}_\Lambda(\vec{k}^2) = - \frac{e^2}{4 m_0} [1 - \eta(\vec{k}^2)] \int \frac{dk_0}{2\pi} \frac{T_{\mu\mu}(k) \Lambda^2}{k^2(k^2+\Lambda^2)} \ .
\end{gather}

\medskip

\begin{lemma}
\label{prop:smooth}
The infinite-volume Euclidean amplitude $T_{\mu\mu}(k)$ is infinitely differentiable for any $k \in \mathbb{R}^4/\{0\}$.
\end{lemma}

\begin{proof}
See app.~\ref{app:analyticity}.
\end{proof}

Thanks to lemma~\ref{prop:smooth} and to the fact that the factor $1 - \eta(\vec{k}^2)$ regularizes the singularity in $\vec{k}=0$, $\mathcal{I}_\Lambda(\vec{k}^2)$ is infinitely differentiable in $\vec{k} \in \mathbb{R}^3$. The reminder $R(L)$ is the difference between the integral of a smooth function and its approximation as a Riemann sum, which vanishes in the infinite-volume limit faster than any inverse power of $L$,
\begin{gather}
\lim_{L \to \infty} L^\omega R(L) = 0 \ , \quad \text{for all } \omega >0 \ .
\end{gather}

By plugging the eq.~\eqref{eq:eu-compton-integr-0} into eq.~\eqref{eq:deltam-2}, and by folding the $k_0$ integral we get
\begin{gather}
\Delta m(L) = 
- \frac{e^2}{2 m_0} 
\left\{ \frac{1}{L^3} \sum_{\vec{k} \in \Pi_-} - \int \frac{d^3k}{(2\pi)^3} \right\}
\eta(\vec{k}^2) \int \frac{dk_0}{2\pi} \frac{M_{\mu\mu}(k)}{k^2} 
+ R(L) \ .
\label{eq:deltam-2.5}
\end{gather}
The integrand is holomorphic in the half plane $\Im k_0 \le 0$ except for the single pole in $k_0 = -\imath |\vec{k}|$. Therefore the $k_0$ integral can be calculated as a Cauchy integral by closing the contour at infinity in the lower half plane yielding
\begin{gather}
\Delta m(L) = 
- \frac{e^2}{4 m_0} 
\left\{ \frac{1}{L^3} \sum_{\vec{k} \in \Pi_-} - \int \frac{d^3k}{(2\pi)^3} \right\}
\eta(\vec{k}^2) \frac{M_{\mu\mu}(-\imath |\vec{k}|,\vec{k})}{|\vec{k}|} 
+ R(L) \ .
\label{eq:deltam-3}
\end{gather}
We will see that the power-law finite-volume corrections come from the behaviour of the integrand in eq.~\eqref{eq:deltam-3} around $\vec{k}=\vec{0}$.

\medskip

\begin{lemma}
\label{prop:laurent}
Because of rotational symmetry, the on-shell retarded function $M_{\mu\mu}(-\imath |\vec{k}|,\vec{k})$ is a function of $\vec{k}$ only via $|\vec{k}|$. It can be decomposed as
\begin{gather}
M_{\mu\mu}(-\imath |\vec{k}|,\vec{k})\big|_{|\vec{k}| = \kappa} = \frac{\mathcal{M}_{-1}}{\kappa} + \mathcal{M}(\kappa) \ ,
\label{eq:regular-part}
\end{gather}
where the function $\mathcal{M}(\kappa)$ is analytical for complex values of $\kappa$ in a neighbourhood of zero, and
\begin{gather}
\mathcal{M}_{-1} = -2 m_0 q^2 \ ,
\end{gather}
where $q$ is the electric charge of the hadron $h$.
\end{lemma}

\begin{proof}
See app.~\ref{app:analyticity}.
\end{proof}

We can use now Poisson summation formula in order to express the discrete sum in eq.~\eqref{eq:deltam-3} over spatial momenta in terms of Fourier integrals. Since the momentum $\vec{k}$ belongs to the anti-periodic set $\hat{\Pi}_-$, an extra sign appears in Poisson summation formula
\begin{gather}
\sum_{\vec{k}\in \hat \Pi_-} f(\vec k) = 
\sum_{\vec{n}\in \Z_3} (-1)^{\langle \vec n\rangle} \int \frac{d^3k}{(2\pi)^3}\, 
e^{\imath \vec{n} \vec{k}}\, f(\vec k)\ ,
\label{eq:poissonmaintext}
\end{gather}
where $\langle \vec n\rangle$ has been defined in eq.~\eqref{eq:langlenrangle}. The term $\vec n=\vec 0$ in the previous expression corresponds to the infinite-volume integral. By plugging the definition~\eqref{eq:regular-part} into eq.~\eqref{eq:deltam-3}, and after calculating the angular integral in $\vec{k}$, we get
\begin{flalign}
\Delta m(L) = &
- \frac{e^2 \mathcal{M}_{-1}}{8 m_0 \pi^2 L}
\sum_{\vec{n}\in \Z_3 / \{ 0\}} \frac{(-1)^{\langle \vec n\rangle}}{|\vec{n}|}
\int_0^\infty d\kappa \ 
\eta(\kappa^2) \frac{\sin (\kappa |\vec{n}| L)}{\kappa} - 
\nonumber \\
& \qquad
- \frac{e^2}{8 \pi^2 m_0 L} 
\sum_{\vec{n}\in \Z_3 / \{ 0\}} \frac{(-1)^{\langle \vec n\rangle}}{|\vec{n}|}
\int_0^\infty d\kappa \ 
\eta(\kappa^2) \mathcal{M}(\kappa) \sin (\kappa |\vec{n}| L)
+ R(L) \ .
\label{eq:deltam-4}
\end{flalign}
We exploit the arbitrariness we have in choosing the function $\eta(\kappa^2)$ and assume that it has support in the analyticity domain of $\mathcal{M}(\kappa)$. Thanks to lemma~\ref{prop:laurent} the function $\eta(\kappa^2) \mathcal{M}(\kappa)$ is smooth for any $\kappa>0$, and has all right derivatives in $\kappa=0$. The expansion in powers of $1/L$ can be written in terms of the following generalized zeta function
\begin{gather}
\xi(s) = \sum_{\vec{n}\in \Z_3 / \{ 0\}} \frac{(-1)^{\langle \vec n\rangle}}{|\vec{n}|^s} \ ,
\label{eq:zeta-function}
\end{gather}
which is analytically extended to a meromorphic function in the whole complex plane, and holomorphic for $\Re s >0$.

The first integral in~\eqref{eq:deltam-4} can be understood by defining the function
\begin{gather}
\tilde{\eta}(x) = \int_{-\infty}^\infty \frac{d\kappa}{2\pi} \, \eta(\kappa^2) e^{i\kappa x} \ .
\end{gather}
Since $\eta(k^2)$ is a Schwartz function so is $\tilde{\eta}(x)$, and in particular it decays at infinity faster than any inverse power of $x$. The sum
\begin{gather}
\frac{2}{\pi} \sum_{\vec{n}\in \Z_3 / \{ 0\}}
\frac{(-1)^{\langle \vec n\rangle}}{|\vec{n}|}
\int_0^\infty d\kappa \ 
\eta(\kappa^2) \frac{\sin (\kappa |\vec{n}| L)}{\kappa}
=
\sum_{\vec{n}\in \Z_3 / \{ 0\}} \frac{(-1)^{\langle \vec n\rangle}}{|\vec{n}|}
\int_{-|\vec{n}|L}^{|\vec{n}|L} dx \ \tilde{\eta}(x)
\end{gather}
converges to $\xi(1)$, and the corrections decay faster than any power in $1/L$. The second integral in~\eqref{eq:deltam-4} has a Taylor expansion in $(|\vec{n}|L)^{-1}$ that can be extracted by using iteratively the identity
\begin{gather}
\int_0^\infty d\kappa \  f(\kappa) \sin (x \kappa)
=
\frac{1}{x} f(0)
-
\frac{1}{x^2} \int_0^\infty d\kappa \ f''(\kappa) \sin (x \kappa) \ .
\end{gather}
Putting everything together we obtain the desired expansion of the finite-volume corrections to the mass
\begin{flalign}
\Delta m(L) = &
- \frac{e^2 \mathcal{M}_{-1} }{16 m_0\pi L} \xi(1)
- \frac{e^2}{8 \pi^2 m_0} 
\sum_{\ell=0}^\infty
\frac{(-1)^{\ell}}{ L^{2+2\ell} }
\mathcal{M}_{2\ell} \xi(2+2\ell)
+ \dots \ ,
\label{eq:deltam-5}
\end{flalign}
where the dots stand for contributions that decay faster than any power of $1/L$, and $\mathcal{M}_{2\ell}$ is the $(2\ell)$-th derivative of $\mathcal{M}(\kappa)$ in $\kappa=0$. Notice that the on-shell forward Compton amplitude is given by
\begin{gather}
\mathcal{T}(\vec{k}^2) = T_{\mu\mu}(\imath |\vec{k}| - \epsilon,\vec{k}) = M_{\mu\mu}(\imath |\vec{k}|,\vec{k}) + M_{\mu\mu}(-\imath |\vec{k}|,\vec{k}) = \mathcal{M}(|\vec{k}|) + \mathcal{M}(-|\vec{k}|) \ ,
\end{gather}
therefore the coefficients $\mathcal{M}_{2\ell}$ are trivially related to the derivatives
\begin{gather}
\mathcal{M}_{2\ell} = \frac{(2\ell)!}{2(\ell!)} \frac{d^\ell}{d(\kappa^2)^\ell} \mathcal{T}(0) \equiv \frac{(2\ell)!}{2(\ell!)} \mathcal{T}_\ell
\end{gather}
of the on-shell forward Compton amplitude for the scattering of soft photons on the hadron $h$ at rest.

We also notice that the coefficients $\mathcal{M}_{-1}$ and $\mathcal{T}_0$ depend only on the mass and charge of the hadron, and not on its spin or internal structure. For the scattering amplitude we use the classical result~\cite{Low:1954kd,GellMann:1954kc} (also reviewed in section 13.5 of~\cite{WeinbergBook}):
\begin{gather}
\mathcal{T}_0 = \lim_{\vec{k} \to \vec{0}} T_{\mu\mu}(\imath |\vec{k}| - \epsilon,\vec{k}) = -4q^2 \ .
\end{gather}

We conclude this appendix by providing a representation of the zeta function $\xi(s)$ defined in eq.~\eqref{eq:zeta-function} which is useful for numerical calculation. We use the identity
\begin{gather}
\frac{1}{|\vec n|^{s/2}}
=
\frac{1}{\Gamma(s/2)}
\int_0^{u_\star} du\,
u^{\frac{s}{2}-1}
e^{-u \vec n^2} 
+
\frac{1}{\Gamma(s/2)}
\int_{u_\star}^\infty du\,
u^{\frac{s}{2}-1}
e^{-u \vec n^2} \ ,
\end{gather}
we plug it into eq.~\eqref{eq:zeta-function} and we use the Poisson summation formula
\begin{gather}
\sum_{\vec{n}\in \Z_3}
e^{-u \vec n^2 +\imath\pi\langle \vec n\rangle} 
=
\left(\frac{\pi}{u}\right)^{\frac{3}{2}}\,
\sum_{\vec k\in \hat \Pi_-} e^{-\frac{\vec k^2}{4u}} \ ,
\end{gather}
only in the integral over $u \in [0,u^\star]$. At this point all integrals can be calculated explicitly in terms of the upper incomplete gamma functions:
\begin{align}
\xi(s) = 
\frac{1}{\Gamma(s/2)}\Bigg\{
-\frac{2 u_\star^{s/2}}{s}
+
\frac{\pi^{3/2}}{2^{s-3}}
\sum_{\vec{k}\in \hat \Pi_-}
(\vec k^2)^{\frac{s-3}{2}}\,&\Gamma\left(\frac{3-s}{2},\frac{\vec k^2}{4u_\star}\right)
+ \nonumber \\
&+
\sum_{\vec{n}\neq 0} \frac{(-1)^{\langle \vec n\rangle}}{|\vec n|^{s/2}}\, 
\Gamma\left(s/2, u_\star \vec n^2 \right)
\Bigg\} \ .
\label{eq:zeta-function-final}
\end{align}
The upper incomplete gamma function $\Gamma(\tau,z)$ is defined for all complex values of $\tau$ except non-positive integers, and it decays exponentially as $|z| \to \infty$. Therefore the infinite sums in the previous formula are rapidly convergent. Also this representation is valid for all values of $s$ needed in the mass formula. The splitting variable $u_\star>0$ is completely arbitrary and can be used to check the result of numerical calculation.

\subsection{Analyticity properties}
\label{app:analyticity}

\textit{In this subsection we work in Minkowski spacetime with metric $g=\text{diag}(1,-1,-1,-1)$.} We also set $L=\infty$. We introduce the Minkowskian electric current:
\begin{gather}
J_\mu(x) = e^{\imath(x_0 H_0-\vec{x} \vec{P})} J_\mu(0) e^{-\imath(x_0 H_0-\vec{x} \vec{P})} \ ,
\end{gather}
which is related to the Euclidean one introduced in eq.~\eqref{eq:eu-compton} via
\begin{gather}
J^0(0) = -\imath j_0(0) \ , \quad J^k(0) = j_k(0) \ .
\end{gather}
While in finite volume, because of \Cstar{}-boundary conditions, eigenstates of the momentum are also eigenstates of the charge-conjugation operator, this is not necessarily true in infinite volume. We perform a change of basis which does not affect the quantities we are interested in, and we choose to work with simultaneous eigenstates of energy, momentum and electric charge.

We consider the retarded two-point function in the forward limit
\begin{gather}
W_+(k) = \imath  \, \lim_{\vec{p} \to \vec{0}} \int d^4 x \, e^{\imath k  x} \theta(x_0) \langle h(\vec{p}) | J_\mu(x) J^\mu(0) | h(\vec{0}) \rangle_c
\label{eq:Wplus-1}
\ ,
\end{gather}
which is related to the function $M_{\mu\mu}(k)$ introduced in eq.~\eqref{eq:eu-compton-integr-1} in infinite volume through a Wick rotation
\begin{gather}
M_{\mu\mu}(k_0,\vec{k}) = - W_+(-\imath k_0,\vec{k}) \ .
\end{gather}
The subtraction of the disconnected part in eq.~\eqref{eq:Wplus-1} can be expanded:
\begin{gather}
\int d^4 x \, e^{\imath k  x} \theta(x_0) \langle h(\vec{p}) | J_\mu(x) J^\mu(0) | h(\vec{0}) \rangle_c \nonumber \\
= \int d^4 x \, e^{\imath k  x} \theta(x_0) \bigg\{ \langle h(\vec{p}) | J_\mu(x) J^\mu(0) | h(\vec{0}) \rangle - 2 E(\vec{p}) (2\pi)^3 \delta^3(\vec{p}) \langle 0 | J_\mu(x) J^\mu(0) | 0 \rangle \bigg\}
\ .
\end{gather}
Notice that in this formula we cannot just take $\vec{p}=\vec{0}$ because the disconnected contribution gives a geometrical divergence proportional to $\delta^3(\vec{0})$. Therefore the limit in eq.~\eqref{eq:Wplus-1} is essential in order to define properly the subtraction. However notice that for any $\vec{p}\neq \vec{0}$ the delta function vanishes exactly and it does not contribute to the limit ($\lim_{\vec{p}\to\vec{0}} \delta^3(\vec{p}) = 0$), allowing us to write equivalently
\begin{gather}
W_+(k) = \imath  \, \lim_{\vec{p} \to \vec{0}} \int d^4 x \, e^{\imath k  x} \theta(x_0) \langle h(\vec{p}) | J_\mu(x) J^\mu(0) | h(\vec{0}) \rangle
\label{eq:Wplus-1.5}
\ ,
\end{gather}
It is possible to prove that this limit is finite, which we will assume in the remaining of this appendix.\footnote{
By means of the LSZ reduction formula, the function $C_+(k)$ defined in~\eqref{eq:Cplus} can be expressed as linear combinations of the \textit{reduced Green's functions} defined in eq. (16.52), chapter 16 of~\cite{BogoBook}. On the other hand, as we will notice later on, $W_+(k)$ is uniquely determined by $C_+(k)$. The finiteness of the limit $\vec{p}\to\vec{0}$ in eq.~\eqref{eq:Cplus} and consequently in eq.~\eqref{eq:Wplus-1.5} derives from the analyticity properties of the \textit{reduced Green's functions} stated in Theorem 16.8, chapter 16 of~\cite{BogoBook}. The reader should notice that the connected part and the limit are systematically dropped in the classical literature, e.g.~\cite{Cottingham:1963zz}.
}

By calculating the coordinate integral in eq.~\eqref{eq:Wplus-1.5} one gets
\begin{gather}
W_+(k) = \lim_{\vec{p} \to \vec{0}} \langle h(\vec{p}) | J_\mu(0) \frac{(2\pi)^3 \delta^3(\vec{P}-\vec{p}-\vec{k})}{H_0- E(\vec{p}) - k_0 - \imath \epsilon} J^\mu(0) | h(\vec{0}) \rangle \ ,
\label{eq:Wplus-2}
\end{gather}
where we have introduced the energy of the external states
\begin{gather}
E(\vec{p}) = \sqrt{m_0^2 + \vec{p}^2} \ .
\end{gather}
Notice that, since the external state corresponds to a stable hadron, the states propagating in between the two currents in eq.~\eqref{eq:Wplus-2} have energy not smaller than $m_0$. Therefore the retarded function $W_+(k)$ has poles only for non-negative values of $k_0$.

It is useful to separate the single particle component from the continuous part of the spectrum:
\begin{gather}
W_+(k) = \frac{Z_{\text{1P}}(\vec{k}^2)}{E(\vec{k}) -m_0 - k_0 + \imath \epsilon} + Z_{\text{MP}}(k_0,\vec{k}^2) \ , \\
Z_{\text{1P}}(\vec{k}^2) = \frac{1}{d_s} \sum_{\mu,\sigma,\sigma'} \frac{g^{\mu\mu}}{2E(\vec{k})} | \langle h(\vec{0}),\sigma | J_\mu(0) | h(\vec{k}) , \sigma' \rangle|^2 \ .
\end{gather}

We want to study the analyticity properties of $Z_{\text{1P}}(\vec{k}^2)$ and $Z_{\text{MP}}(k_0,\vec{k}^2)$ in the spatial momentum $\vec{k}$ (the analyticity properties in $k_0$ are obvious from the spectral decomposition~\eqref{eq:Wplus-2}), which we summarize here:
\begin{enumerate}[noitemsep,topsep=-.7em]
\item $Z_{\text{1P}}(\vec{k}^2)$ is analytical for any real value of $\vec{k}$, and can be analytically continued to a complex neighbourhood of $\vec{k}^2 = 0$;
\item The Euclidean function $Z_{\text{MP}}(-\imath k_0,\vec{k}^2)$ is analytical for any real value of $\vec{k}$;
\item The on-shell function $Z_{\text{MP}}(-|\vec{k}|,\vec{k}^2)$ is analytical for any real value of $\vec{k}$;
\item The on-shell function $Z_{\text{MP}}(|\vec{k}|,\vec{k}^2)$ is analytical for real values of $\vec{k}$ in a neighbourhood of $\vec{k}=\vec{0}$, and can be analytically continued to a complex neighbourhood of $|\vec{k}|=0$.
\end{enumerate}
From these properties it follows that:
\begin{enumerate}[noitemsep,topsep=-.7em]
\item The off-shell Euclidean Compton amplitude
\begin{gather}
T(k_0,\vec{k}) = M_{\mu\mu}(k_0,\vec{k}) + M_{\mu\mu}(-k_0,-\vec{k})
\nonumber \\=
- \frac{2 [E(\vec{k}) -m_0] Z_{\text{1P}}(\vec{k}^2)}{[E(\vec{k}) -m_0]^2 + k_0^2} - Z_{\text{MP}}(\imath k_0,\vec{k}^2) - Z_{\text{MP}}(-\imath k_0,\vec{k}^2)
\end{gather}
is analytical for any real value of $\vec{k}\neq \vec{0}$ (lemma \ref{prop:smooth}).

\item The on-shell quantity
\begin{gather}
M_{\mu\mu}(-i|\vec{k}|,\vec{k}) = - \frac{Z_{\text{1P}}(\vec{k}^2)}{E(\vec{k}) -m_0 + |\vec{k}|} - Z_{\text{MP}}(-|\vec{k}|,\vec{k}^2) \ ,
\end{gather}
as a function of $|\vec{k}|$, admits a meromorphic extension to a complex neighbourhood of $|\vec{k}|=0$. In particular it admits a Laurent series in $|\vec{k}|=0$, the first term being:
\begin{gather}
M_{\mu\mu}(-i|\vec{k}|,\vec{k}) = -\frac{Z_{\text{1P}}(0)}{|\vec{k}|} + \mathcal{O}(|\vec{k}|^0) \ .
\end{gather}
We will show that, eq.~\eqref{eq:first-universal-term}, $Z_{\text{1P}}(0) = 2m_0q^2$ (lemma \ref{prop:laurent}).

\item The on-shell Compton amplitude
\begin{gather}
T(\imath |\vec{k}| -\epsilon,\vec{k}) =
- \frac{2 [E(\vec{k}) -m_0] Z_{\text{1P}}(\vec{k}^2)}{[E(\vec{k}) -m_0]^2 - \vec{k}^2 - \imath \epsilon} - Z_{\text{MP}}(|\vec{k}|,\vec{k}^2) - Z_{\text{MP}}(-|\vec{k}|,\vec{k}^2)
\end{gather}
is an analytic function of $\vec{k}^2$ in a complex neighbourhood of $\vec{k}^2=0$. This follows from the analyticity properties discussed above, from the fact that the odd powers in $|\vec{k}|$ generated by the expansion of $Z_{\text{MP}}(\pm|\vec{k}|,\vec{k}^2)$ cancel out, and from the fact that
\begin{gather}
\lim_{\vec{k} \to \vec{0}} \frac{2 [E(\vec{k}) -m_0]}{[E(\vec{k}) -m_0]^2 - \vec{k}^2 - \imath \epsilon} = - \frac{1}{m_0} \ .
\end{gather}

\end{enumerate}

The full analyticity properties of $W_+(k)$, and consequently of $Z_{\text{1P}}(\vec{k}^2)$ and $Z_{\text{MP}}(k_0,\vec{k}^2)$, can be derived by the analyticity properties of four-point reduced Green's functions discussed in chapter 16 of~\cite{BogoBook}. However we provide here a hopefully more digestible proof of the particular properties we are interested in, based on the Jost-Lehmann-Dyson representation of the expectation values of certain retarded commutators. We also point out that the same analyticity properties we are interested in can also be obtained by assuming an effective theory describing hadrons and by using results and methods discussed in section 2.4 of ref.~\cite{Luscher:1985dn} and in appendix~\ref{app:mixing}.

\paragraph{Analysis of $Z_{\text{1P}}(\vec{k}^2)$.}

We notice first that $Z_{\text{1P}}(\vec{k}^2)$ can be extracted also from the retarded commutator
\begin{flalign}
C_+(k) = & \imath \, \lim_{\vec{p} \to \vec{0}} \int d^4 x \, e^{\imath k  x} \theta(x_0) \langle h(\vec{p}) | \, [ J_\mu(x) , J^\mu(0) ] \, | h(\vec{0}) \rangle \nonumber \\
= & W_+(k) + W_-(k) \ ,
\label{eq:Cplus}
\end{flalign}
where $W_-(k)$ is a functions with poles only for negative value of $k_0$. Therefore the following reduction formula holds
\begin{gather}
\lim_{k_0 \to E(\vec{k}) - m_0} [E(\vec{k}) -m_0 - k_0] C_+(k) \nonumber \\
\qquad\qquad
= \lim_{k_0 \to E(\vec{k}) - m_0} [E(\vec{k}) -m_0 - k_0] W_+(k) = Z_{\text{1P}}(\vec{k}^2) \ .
\label{eq:reduction-1}
\end{gather}
Then we extract the single-hadron pole from both orderings of the retarded commutator by means of the following trick. We introduce the auxiliary retarded commutator
\begin{gather}
\tilde{C}_+(k) = \imath \, \lim_{\vec{p} \to \vec{0}} \int d^4 x \, e^{\imath k  x} \theta(x_0) \langle h(\vec{p}) | [ \bar{J}_\mu(x) , \bar{J}^\mu(0) ] | h(\vec{0}) \rangle \ ,
\label{eq:Ctildeplus}
\\
\bar{J}_\mu(x) = (-\Box+ 2\imath m_0 \partial_0) J^\mu(x) \ .
\end{gather}
The relation between this retarded commutator and the original one is obtained through integration by parts of the differential operator $(-\Box+ 2\imath m_0 \partial_0)$
\begin{gather}
(k^2 + 2 m_0 k_0) (k^2 - 2 m_0 k_0) C_+(k) = \tilde{C}_+(k) + \tilde{P}(k) \ .
\label{eq:separation-tilde}
\end{gather}
The boundary term has the form
\begin{gather}
\tilde{P}(k) = \imath \, \lim_{\vec{p} \to \vec{0}} \int d^4 x \, \delta(x_0) e^{\imath k  x} \, \langle h(\vec{p}) | \, [ \mathcal{D} J_\mu(x) , J^\mu(0) ] \, | h(\vec{0}) \rangle
\end{gather}
where $\mathcal{D}$ is some local differential operator. The integrand of $\tilde{P}(k)$ involves only commutators of local operators at equal time, which are linear combinations of delta functions and their derivatives. Therefore $\tilde{P}(k)$ is a polynomial in the quadrimomentum $k$. In terms of the auxiliary retarded commutator, the reduction formala reads
\begin{gather}
Z_{\text{1P}}(\vec{k}^2) =  \frac{1}{8 m_0 E(\vec{k}) [E(\vec{k})-m_0] } \lim_{k_0 \to E(\vec{k}) - m_0} [ \tilde{C}_+(k) + \tilde{P}(k) ] \ .
\label{eq:reduction-2}
\end{gather}

The analyticity properties of the modified retarded commutator can be exposed by means of the Jost-Lehmann-Dyson (JLD) representation~\cite{JLpaper,Dyson:1997gw},
\begin{gather}
\tilde{C}_+(k) =
\int_{\tilde{\mathcal{S}}} \frac{d^4u d\lambda^2 \ \tilde{\rho}(u,\lambda^2)}{(k-u)^2 - \lambda^2 + \imath (k_0-u_0) \epsilon}
\ ,
\label{eq:V:JLD-tilde}
\end{gather}
where the JLD spectral function $\tilde{\rho}(u,\lambda^2)$ is uniquely determined by the retarded commutator. The integration domain $\tilde{\mathcal{S}}$ encodes all known information about the spectrum, and can be represented as the set of 5-tuples $(u,\lambda^2)$ such that
\begin{gather}
\begin{dcases}
\sqrt{(\vec{u}-\vec{k})^2+\lambda^2} + u_0 \ge \sqrt{M_1^2 + \vec{k}^2} - m \\
\sqrt{(\vec{u}-\vec{k})^2+\lambda^2} - u_0 \ge \sqrt{M_2^2 + \vec{k}^2} - m
\end{dcases}
\ ,\label{eq:V:JLD-S-tilde}
\end{gather}
for any value of the momentum $\vec{k}$. The masses $M_1$ and $M_2$ are determined in the following way. Consider the commutator
\begin{gather}
\lim_{\vec{p} \to \vec{0}} \int d^4 x \, e^{\imath k  x} \langle h(\vec{p}) | [ \bar{J}_\mu(x) , \bar{J}^\mu(0) ] | h(\vec{0}) \rangle \nonumber \\
=
\lim_{\vec{p} \to \vec{0}} \langle h(\vec{p}) | \bar{J}_\mu(0) (2\pi)^4 [ \delta(P - p - k) - \delta(P - p + k) ]_{p=(E(\vec{p}),\vec{p})} \bar{J}^\mu(0) ] | h(\vec{0}) \rangle \ ,
\end{gather}
with $P=(H_0,\vec{P})$. The two delta functions come from the two different orderings of the currents in the commutator. $M_1$ and $M_2$ are the masses of the lightest states propagating in between the two currents in the first and second ordering respectively. Had we considered the original current $J_\mu(x)$, the lightest state would have been the hadron $h$ itself. However it is easy to check that the insertion of the operators $(-\Box+ 2\imath m_0 \partial_0)$ kills the contribution of the single-hadron states in the above commutator, therefore
\begin{gather}
M_1 = M_2 = m+\Delta
\end{gather}
where $\Delta > 0$ is some mass gap (if no bound states exist $\Delta=2 m_\pi$). 

The relevant limit for the reduction formula~\eqref{eq:reduction-2} is
\begin{flalign}
f(\vec{k}^2) = & \lim_{k_0 \to E(\vec{k}) - m_0} [ \tilde{C}_+(k) + \tilde{P}(k) ] \nonumber \\
= & P(E(\vec{k}),\vec{k}) +
\int_{\tilde{\mathcal{S}}} \frac{d^4u d\lambda^2 \ \tilde{\rho}(u,\lambda^2)}{[E(\vec{k})- m_0-u_0]^2 - (\vec{k}-\vec{u})^2 - \lambda^2} \ .
\label{eq:V:ef}
\end{flalign}
The denominator vanishes only if
\begin{gather}
u_0 \pm \sqrt{(\vec{k}-\vec{u})^2 + \lambda^2} = E(\vec{k}) - m_0 \ ,
\end{gather}
which is satisfied for no real value of $\vec{k}$ if $(u,\lambda^2)$ is in the domain $\tilde{\mathcal{S}}$.

As pointed out in~\cite{Dyson:1997gw}, if $(u,\lambda^2)$ is in the domain $\tilde{\mathcal{S}}$ then necessarily
\begin{gather}
|u_0| + |\vec{u}| \le m_0 \ .
\end{gather}
Thanks to this, it is easy to show that two positive constants $\kappa$ and $\alpha$ exist such that the denominator in eq.~\eqref{eq:V:ef} is limited from below by
\begin{gather}
| [E(\vec{k})- m_0-u_0]^2 - (\vec{k}-\vec{u})^2 - \lambda^2 | \ge | u_0^2 - \vec{u}^2 - \lambda^2 | - \alpha |\vec{k}| \ge \Delta^2 - \alpha |\vec{k}|
\end{gather}
for any $(u,\lambda^2) \in \tilde{\mathcal{S}}$ and for any complex $\vec{k}$ such that $|\vec{k}|<\kappa$. In the last step we have used eqs.~\eqref{eq:V:JLD-S-tilde} for $\vec{k}=\vec{0}$. From the above bound it is clear that if $\kappa$ is small enough, then the denominator never vanishes. Therefore $f(\vec{k}^2)$ can be continued by analyticity to a complex neighbourhood of $\vec{k}^2 = 0$.

From eq.~\eqref{eq:reduction-2} it might seem that $Z_{\text{1P}}(\vec{k}^2)$ has a singularity in $\vec{k} \to \vec{0}$. However this limit is fixed by symmetries:
\begin{gather}
\lim_{\vec{k} \to \vec{0}} Z_{\text{1P}}(\vec{k}^2) = \lim_{\vec{k} \to \vec{0}}  \frac{1}{d_s} \sum_{\mu,\sigma,\sigma'} \frac{g^{\mu\mu}}{2E(\vec{k})} | \langle h(\vec{0}),\sigma | J_\mu(0) | h(\vec{k}) , \sigma' \rangle|^2 = 2m_0 q^2 \ ,
\label{eq:first-universal-term}
\end{gather}
where $q$ is the electric charge of the hadron $h$. This relation implies that $f(0) = 0$ and $Z_{\text{1P}}(\vec{k}^2)$ is an analytic function for any real value of $\vec{k}^2$, and for complex values of $\vec{k}^2$ in a neighbourhood of zero.

\paragraph{Analysis of $Z_{\text{MP}}(k_0,\vec{k}^2)$.}

$Z_{\text{MP}}(k_0,\vec{k}^2)$ is obtained by selecting all poles in $W_+(k)$, or equivalently in $C_+(k)$, with $\Re k_0 > E(\vec{k})-m_0$. In this case we find more convenient to write $C_+(k)$ in terms of the auxiliary retarded commutator
\begin{gather}
\hat{C}_+(k) = \imath \, \lim_{\vec{p} \to \vec{0}} \sum_\sigma \int d^4 x \, e^{\imath k  x} \theta(x_0) \langle h(\vec{p}) | [ \bar{J}_\mu(x) , J^\mu(0) ] | h(\vec{0}) \rangle \ .
\label{eq:Chatplus}
\end{gather}
In complete analogy to eq.~\eqref{eq:separation-tilde}, the original retarded commutator can be written in terms of the auxiliary one as
\begin{gather}
(k^2 + 2 m_0 k_0) C_+(k) = \hat{C}_+(k) + \hat{P}(k) \ ,
\label{eq:separation-hat}
\end{gather}
where $\hat{P}(k)$ is a polynomial in the quadrimomentum $k$. The above equation can be inverted by noticing that all poles of $C_+(x)$ have negative imaginary part. We introduce a JLD representation for the retarded commutator $\hat{C}_+(k)$ and we get
\begin{gather}
C_+(k) = \frac{1}{k^2 + 2 m_0 k_0 + \imath (k_0 + m_0) \epsilon} \left\{ \hat{P}(k) + \int_{\hat{\mathcal{S}}} \frac{d^4u d\lambda^2 \ \hat{\rho}(u,\lambda^2)}{(k-u)^2 - \lambda^2 + \imath (k_0-u_0) \epsilon} \right\}
\ .
\label{eq:V:JLD-hat}
\end{gather}
The integration domain $\hat{\mathcal{S}}$ is the set of 5-tuples $(u,\lambda^2)$ such that
\begin{gather}
\begin{dcases}
\sqrt{(\vec{u}-\vec{k})^2+\lambda^2} + u_0 \ge \sqrt{(m+\Delta)^2 + \vec{k}^2} - m \\
\sqrt{(\vec{u}-\vec{k})^2+\lambda^2} - u_0 \ge \sqrt{m^2 + \vec{k}^2} - m
\end{dcases}
\ ,\label{eq:V:JLD-S-hat}
\end{gather}
for any value of the momentum $\vec{k}$. The denominator outside of the integral in eq.~\eqref{eq:V:JLD-hat} has poles for $\Re k_0 \le E(\vec{k})-m_0$ which do not contribute to $Z_{\text{MP}}(k_0,\vec{k}^2)$. The integrand can be decomposed in partial fractions
\begin{gather}
\frac{1}{(k-u)^2 - \lambda^2 + \imath (k_0-u_0) \epsilon} =
\frac{1}{2X} \left( \frac{1}{k_0 - u_0 - X + \imath \epsilon} - \frac{1}{k_0 - u_0 + X + \imath \epsilon} \right) \ ,
\\
X = \sqrt{(\vec{k}-\vec{u})^2 + \lambda^2} \ ,
\end{gather}
and only the first one contributes with a pole to $Z_{\text{MP}}(k_0,\vec{k}^2)$. By calculating the residue at this pole we get
\begin{gather}
Z_{\text{MP}}(k_0,\vec{k}^2)
=
\int_{\hat{\mathcal{S}}} d^4u d\lambda^2 \ 
\frac{\hat{\rho}(u,\lambda^2)}{2 X [(u_0 + X + m_0)^2 - E(\vec{k})^2]} \ 
\frac{1}{k_0 - u_0 - X + \imath \epsilon} \ .
\end{gather}
Using the definition of the domain $\hat{\mathcal{S}}$ it is straightforward to check that the denominator $2 X [(u_0 + X + m_0)^2 - E(\vec{k})^2]$ never vanishes for any real value of $\vec{k}$ and for any value of $(u,\lambda^2) \in \hat{\mathcal{S}}$. We derive some particular properties.

The Wick-rotated function
\begin{gather}
Z_{\text{MP}}(-\imath k_0,\vec{k}^2)
=
\int_{\hat{\mathcal{S}}} d^4u d\lambda^2 \ 
\frac{\hat{\rho}(u,\lambda^2)}{2 X [(u_0 + X + m_0)^2 - E(\vec{k})^2]} \ 
\frac{1}{-\imath k_0 - u_0 - X} \ .
\end{gather}
is analytical for any real value of $k$, since the two denominators never vanish (as $u_0+X>0$).

The on-shell function
\begin{gather}
Z_{\text{MP}}(-|\vec{k}|,\vec{k}^2)
=
\int_{\hat{\mathcal{S}}} d^4u d\lambda^2 \ 
\frac{\hat{\rho}(u,\lambda^2)}{2 X [(u_0 + X + m_0)^2 - E(\vec{k})^2]} \ 
\frac{1}{-|\vec{k}| - u_0 - X} \ .
\end{gather}
is analytical for any real value of $\vec{k}$, since the two denominators never vanish (as $|\vec{k}| + u_0 + X \ge u_0+X>0$). Moreover, as a function of $|\vec{k}|$, $Z_{\text{MP}}(-|\vec{k}|,\vec{k}^2)$ can be analytically continued to a complex neighbourhood of $|\vec{k}|=0$.

The on-shell function
\begin{gather}
Z_{\text{MP}}(|\vec{k}|,\vec{k}^2)
=
\int_{\hat{\mathcal{S}}} d^4u d\lambda^2 \ 
\frac{\hat{\rho}(u,\lambda^2)}{2 X [(u_0 + X + m_0)^2 - E(\vec{k})^2]} \ 
\frac{1}{|\vec{k}| - u_0 - X + \imath \epsilon} \ .
\end{gather}
is analytical for real values of $\vec{k}$ such that
\begin{gather}
|\vec{k}| < \Delta \ ,
\end{gather}
as in this range the denominator $|\vec{k}| - u_0 - X$ can be shown not to vanish for any value of $(u,\lambda^2) \in \hat{\mathcal{S}}$. As a function of $|\vec{k}|$, $Z_{\text{MP}}(|\vec{k}|,\vec{k}^2)$ can be analytically continued to a complex neighbourhood of $|\vec{k}|=0$.

\section{Classical vacua of compact \CQED{}}
\label{app:compact}
We consider an abelian gauge field on a lattice with \Cstar{} boundary conditions along the directions included in the set $\mathcal{C}$
\begin{gather}
U(x+\hat L_\mu,\rho) = \begin{cases}
U(x,\rho) & \text{if } \mu \not\in \mathcal{C} \\
U(x,\rho)^* \quad & \text{if } \mu \in \mathcal{C}
\end{cases} \ ,
\end{gather}
where the coordinates are integer numbers in the range
\begin{gather}
0 \le x_\mu \le L_\mu-1 \ . \label{appA:domain}
\end{gather}
We assume direction $\mu=3$ \Cstar{}-periodic, and direction $\mu=0$ periodic.

We want to characterise all gauge-field configurations corresponding to absolute minima of the Wilson action. In terms of the plaquette $P(x,\mu,\nu)$ the minimum condition reads
\begin{gather}
P(x,\mu,\nu) = 1 \ .
\label{appA:minimum_condition}
\end{gather}
We can always gauge-transform  to axial gauge along a given direction $\mu$, i.e. to a gauge in which all the link variables $U(x,\mu)$ are equal to one except the ones on the hyperplane $\pi_\mu$ defined by the equation
\begin{gather}
\pi_\mu \ : \ x_\mu=L_\mu-1 \ .
\end{gather}
Because of condition~\eqref{appA:minimum_condition} it is easy to show that we can gauge-transform to simultaneous axial gauge for all directions. We will refer to those link variables that are different from unity as \textit{active} link variables (see figure~\ref{fig:app2_lattice}).

\begin{figure}[tb]
\centering

\begin{tikzpicture}

\draw[thick,->] (-.5,-.7)++(1,1) -- ++(0,1.5) node [anchor=south,rotate=90] {$\mu=0$};
\draw[thick,->] (-.7,-.5)++(1,1) -- ++(1.5,0) node [anchor=north] {$\mu=3$};

\fill[red!20] (6,6) rectangle ++(1,1);
\fill[blue!20] (6,2) rectangle ++(1,1);

\draw[black!40,step=1] (0,0)++(1,1) grid (6,6);

\foreach \x in {1,2,...,6}
   \draw[thick,postaction={decorate},decoration={markings,mark=at position .65 with {\arrow[scale=1.4]{latex}}}] (\x,6) -> (\x,7);
\foreach \x in {1,2,...,6}
   \draw[thick,postaction={decorate},decoration={markings,mark=at position .65 with {\arrow[scale=1.4]{latex}}}] (6,\x) -> (7,\x);

\draw[thick,dotted,postaction={decorate},decoration={markings,mark=at position .65 with {\arrow[scale=1.4]{latex}}}] (7,6) -> (7,7);
\draw[thick,dotted,postaction={decorate},decoration={markings,mark=at position .65 with {\arrow[scale=1.4]{latex}}}] (6,7) -> (7,7);

\path (-0.2,6)++(1,0) node [anchor=east] {$\pi_0 \to $};
\path (6,-0.2)++(0,1) node [anchor=east,rotate=90] {$\pi_3 \to $};

\path (6.5,6) node[anchor=north] {$W_3$};
\path (6,6.5) node[anchor=east] {$W_0$};
\path (6.5,7) node[anchor=south] {$W_3$};
\path (7,6.5) node[anchor=west] {$W_0^*$};

\end{tikzpicture}

\caption{\label{fig:app2_lattice} A two-dimensional representation of the problem discussed in this appendix. Direction $\mu=0$ is periodic, while direction $\mu=3$ is \Cstar{}-periodic. In simultaneous axial gauge, the only links that are different from unity are the ones represented with a thick line and an arrow (\textit{active} link variables). The condition that the blue plaquette be equal to one implies that the two active link variables in the plaquette are equal. Nontrivial constraints come from minimum condition for the red plaquette at the intersection of the $\pi_0$ and $\pi_3$ hyperplanes.
}
\end{figure}
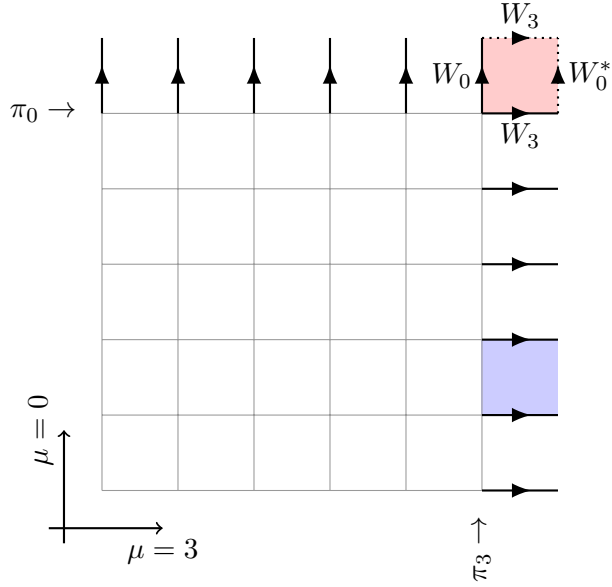

Plaquettes at the intersection of two distinct $\pi$-planes involve four active link variables, all other plaquettes on the $\pi$-planes involve two parallel link variables. Constraint \eqref{appA:minimum_condition} on the latter ones, together with the fact that unity is left unchanged by the boundary conditions, implies
\begin{gather}
U(x,\mu) = U(x+\hat{\nu},\mu) \ , \text{ for any } x \in \pi_\mu, \ \nu \neq \mu \ .
\label{appA:active-links-1}
\end{gather}
Using this equation recursively we get that, given some direction $\mu$, all active link variables along $\mu$ are equal to each other. We will define
\begin{gather}
W_\mu = U(x,\mu) \ , \text{ for any } x \in \pi_\mu \ .
\end{gather}

We use now the minimum condition~\eqref{appA:minimum_condition} for plaquettes at the intersection of two distinct $\pi$-planes. Let us consider first the plaquette in some point $x \in \pi_\mu \cap \pi_3$ where $\mu$ is a periodic direction (see figure~\ref{fig:app2_lattice}):
\begin{gather}
1 = P(x,\mu,3) = W_\mu W_3 W_\mu W_3^{-1} \ ,
\end{gather}
which implies
\begin{gather}
W_\mu = \pm 1 \ , \text{ if } \mu \not\in \mathcal{C} \ .
\end{gather}
If $\mu \neq 3$ is a \Cstar{} direction we get instead
\begin{gather}
1 = P(x,\mu,3) = W_\mu W_3^{-1} W_\mu W_3^{-1} \ ,
\end{gather}
which implies
\begin{gather}
W_\mu = \pm W_3 \ , \text{ if } \mu \in \mathcal{C} \ .
\end{gather}

Finally we show that $W_3$ can be set to $1$ with a gauge transformation. Let $w$ be a complex number such that $W_3 = w^{-2}$ and we define the gauge transformation
\begin{gather}
\Lambda(x) = w \ , \text{ for } 0 \le x_\mu \le L_\mu-1 \ ,
\end{gather}
and extended outside the above domain by means of the boundary conditions
\begin{gather}
\Lambda(x+\hat L_\mu) = \begin{cases}
\Lambda(x) & \text{if } \mu \not\in \mathcal{C} \\
\Lambda(x)^* \quad & \text{if } \mu \in \mathcal{C}
\end{cases} \ .
\end{gather}
First notice that this gauge transformation preserves the gauge-field boundary conditions and the axial gauge. All active link variables along periodic directions are left unchanged under this gauge transformation. If $\mu$ is a \Cstar{}-direction, the active link variable along $\mu$ transforms like
\begin{gather}
W_\mu \to w W_\mu w = W_3^{-1} W_\mu = \pm 1 \ ,
\end{gather}
with the particular case of
\begin{gather}
W_3 \to w W_3 w = W_3^{-1} W_3 = 1 \ .
\end{gather}
This concludes the proof of part 1 of the following proposition.

\begin{proposition}
\label{appA:prop1}
Let $U(x,\mu)$ be a gauge configuration that minimizes the Wilson action.
\begin{enumerate}[noitemsep,topsep=-.7em]
\item A vector $z$ satisfying the conditions
\begin{gather}
\label{appA:prop1_eq4}
z_3=1 \ , \qquad z_\mu^2=1 \ .
\end{gather}
exists such that $U(x,\mu)$ is gauge-equivalent to the gauge configuration
\begin{gather}
\label{appA:prop1_eq3}
\bar{U}_z(x,\mu) = \begin{cases}
z_\mu & \text{if } x_\mu=L_\mu-1 \\
1 & \text{otherwise} \\
\end{cases} \ .
\end{gather}
\item The vector $z$ is unique.
\end{enumerate}
\end{proposition}

Uniqueness is proven by noticing that the vector $z$ is therefore uniquely determined by the original gauge configuration $U(x,\mu)$
\begin{gather}
z_\mu = \begin{cases}
W(\mu) & \text{if } \mu \not\in \mathcal{C} \\
W(\mu)W(3)^{-1} & \text{if } \mu \in \mathcal{C}
\end{cases} \ ,
\end{gather}
where we have introduced the Wilson lines
\begin{gather}
W(\mu) = \prod_{s=0}^{L_\mu-1} U(s\hat{\mu},\mu) \ .
\end{gather}
It is easy to show that $W(\mu)$ is gauge invariant if and only if $\mu$ is a periodic direction, while the L-shaped parallel transport $W(\mu)W(3)^{-1}$ is gauge invariant if and only if $\mu$ is a \Cstar{}-direction.

\section{Anatomy of the sign problem}
\label{app:sign}
Integration of the fermion fields in a periodic setup yields the
determinant of the Dirac operator. This result relies on the fact that
the Grassman variables $\psi(x)$ and $\bar{\psi}(x)$ are independent,
which is not true in the case of \Cstar{} boundary
conditions. A possible way to get an
explicit expression for the fermionic path integral is to use the
change of variable 
\begin{gather}
\psi_{\pm}(x) = \frac{\psi(x) \pm C^{-1} \bar{\psi}^T(x)}{\sqrt{2}} \ ,
\end{gather}
and to define the new two-component field
\begin{gather}
\eta(x) = \begin{pmatrix}
\psi_+(x) \\ -\imath \psi_-(x) 
\end{pmatrix} \ .
\end{gather}
It is straightforward to verify that \Cstar{} boundary conditions for the field $\psi(x)$ are equivalent to
\begin{gather}
\eta(x+\hat{L}_k) = K \eta(x)
\ , 
\qquad
K = \begin{pmatrix}
1 & 0 \\
0 & -1
\end{pmatrix}
\ ,
\end{gather}
where the matrix $K$ acts on the two components of $\eta$. We will refer to these boundary conditions as $K$ boundary conditions.

By using the identity for the Wilson-Dirac operator (valid for a general non-abelian gauge theory)
\begin{gather}
C^{-1} D[V]^T C = D[V^*] \ ,
\end{gather}
and a few lines of algebra, one can write the fermionic action in
terms of the new fields
\begin{gather}
S_F = \bar{\psi} D[V] \psi = - \frac{1}{2} \eta^T C D[\mathcal{J}(V)] \eta \ ,
\label{eq:sign:action}
\end{gather}
where $D_{\mathcal{J}} \equiv D[\mathcal{J}(V)]$ is the Wilson-Dirac operator calculated with the gauge field $\mathcal{J}(V)$ defined as
\begin{gather}
\mathcal{J}(V) = 1_2 \otimes \Re V + J \otimes \Im V \ ,
\qquad
J = \begin{pmatrix}
0 & -1 \\
1 & 0
\end{pmatrix}
\ .
\end{gather}
The matrices $1_2$ and $J$ act on the two components of $\eta$. Notice that $\mathcal{J}(V)$ defines a representation of the gauge group, unitarily equivalent to the representation defined by $V$. Integration of the fermionic action in the form obtained in eq.~\eqref{eq:sign:action} yields
\begin{gather}
\int_{\text{\Cstar{} b.c.s}} \mathcal{D} \bar{\psi} \mathcal{D} \psi \ e^{- \bar{\psi} D[V] \psi} = \int_{\text{K b.c.s}} \mathcal{D} \eta \ e^{\frac{1}{2} \eta^T C D_{\mathcal{J}} \eta} =
\text{Pf}_K \, C D_{\mathcal{J}} \ ,
\label{eq:sign:pfaffian}
\end{gather}
where the subscript $K$ reminds that the derivative appearing in the
Dirac operator are defined on the space of fields satisfying $K$
boundary conditions, and $C D_{\mathcal{J}}$ is an antisymmetric
complex matrix. 
In eq.~(\ref{eq:sign:pfaffian}) $\text{Pf}_K \, C D_{\mathcal{J}}$ is the Pfaffian of $C D_{\mathcal{J}}$ that, by using the algebraic identities  
\begin{gather}
( \text{Pf}_K \, C D_{\mathcal{J}} )^2 = \text{Det}_K \, C D_{\mathcal{J}} = \text{Det}_K \, D_{\mathcal{J}} \ ,
\label{eq:squarepfaff}
\end{gather}
can be related to the determinant of $D_{\mathcal{J}}$. Algorithms for the lattice simulation of theories involving Pfaffians have been discussed in the context of \Cstar{} boundary conditions 
or the closely-related $\mathrm G$-parity boundary conditions and also in the context of lattice super-symmetric models (see~\cite{Carmona:2000ds,Kelly:2013ana,Kelly:2014usa,Campos:1999du} for a list of references on this subject). 

We shall now discuss if a sign problem is associated to
$\text{Pf}_K\, C D_{\mathcal{J}}$. By using eq.~\eqref{eq:squarepfaff} and the
$\gamma_5$-hermiticity of the Dirac operator, one concludes easily
that the squared Pfaffian is real.  We want to show now that a
stronger result holds: the Pfaffian itself is real. Let us consider
the Pfaffian of the auxiliary operator $C ( D_{\mathcal{J}}-s )$ for a
generic complex number $s$. This Pfaffian is a polynomial in the
matrix elements and in particular in $s$,  
\begin{gather}
\text{Pf}_K \, C ( D_{\mathcal{J}}-s ) = \prod_\alpha (s - \lambda_\alpha)^{m_\alpha} \ ,
\end{gather}
where the $\lambda_\alpha$'s are distinct roots. The overall normalization is determined by the value of the Pfaffian in the $s \to \infty$ limit. By using the relation between the Pfaffian and the determinant we calculate the characteristic polynomial of $D_{\mathcal{J}}$
\begin{gather}
\text{Det}_K \, ( D_{\mathcal{J}}-s ) = [ \text{Pf}_K \, C ( D_{\mathcal{J}}-s ) ]^2 = \prod_\alpha (s - \lambda_\alpha)^{2 m_\alpha} \ .
\end{gather}
The $\lambda_\alpha$'s are the roots of the characteristic polynomial of $D_{\mathcal{J}}$, i.e. they are the eigenvalues of $D_{\mathcal{J}}$. Notice that the algebraic multiplicity of $\lambda_\alpha$ is $2m_\alpha$. Because of $\gamma_5$-hermiticity either the eigenvalues of $D_{\mathcal{J}}$ are real or they appear in pairs of complex conjugates. Since all multiplicities are even, the determinant is positive if $s$ is real, and consequently the Pfaffian is real. For $s=0$ one gets
\begin{gather}
\text{Pf}_K \, C D_{\mathcal{J}} = \prod_{\alpha | \Im \lambda_\alpha = 0} \lambda_\alpha^{m_\alpha} \ \prod_{\alpha | \Im \lambda_\alpha > 0} | \lambda_\alpha^{m_\alpha} |^2 \ .
\label{eq:sign:eigenvalues}
\end{gather}

Once established that the fermionic Pfaffian~\eqref{eq:sign:pfaffian} is real, we need to wonder about its sign. From eq.~\eqref{eq:sign:eigenvalues}, clearly the Pfaffian is negative only if the Dirac operator $D_{\mathcal{J}}$ has some negative eigenvalues, which can happen with Wilson fermions. However, in the continuum limit, the real part of the eigenvalues of the Dirac operator is always positive (and equal to $m$) therefore the Pfaffian is positive. At finite lattice spacing the fermionic Pfaffian~\eqref{eq:sign:pfaffian} has a mild sign problem that is completely analogous to the single-flavour case with periodic boundary conditions.

\bibliographystyle{JHEP}
\bibliography{biblio}

\providecommand{\href}[2]{#2}\begingroup\raggedright\begin{thebibliography}{10}

\bibitem{Duncan:1996xy}
A.~Duncan, E.~Eichten, and H.~Thacker, {\it {Electromagnetic splittings and
  light quark masses in lattice QCD}},  {\em Phys.Rev.Lett.} {\bf 76} (1996)
  3894--3897, [\href{http://arxiv.org/abs/hep-lat/9602005}{{\tt
  hep-lat/9602005}}].

\bibitem{Borsanyi:2014jba}
S.~Borsanyi, S.~Durr, Z.~Fodor, C.~Hoelbling, S.~Katz, et~al., {\it {Ab initio
  calculation of the neutron-proton mass difference}},  {\em Science} {\bf 347}
  (2015) 1452--1455, [\href{http://arxiv.org/abs/1406.4088}{{\tt
  arXiv:1406.4088}}].

\bibitem{deDivitiis:2013xla}
{\bf RM123} Collaboration, G.~de~Divitiis et~al., {\it {Leading isospin
  breaking effects on the lattice}},  {\em Phys.Rev.} {\bf D87} (2013), no.~11
  114505, [\href{http://arxiv.org/abs/1303.4896}{{\tt arXiv:1303.4896}}].

\bibitem{Basak:2014vca}
{\bf MILC} Collaboration, S.~Basak et~al., {\it {Finite-volume effects and the
  electromagnetic contributions to kaon and pion masses}},  {\em PoS} {\bf
  LATTICE2014} (2014) 116, [\href{http://arxiv.org/abs/1409.7139}{{\tt
  arXiv:1409.7139}}].

\bibitem{Ishikawa:2012ix}
T.~Ishikawa, T.~Blum, M.~Hayakawa, T.~Izubuchi, C.~Jung, et~al., {\it {Full
  QED+QCD low-energy constants through reweighting}},  {\em Phys.Rev.Lett.}
  {\bf 109} (2012) 072002, [\href{http://arxiv.org/abs/1202.6018}{{\tt
  arXiv:1202.6018}}].

\bibitem{Aoki:2012st}
S.~Aoki, K.~Ishikawa, N.~Ishizuka, K.~Kanaya, Y.~Kuramashi, et~al., {\it {1+1+1
  flavor QCD + QED simulation at the physical point}},  {\em Phys.Rev.} {\bf
  D86} (2012) 034507, [\href{http://arxiv.org/abs/1205.2961}{{\tt
  arXiv:1205.2961}}].

\bibitem{Blum:2010ym}
T.~Blum, R.~Zhou, T.~Doi, M.~Hayakawa, T.~Izubuchi, et~al., {\it
  {Electromagnetic mass splittings of the low lying hadrons and quark masses
  from 2+1 flavor lattice QCD+QED}},  {\em Phys.Rev.} {\bf D82} (2010) 094508,
  [\href{http://arxiv.org/abs/1006.1311}{{\tt arXiv:1006.1311}}].

\bibitem{Tantalo:2013maa}
N.~Tantalo, {\it {Isospin Breaking Effects on the Lattice}},  {\em PoS} {\bf
  LATTICE2013} (2014) 007, [\href{http://arxiv.org/abs/1311.2797}{{\tt
  arXiv:1311.2797}}].

\bibitem{Portelli:2015wna}
A.~Portelli, {\it {Inclusion of isospin breaking effects in lattice
  simulations}},  {\em PoS} {\bf LATTICE2014} (2015) 013,
  [\href{http://arxiv.org/abs/1505.07057}{{\tt arXiv:1505.07057}}].

\bibitem{Carrasco:2015xwa}
N.~Carrasco, V.~Lubicz, G.~Martinelli, C.~Sachrajda, N.~Tantalo, et~al., {\it
  {QED Corrections to Hadronic Processes in Lattice QCD}},  {\em Phys.Rev.}
  {\bf D91} (2015), no.~7 074506, [\href{http://arxiv.org/abs/1502.00257}{{\tt
  arXiv:1502.00257}}].

\bibitem{Endres:2015gda}
M.~G. Endres, A.~Shindler, B.~C. Tiburzi, and A.~Walker-Loud, {\it {Massive
  photons: an infrared regularization scheme for lattice QCD+QED}},
  \href{http://arxiv.org/abs/1507.08916}{{\tt arXiv:1507.08916}}.

\bibitem{Lehner:2015bga}
C.~Lehner and T.~Izubuchi, {\it {Towards the large volume limit - A method for
  lattice QCD + QED simulations}},  {\em PoS} {\bf LATTICE2014} (2015) 164,
  [\href{http://arxiv.org/abs/1503.04395}{{\tt arXiv:1503.04395}}].

\bibitem{lehnerlattice2015}
C.~Lehner, T.~Izubuchi, and L.~Jin, ``{Improving the volume-dependence of
  lattice QCD+QED simulations}.'' 33rd International Symposium on Lattice Field
  Theory, 2015.

\bibitem{Hayakawa:2008an}
M.~Hayakawa and S.~Uno, {\it {QED in finite volume and finite size scaling
  effect on electromagnetic properties of hadrons}},  {\em Prog.Theor.Phys.}
  {\bf 120} (2008) 413--441, [\href{http://arxiv.org/abs/0804.2044}{{\tt
  arXiv:0804.2044}}].

\bibitem{Davoudi:2014qua}
Z.~Davoudi and M.~J. Savage, {\it {Finite-Volume Electromagnetic Corrections to
  the Masses of Mesons, Baryons and Nuclei}},  {\em Phys.Rev.} {\bf D90}
  (2014), no.~5 054503, [\href{http://arxiv.org/abs/1402.6741}{{\tt
  arXiv:1402.6741}}].

\bibitem{Fodor:2015pna}
Z.~Fodor, C.~Hoelbling, S.~Katz, L.~Lellouch, A.~Portelli, et~al., {\it
  {Quantum electrodynamics in finite volume and nonrelativistic effective field
  theories}},  \href{http://arxiv.org/abs/1502.06921}{{\tt arXiv:1502.06921}}.

\bibitem{Polley:1993bn}
L.~Polley, {\it {Boundaries for SU(3)(C) x U(1)-el lattice gauge theory with a
  chemical potential}},  {\em Z. Phys.} {\bf C59} (1993) 105--108.

\bibitem{Wiese:1991ku}
U.~Wiese, {\it {C periodic and G periodic QCD at finite temperature}},  {\em
  Nucl.Phys.} {\bf B375} (1992) 45--66.

\bibitem{Kronfeld:1990qu}
A.~S. Kronfeld and U.~Wiese, {\it {SU(N) gauge theories with C periodic
  boundary conditions. 1. Topological structure}},  {\em Nucl.Phys.} {\bf B357}
  (1991) 521--533.

\bibitem{Kronfeld:1992ae}
A.~S. Kronfeld and U.~Wiese, {\it {SU(N) gauge theories with C periodic
  boundary conditions. 2. Small volume dynamics}},  {\em Nucl.Phys.} {\bf B401}
  (1993) 190--205, [\href{http://arxiv.org/abs/hep-lat/9210008}{{\tt
  hep-lat/9210008}}].

\bibitem{Lee:2015rua}
J.-W. Lee and B.~C. Tiburzi, {\it {On Finite Volume Corrections to the
  Electromagnetic Mass of Composite Particles}},
  \href{http://arxiv.org/abs/1508.04165}{{\tt arXiv:1508.04165}}.

\bibitem{Luscher:1985dn}
M.~Luscher, {\it {Volume Dependence of the Energy Spectrum in Massive Quantum
  Field Theories. 1. Stable Particle States}},  {\em Commun.Math.Phys.} {\bf
  104} (1986) 177.

\bibitem{Dirac:1955uv}
P.~A. Dirac, {\it {Gauge invariant formulation of quantum electrodynamics}},
  {\em Can.J.Phys.} {\bf 33} (1955) 650.

\bibitem{HaagBook}
R.~Haag, {\em Local Quantum Physics: Fields, Particles, Algebras}.
\newblock Springer, 2nd rev. and enlarged~ed., August, 1996.

\bibitem{StrocchiBook}
F.~Strocchi, {\em An Introduction to the Non-Perturbative Foundations of
  Quantum Field Theory}.
\newblock Oxford University Press, 1st~ed., March, 2013.

\bibitem{Polley:1990tf}
L.~Polley and U.~Wiese, {\it {Monopole condensate and monopole mass in U(1)
  lattice gauge theory}},  {\em Nucl.Phys.} {\bf B356} (1991) 629--654.

\bibitem{Cottingham:1963zz}
W.~Cottingham, {\it {The neutron proton mass difference and electron scattering
  experiments}},  {\em Annals Phys.} {\bf 25} (1963) 424--432.

\bibitem{Collins:1978hi}
J.~C. Collins, {\it {Renormalization of the Cottingham Formula}},  {\em
  Nucl.Phys.} {\bf B149} (1979) 90.

\bibitem{ProgressBook}
M.~{L\"uscher}, {\it On a relation between finite size effects and elastic
  scattering processes},  in {\em Progress in Gauge Field Theory} (G.~'t~Hooft,
  A.~Jaffe, G.~Lehmann, P.~Mitter, and I.~Singer, eds.).
\newblock Springer US, 1st~ed., 1984.

\bibitem{Low:1954kd}
F.~Low, {\it {Scattering of light of very low frequency by systems of spin
  1/2}},  {\em Phys.Rev.} {\bf 96} (1954) 1428--1432.

\bibitem{GellMann:1954kc}
M.~Gell-Mann and M.~Goldberger, {\it {Scattering of low-energy photons by
  particles of spin 1/2}},  {\em Phys.Rev.} {\bf 96} (1954) 1433--1438.

\bibitem{WeinbergBook}
S.~Weinberg, {\em The Quantum Theory of Fields}, vol.~1.
\newblock Cambridge University Press, 1995.

\bibitem{BogoBook}
N.~N. Bogolyubov, A.~A. Logunov, A.~I. Oksak, and I.~T. Todorov, {\em General
  principles of quantum field theory}.
\newblock Kluwer Academic Publishers, 1st~ed., 1990.

\bibitem{JLpaper}
R.~Jost and H.~Lehmann, {\it Integral-darstellung kausaler kommutatoren},  {\em
  Il Nuovo Cimento} {\bf 5} (1957) 1598--1610.

\bibitem{Dyson:1997gw}
F.~Dyson, {\it {Integral representations of causal commutators}},  {\em
  Phys.Rev.} {\bf 110} (1958) 1460--1464.

\bibitem{Carmona:2000ds}
J.~M. Carmona, M.~D'Elia, A.~Di~Giacomo, and B.~Lucini, {\it {Implementation of
  C* boundary conditions in the hybrid Monte Carlo algorithm}},  {\em
  Int.J.Mod.Phys.} {\bf C11} (2000) 637--654,
  [\href{http://arxiv.org/abs/hep-lat/0003002}{{\tt hep-lat/0003002}}].

\bibitem{Kelly:2013ana}
C.~Kelly, {\it {Progress Towards an ab initio, Standard Model Calculation of
  Direct CP-Violation in K-decays}},
  \href{http://arxiv.org/abs/1310.0434}{{\tt arXiv:1310.0434}}.

\bibitem{Kelly:2014usa}
{\bf RBC, UKQCD} Collaboration, C.~Kelly, T.~Blum, N.~Christ, A.~Lytle, and
  C.~Sachrajda, {\it {Progress Towards an ab initio, Standard Model Calculation
  of Direct CP-Violation in K-decays}},  {\em PoS} {\bf LATTICE2013} (2014)
  401.

\bibitem{Campos:1999du}
{\bf DESY-Munster} Collaboration, I.~Campos et~al., {\it {Monte Carlo
  simulation of SU(2) Yang-Mills theory with light gluinos}},  {\em
  Eur.Phys.J.} {\bf C11} (1999) 507--527,
  [\href{http://arxiv.org/abs/hep-lat/9903014}{{\tt hep-lat/9903014}}].

\end{thebibliography}\endgroup

\end{document}